\documentclass[aps,prl,twocolumn,floatfix,superscriptaddress,footinbib]{revtex4-2}
\usepackage{graphicx}
\usepackage{amsmath,amssymb}
\usepackage{physics}
\usepackage[dvipsnames]{xcolor}
\usepackage{hyperref}
\hypersetup{colorlinks,linkcolor=blue,citecolor=blue,urlcolor=blue}
\usepackage{float}
\def\Bc{{\beta_{\mathrm{C}}}}
\def\Bq{{\beta_{\mathrm{Q}}}}
\newcommand{\BqGS}{{\beta_{\mathrm{Q}}^{\min}}}
\newcommand{\LDSset}{\mathcal{S}_{\mathrm{LDS}}}

\usepackage{bm}
\usepackage{amsthm}
\usepackage{cleveref}
\newtheorem{theorem}{Theorem}
\newtheorem{cor}{Corollary}
\crefname{cor}{Corollary}{Corollaries}
\newcommand{\avgD}[1]{\langle #1 \rangle_{D_N^{\,N/2}}}
\newcommand{\avgDs}[1]{\langle #1 \rangle_{D_6^{\,3}}}
\newcommand{\pexp}{p_{\mathrm{exp}}}
\begin{document}

\title{Experimental certification of multipartite Bell correlations using only few-body symmetric correlations}

\newcommand{\AffCQT}{Centre for Quantum Technologies, National University of Singapore, 3 Science Drive 2, Singapore 117543, Singapore}
\newcommand{\AffDoP}{Department of Physics, National University of Singapore,  2 Science Drive 3, Singapore 117542, Singapore}
\newcommand{\AffLeiden}{Instituut-Lorentz, Universiteit Leiden, P.O. Box 9506, 2300 RA Leiden, The Netherlands}
\newcommand{\AffAqa}{\(\langle \text{aQa}^{\text{L}} \rangle\) Applied Quantum Algorithms Leiden, The Netherlands}
\newcommand{\AffLPS}{Advanced Institute of Precision Spectroscopy, East China Normal University, Shanghai 200241, China}
\newcommand{\AffHNL}{Shanghai Branch, Hefei National Laboratory, Shanghai 201315, China}
\newcommand{\AffCQ}{Chongqing Key Laboratory of Precision Optics, Chongqing Institute of East China Normal University, Chongqing 401121, China}
\newcommand{\AffSX}{\mbox{Collaborative Innovation Center of Extreme Optics,\! Shanxi University,\! Taiyuan,\! Shanxi~030006,\! China}}
\newcommand{\AffQSG}{Quantum Science Center of Guangdong–Hong Kong–Macao Greater Bay Area (Guangdong), Shenzhen 518045,  China}
\newcommand{\AffIQCS}{Institute of Quantum Computing and Software, School of Computer Science and Engineering, Sun Yat-sen University, Guangzhou 510006, China}

\author{Jiacheng Sun}
\thanks{These authors contributed equally to this work.}
\affiliation{\AffLPS}

\author{Mengyao Hu}
\thanks{These authors contributed equally to this work.}
\affiliation{\AffCQT}

\author{Pawe\l{} Cie\'{s}li\'{n}ski}
\affiliation{\AffCQT}

\author{Jinfu Chen}
\affiliation{\AffAqa}
\affiliation{\AffLeiden}

\author{Fei Shi}
\affiliation{\AffIQCS}

\author{Bowen Wang}
\affiliation{\AffLPS}

\author{Hongyu Yang}
\affiliation{\AffLPS}

\author{Jizhou Wu}
\affiliation{\AffQSG}

\author{Jordi Tura}
\email{tura@lorentz.leidenuniv.nl}
\affiliation{\AffAqa}
\affiliation{\AffLeiden}

\author{Valerio Scarani}
\email{physv@nus.edu.sg}
\affiliation{\AffCQT}
\affiliation{\AffDoP}

\author{Dian Wu}
\email{dwu@lps.ecnu.edu.cn}
\affiliation{\AffLPS}
\affiliation{\AffHNL}

\author{Jian Wu}
\email{jwu@phy.ecnu.edu.cn}
\affiliation{\AffLPS}
\affiliation{\AffCQ}
\affiliation{\AffSX}

\begin{abstract}

Certifying Bell correlations in increasingly larger multipartite quantum systems is challenging because conventional Bell inequalities rely on many-body correlations, which become experimentally demanding to estimate as system size grows. Here we demonstrate that multipartite Bell correlations can be certified using few-body permutation-invariant (PI) correlations with a fixed maximal correlation order. We experimentally test few-body PI Bell correlation witnesses tailored for six- and eight-photon postselected Dicke states. The addition of a few four-body correlators is found to significantly enhance the noise tolerance compared with previously studied two-body witnesses: for instance, for six qubits, the selected witness tolerates a white-noise fraction of $p_c\approx0.301$, compared with $p_c\approx0.048$ for the best two-body PI inequality with the same two settings per party. For six and eight photons respectively, we observe violations of the witness bounds by $32$ and $8.9$ standard deviations, arising from postselected polarisation states with fidelities of $0.933(4)$ and $0.916(11)$ to the corresponding half-excitation Dicke states. The witnesses for the six- and eight-photon cases use the same symmetric two- and four-body correlators, just with different coefficients. These results demonstrate an experimentally accessible approach for certifying multipartite Bell correlations in larger systems without requiring full-body correlation measurements.
\end{abstract}

\maketitle

\textit{Introduction.}---
Bell inequalities are fundamental tools for detecting quantum nonlocality, the fact that correlations between spatially separated systems cannot, in general, be reproduced by any local hidden-variable (LHV) model~\cite{Brunner_2014, bell_einstein_1964, clauser_proposed_1969}. Loophole-free Bell tests in bipartite systems~\cite{Hensen2015,Giustina2015,Shalm2015,Rosenfeld2017,Rauch2018} have established nonlocality as a feature of nature and as a resource for device-independent quantum information~\cite{Ekert1991,Acin_2007,Pironio_2009,Pironio_2010,Colbeck_2011,Mayers_2003,Pawlowski_review_2025}.
Extending Bell scenarios to multipartite systems, however, remains challenging. From a theoretical perspective, determining the extremal predictions of LHV theories generally becomes exponentially hard with the number of parties. Although important families of multipartite inequalities, such as the Mermin-Ardehali-Belinskii-Klyshko (MABK) \cite{Mermin_1990, Ardehali_1992, Belinskii1993} inequalities, were constructed, they rely on full-body correlators, whose measurement becomes increasingly demanding with growing system size. It is therefore important to identify multipartite Bell signatures whose required correlation order remains experimentally accessible as the number of parties increases.

In multipartite systems, Bell inequalities can be employed as \emph{Bell correlation witnesses}~\cite{schmied_bell_2016,engelsen_bell_2017,cieslinski_no-cost_2025}. Their LHV bounds serve as a threshold: correlations that violate it cannot be reproduced by any LHV model and are therefore Bell correlations, which imply entanglement, whereas the converse does not hold~\cite{acin_violation_2003,Brunner_2014}.
Such witnesses have revealed Bell correlations in atomic ensembles~\cite{schmied_bell_2016,engelsen_bell_2017} and a superconducting processor~\cite{wang_probing_2025}, and they are particularly economical with few-body correlators.
Which few-body inequality to use depends on the target state. In this work, it is chosen as the half-excitation Dicke state $\ket{D_N^{\,N/2}}$, a long-standing benchmark for multiphoton experiments~\cite{kiesel_experimental_2007,prevedel_experimental_2009,wieczorek_experimental_2009,krischek_useful_2011,schwemmer_experimental_2014}. Symmetry of the Dicke state makes permutation-invariant (PI) Bell inequalities~\cite{tura_detecting_2014,tura_nonlocality_2015,cieslinski_unmasking_2023, Munn2026} with few-body correlators a natural choice. Particle exchange symmetry reduces the determination of the LHV bound to polynomially many strategies~\cite{tura_detecting_2014,tura_nonlocality_2015,fadel_bounding_2017,aloy_deriving_2024}, and few-body correlators avoid the need to access increasingly high-order joint correlations as the number of parties grows~\cite{guo_detecting_2023,hu_tropical_2026,chen_optimizing_2025,chen_optimizing_2026}.

Here, we experimentally demonstrate PI Bell correlation witnesses containing only two- and four-body correlators, enabling, to our knowledge, the first certification of Bell correlations in an eight-photon system within this framework. Our setup generates high-fidelity polarization-entangled Dicke states with half excitations for $N=6$ and $N=8$ photons in the case where each photon is found in a different output station.  Owing to no-signaling and to the symmetrization of the correlators in the inequality, one $N$-fold coincidence event (a single sample for the estimation of an $N$-body correlator) provides up to $\binom{N}{2}$ samples of a two-body PI correlator and $\binom{N}{4}$ samples of a four-body one.

We identify experimentally accessible Bell correlation witnesses with two measurement settings per party and optimize their parameters for enhanced robustness against white-noise admixture. Since two-body PI inequalities with two settings per party tolerate at most $5\%$ white noise for $N=6$ and $8$~\cite{tura_nonlocality_2015}, our witnesses use also four-body correlators. The obtained data certify Bell correlations in both systems, in agreement with the theoretical predictions, and show that the violations survive significant deviations of the measurement directions from their optimal values. Our results demonstrate that PI few-body Bell correlation witnesses provide a practical framework for certifying multipartite Bell correlations in multiphoton entangled states.

\textit{PI Bell inequalities as correlation witnesses for Dicke states.}---
Consider a Bell scenario with $N$ parties, where each of them is performing $m$ measurements $\{A_{x}^{(i)}\}_{x=0}^{m-1}$ with two $\pm1$ outcomes. Using the standard notation, we denote it as $(N,m,2)$. A general PI Bell inequality involving
correlators up to the $r$th order takes the form
\begin{equation}
I(\boldsymbol{\alpha})
= \sum_{r=1}^k \frac{1}{r!}\sum_{x_1,\dots,x_r=0}^{m-1}
\alpha_{x_1\cdots x_r}\,\mathcal{S}_{x_1\cdots x_r}
\;\geq\; \Bc(\boldsymbol{\alpha}),
\label{eq:PIansatz}
\end{equation}
where
$\mathcal{S}_{x_1\cdots x_r}=\sum_{(i_1,\cdots,i_r)}
\langle A_{x_1}^{(i_1)}\cdots A_{x_r}^{(i_r)}\rangle$
sums over all $r$-tuples of pairwise different parties, the coefficients
$\alpha_{x_1\cdots x_r}$ are symmetric under index exchange, and
$\Bc(\boldsymbol{\alpha})$
corresponds to the classical bound. Here $\Bc(\boldsymbol{\alpha})$ is evaluated over the vertices of the PI local
polytope $\mathcal{L}_\Pi$ being the projection of the full local polytope $\mathcal{L}$
onto the symmetrized correlators~\cite{tura_detecting_2014,tura_nonlocality_2015,fadel_bounding_2017,aloy_device-independent_2019,guo_detecting_2023,aloy_deriving_2024}. The vertices of the full $\mathcal{L}$ are the local deterministic strategies, whose number scales as $2^{mN}$. By contrast, the number of vertices of  $\mathcal{L}_\Pi$ grows only polynomially in $N$, thus allowing one to probe greater system sizes.

Throughout this work, Eq.~\eqref{eq:PIansatz} is used as a \emph{Bell correlation witness} rather than as a Bell test. For a state $\rho$ measured with fixed local observables, $I=\mathrm{Tr}[\rho\,\hat{I}(\boldsymbol{\theta})]$, where the Bell operator $\hat{I}(\boldsymbol{\theta})$ is obtained from Eq.~\eqref{eq:PIansatz} by replacing every correlator with the corresponding product of observables. If $I<\Bc$, the correlations of $\rho$ cannot be reproduced by any LHV model: $\rho$ is Bell correlated~\cite{schmied_bell_2016} and, in particular, entangled~\cite{acin_violation_2003}. 
Since the data are postselected on $N$-fold coincidences, the violation characterizes the postselected ensemble and does not constitute a loophole-free or device-independent Bell test~\cite{Larsson2014}.

To design the witness, we need a target state and a preferred model of noise (the Dicke states and the white noise) needed to estimate the robustness. Moreover, we use identical measurement settings for every party, for the sake of a PI witness. Importantly, these are guidelines for the design, not assumptions of the certification. They are used to choose the coefficients $\boldsymbol{\alpha}$ and the angles $\boldsymbol{\theta}$ that are expected to produce a robust violation. 
Once these parameters have been fixed, the LHV bound $\Bc$ remains valid for arbitrary states and arbitrary local measurements, and the witness is evaluated directly from the measured correlators. 
The certificate rests on two assumptions. (i) Because the correlators are assembled from $\binom{N}{2}$ measurement configurations recorded one after another, the same state must be prepared, and each setting must correspond to the same local measurement, in all of them (see Fig.~S4 of the Supplemental Material for the stability of the source). (ii) The postselected $N$-fold coincidences must be a fair sample of the emitted states. 
Under these assumptions, deviations of the actual state or measurements from the design model can reduce the observed violation but cannot lead to a false violation.

In this work, we target the half-excitation Dicke states
\begin{equation*}
    \ket{D_N^{\,N/2}} = \binom{N}{N/2}^{-1/2} \sum_i\mathcal{P}_i (|\underbrace{1\cdots 1}_{N/2} 0\cdots 0 \rangle ),
\end{equation*}
for $N=6$ and $N=8$ qubits, where the sum runs over the $\binom{N}{N/2}$ distinct qubit permutations
$\mathcal{P}_i$. Every party $i$ measures the same observables
$A_x^{(i)}=\cos\theta_x\,\sigma_z^{(i)}+\sin\theta_x\,\sigma_x^{(i)}$, $x=0,\dots,m-1$, with angles
$\boldsymbol{\theta}=(\theta_0,\dots,\theta_{m-1})$ shared by all parties.
Then the quantum value on the target state is
$\Bq(\boldsymbol{\alpha})=\min_{\boldsymbol{\theta}}
\bra{D_N^{\,N/2}}\hat{I}(\boldsymbol{\theta})\ket{D_N^{\,N/2}}$.
Since the state is symmetric, $\hat{I}(\boldsymbol{\theta})$ can be evaluated in the $(N{+}1)$-dimensional symmetric
subspace~\cite{tura_nonlocality_2015,hassani_many-body_2025}.
Moreover, $\ket{D_N^{\,N/2}}$ is a common eigenstate of
$\Pi_z=\bigotimes_i\sigma_z^{(i)}$ and $\Pi_x=\bigotimes_i\sigma_x^{(i)}$, which together force every odd-body correlator built from measurements in the $xz$ plane to
vanish on the state, so only even-body terms in
\eqref{eq:PIansatz} contribute to the possible violation (see Supplemental Material, Sec.~SI\,C).

\textit{Noise-robust witnesses with two- and four-body correlators.}---To quantify the noise robustness of a witness, we use the depolarized state 
$\rho(p)=(1-p)\ketbra{D_N^{\,N/2}}{D_N^{\,N/2}}+p\,\openone/2^N$, for which the critical noise threshold
\begin{equation}
p_c(\boldsymbol{\alpha})
= 1 - \frac{\Bc(\boldsymbol{\alpha})}{\Bq(\boldsymbol{\alpha})}
\label{eq:pcrit}
\end{equation}
sets the certification condition $p_\mathrm{exp}<p_c$, where
$p_\mathrm{exp}=(1-F)/(1-2^{-N})$ is the white-noise fraction
equivalent to an experimental fidelity $F$.
Equivalently, the witness certifies the state whenever $F$ exceeds the critical fidelity $F_c=1-p_c(1-2^{-N})$.
The states prepared in this work, with $F=0.933(4)$ ($N=6$) and $F=0.916(11)$ ($N=8$), correspond to $p_\mathrm{exp}\approx0.068$ and $0.084$, respectively.

Among two-body PI inequalities with two settings per party, the largest threshold for $\ket{D_6^{\,3}}$ is attained by $\mathcal{I}_{\mathrm{T}}= 15\,\mathcal{S}_{00} + 6\,\mathcal{S}_{01} - \mathcal{S}_{11}\geq -120$~\cite{tura_detecting_2014,tura_nonlocality_2015}, which tolerates only $p_c\approx 0.048$ ($F_c\approx0.95$). 
No two-body PI inequality with two settings certifies the prepared state, and a third setting raises the threshold only to $p_c\approx 0.093$ (see Supplemental Material below). 
We therefore extend the PI ansatz to four-body correlators while keeping two settings per party. By the parity selection rule, the resulting family in the $(6,2,2)$ scenario has eight independent coefficients (End Matter), and a see-saw optimization of $p_c$ over them in the symmetric subspace reaches $p_c\approx 0.46$ (see Supplemental Material below). The optimum, however, places large weights on all eight correlators, which would amplify statistical uncertainties in a finite-count experiment. We therefore select sparse members involving only $\mathcal{S}_{00}$, $\mathcal{S}_{11}$, $\mathcal{S}_{0011}$, and $\mathcal{S}_{1111}$, all of which are obtained from the $\binom{N}{2}$ configurations with two photons in setting $0$ and $N-2$ in setting $1$. The witness for $N=6$ reads
\begin{equation}
\mathcal{I}_1
= \mathcal{S}_{00}
+ \mathcal{S}_{11}
- \frac{1}{4}\,\mathcal{S}_{0011}
+ \frac{1}{6}\,\mathcal{S}_{1111}
\;\geq\; -18,
\label{eq:BI_24body_sparse}
\end{equation}
with $\Bq\approx-25.764$ at $\boldsymbol{\theta}\approx(3.550,\,5.730)$, $p_c\approx0.301$, and $F_c\approx0.70$, far above $p_\mathrm{exp}\approx0.068$. A second member, $\mathcal{I}_2=15\,\mathcal{S}_{00} - 6\,\mathcal{S}_{11} - \tfrac{5}{4}\,\mathcal{S}_{0011} + \tfrac{1}{8}\,\mathcal{S}_{1111} \geq -135$, has the same correlator content but $p_c\approx0.145$ ($F_c\approx0.86$). The two provide certificates of contrasting robustness, which we probe with states of deliberately reduced fidelity. Both are facets of PI local polytopes (see Supplemental Material below). 
Their quantum values, optimal angles, thresholds, and measured values are listed in Table~\ref{tab:witnesses} (End Matter).

The construction extends to larger $N$: for $N=8$, the family optimum reaches $p_c\approx0.38$, and the sparse member $\mathcal{I}_3=\tfrac{3}{2}\,\mathcal{S}_{00} + \tfrac{9}{2}\,\mathcal{S}_{11} - \tfrac{1}{4}\,\mathcal{S}_{0011} + \tfrac{1}{4}\,\mathcal{S}_{1111} \geq -48$ (Table~\ref{tab:witnesses}), built from the same four correlators, has $p_c\approx0.226$ ($F_c\approx0.78$), far above $p_\mathrm{exp}\approx0.084$. Up to $N=10$, the four-body family retains $p_c\approx0.37$ where the two-body families drop below $0.1$ (End Matter, Fig.~\ref{fig:p_crit_vs_n}).

\begin{figure}
\centering
\includegraphics[width=0.43\textwidth]{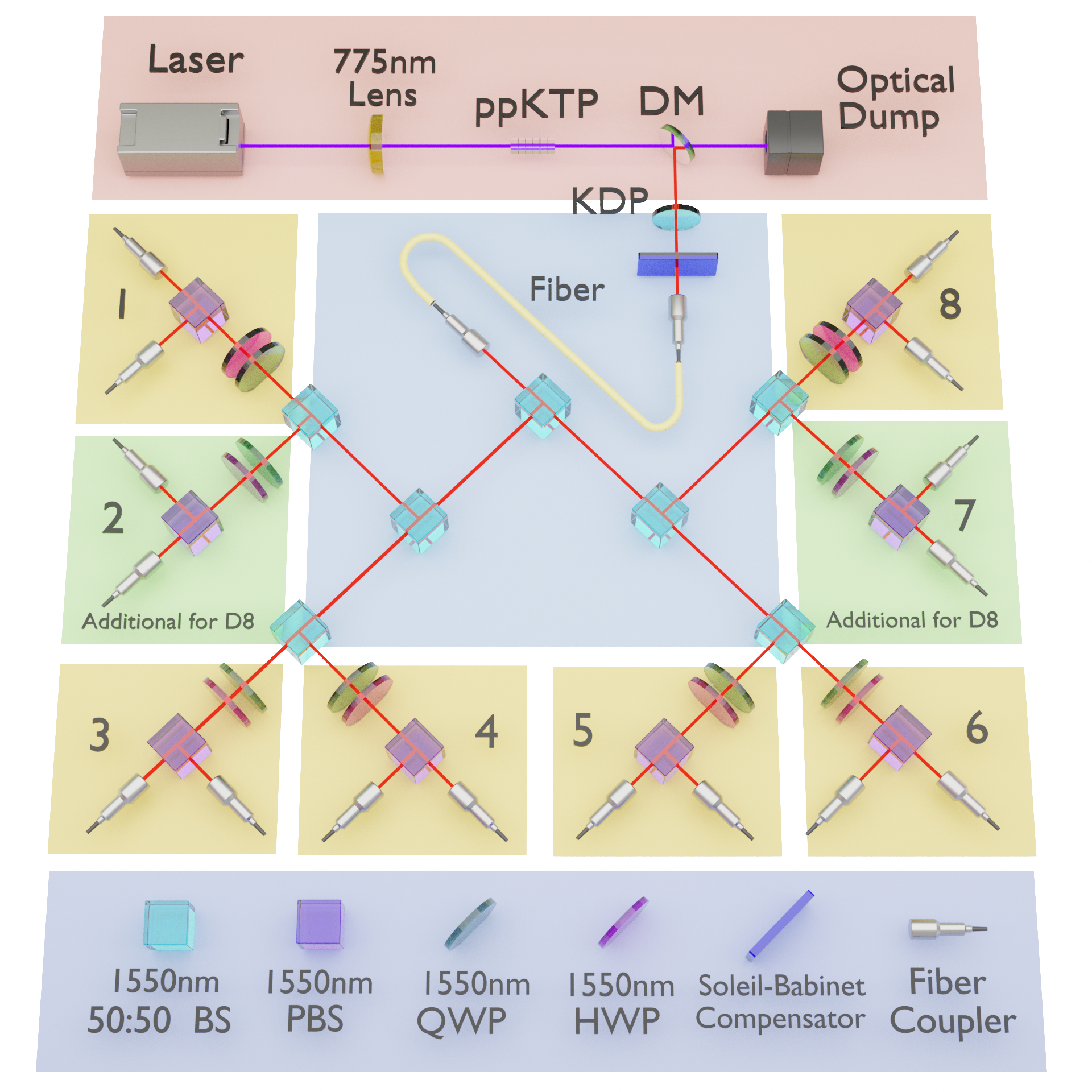}
\caption{Experimental setup. Pulses from a mode-locked Ti:sapphire laser (center wavelength 775 nm, repetition rate 80 MHz) are focused onto a periodically poled $\mathrm{KTiOPO_{4}}$ (ppKTP) crystal using a lens, producing photon pairs via type-II spontaneous parametric down-conversion (SPDC). Residual pump light is removed with a dichroic mirror (DM). The temporal walk-off is compensated by a potassium dihydrogen phosphate (KDP) crystal and accurately adjusted by a pair of Soleil-Babinet compensators. Each output mode is analyzed by a standard QWP-HWP-PBS module implementing the required local projective settings. Modes 2 and 7, including two additional BSs, are used in $\ket{D_8^{\,4}}$ preparation. Photons are coupled into single-mode fibers (SMFs) and detected with superconducting nanowire single-photon detectors (SNSPDs, $\approx90\%$ efficiency). Sixfold coincidences identify $\ket{D_6^{\,3}}$ and eightfold coincidences identify $\ket{D_8^{\,4}}$, which are recorded with a multichannel coincidence counter.}
\label{p1}
\end{figure}

\begin{figure*}
\centering
\includegraphics[width=0.9\textwidth]{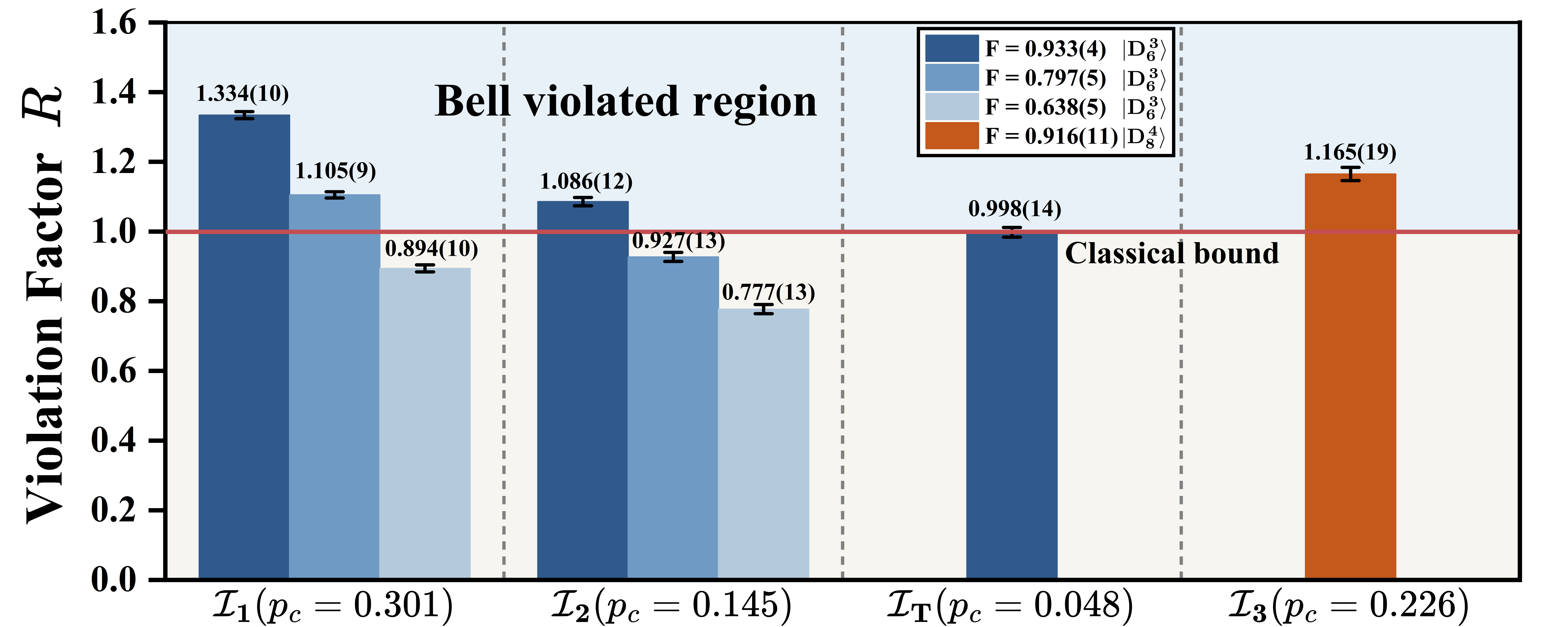}
\caption{
Certification of Bell correlations with few-body witnesses at different state fidelities.
Bars show the violation factor $R=I/\Bc$ for $\mathcal{I}_{1}$, $\mathcal{I}_{2}$, and $\mathcal{I}_{\mathrm{T}}$ evaluated on $\ket{D_6^{\,3}}$ states prepared with fidelities $F = 0.933(4)$, $0.797(5)$, and $0.638(5)$, as well as for $\mathcal{I}_{3}$ evaluated on $\ket{D_8^{\,4}}$ with $F=0.916(11)$.
All violation factors are normalized to the classical bound(red horizontal line), with the shaded regions indicating Bell violation. The predicted critical white-noise fractions $p_{c}$ are indicated below the corresponding inequalities where applicable.  
Because several expectation values are extracted from the same multiphoton coincidence data, their statistical correlations are retained using a Poisson parametric bootstrap (Supplemental Material).
Error bars denote one standard deviation obtained from Poisson parametric-bootstrap realizations of the raw coincidence counts.
}
\label{p2}
\end{figure*}

\textit{Experimental realization.}---
The optical setup used for the Dicke states preparation is shown in FIG.~\ref{p1}. The resulting state is entangled in the polarization basis $\{|H\rangle, |V\rangle \}$, corresponding to $\{|0\rangle, |1\rangle \}$ in the computational basis. To prepare the six-photon Dicke state $\ket{D_6^{\,3}}$, we generate multiple pairs of $\ket{HV}$ photons via a higher-order type-II SPDC process. Photons are then distributed into six spatial modes by five $50:50$ beam splitters (BSs). Detection of one photon in each mode identifies the postselected six-photon Dicke state~\cite{wieczorek_experimental_2009}. 

Using an exact projector decomposition protocol, we characterize the fidelity of $|D_6^{\,3}\rangle$ with $21$ polarization settings~\cite{toth_practical_2009}. By adjusting the walk-off compensation and the temperature of the ppKTP crystal, we prepare $\ket{D_6^{\,3}}$ states with different fidelities of $0.933(4)$, $0.797(5)$, and $0.638(5)$, with a sixfold coincidence count rate of $6.8$ Hz. Preparation of the $N$-photon resource relies on $N$-fold coincidence postselection, and the witnesses are evaluated from the same sixfold or eightfold records using only two- and four-body polarization correlators. The advantage demonstrated here is therefore a reduction of the required correlation order, not of the photon-number detection requirements.
To evaluate the witnesses in the $(6,2,2)$ scenario, we separately measure the $15$ permutation-related measurement configurations required for $\mathcal{I}_1$ and $\mathcal{I}_2$ (Table~\ref{tab:witnesses}).
The experimental witness values $I$ are constructed from the two-body ($m_{00}$ and $m_{11}$) and four-body ($m_{0011}$ and $m_{1111}$) marginal correlators extracted from sixfold coincidence events (Supplemental Material).
To compare witnesses with different numerical scales, we use the violation factor $R=I/\Bc$. Since $I$ and $\Bc$ are negative, $R>1$ signals a violation.
Fig.~\ref{p2} compares the measured violation factors at the three fidelities. Note that the underlying PI correlators are compared with the ideal-state predictions in Fig.~\ref{fig:correlators} (End Matter).

We experimentally compare the certification performance of $\mathcal{I}_{1}$ and $\mathcal{I}_{2}$ (Table~\ref{tab:witnesses}), with the reference two-body PI inequality $\mathcal{I}_{\mathrm{T}}$~\cite{tura_detecting_2014, tura_nonlocality_2015}. As shown in Fig.~\ref{p2}, for $|D_6^{\,3}\rangle$ with $F=0.933(4)$, $\mathcal{I}_{1}$ yields a violation factor of $R=1.334(10)$, whereas $\mathcal{I}_{\mathrm{T}}$ gives $R=0.998(14)$ and remains statistically compatible with its classical bound.
At the lower fidelity of $F=0.797(5)$, $\mathcal{I}_{1}$ still detects Bell correlations, with $R=1.105(9)$, while the second four-body witness $\mathcal{I}_{2}$ gives $R=0.927(13)$ and no longer violates its bound. Neither $\mathcal{I}_{1}$ nor $\mathcal{I}_{2}$ detects a violation at $F=0.638(5)$, in line with the predicted hierarchy $p_c(\mathcal{I}_1)>p_c(\mathcal{I}_2)>p_c(\mathcal{I}_{\mathrm T})$.

We further demonstrate the robustness of the PI witness $\mathcal{I}_{1}$ against collective angular misalignments of the local measurement bases in the $xz$ plane.
Theoretically, we map the theoretical violation landscape for an ideal six-photon Dicke state around the optimal angles $(\theta_{0},\theta_{1})=(3.550,5.730)$ and experimentally probe six displaced points ($P_{1}$--$P_{6}$) using $\ket{D_6^{\,3}}$ with fidelity $F=0.933(4)$. 
At each point, the same pair of perturbed measurement angles is applied to all six photons.
The positions of these points and the corresponding measured violation factors are shown in Fig.~\ref{fig:tolerance}.
All six displaced settings retain Bell violations of $6.8$--$14$ standard deviations, demonstrating that the PI witness preserves its certification capability under measurement angular errors. 
Because $\Bc$ holds for arbitrary local measurements, these violations are valid certificates irrespective of the calibration of the settings, under assumptions (i) and (ii) above. The design model only determines how large the violation is.

To test the witnesses at a larger system size, we extend the experiment to the eight-photon Dicke state $\ket{D_8^{\,4}}$ by adding two 50:50 BSs to the setup in Fig.~\ref{p1}. The state characterization and witness evaluation are conditioned on eightfold coincidences, recorded at a rate of $0.031$~Hz. Using an exact decomposition of the target-state projector with $29$ polarization settings (Supplemental Material), we measure a fidelity of $F=0.916(11)$. 
Crucially, this extension preserves the sparse measurement structure of the witnesses. Despite two available local settings per photon, only $\binom{N}{2}$ permutation-related joint configurations are required, instead of $2^N$ combinations in the complete two-setting measurement space.
From the $28$ permutation-related configurations of $\mathcal{I}_3$ (Table~\ref{tab:witnesses}) we obtain $\mathcal{I}_3=-55.9(9)$, i.e., $R=1.165(19)$, a violation of the classical bound $\Bc=-48$ by $8.9$ standard deviations. 
This result demonstrates that the increase from six to eight photons does not require a higher experimentally accessed correlation order: the same four PI correlators remain sufficient to certify Bell correlations in both systems.

\begin{figure*}
\centering
\includegraphics[width=2\columnwidth]{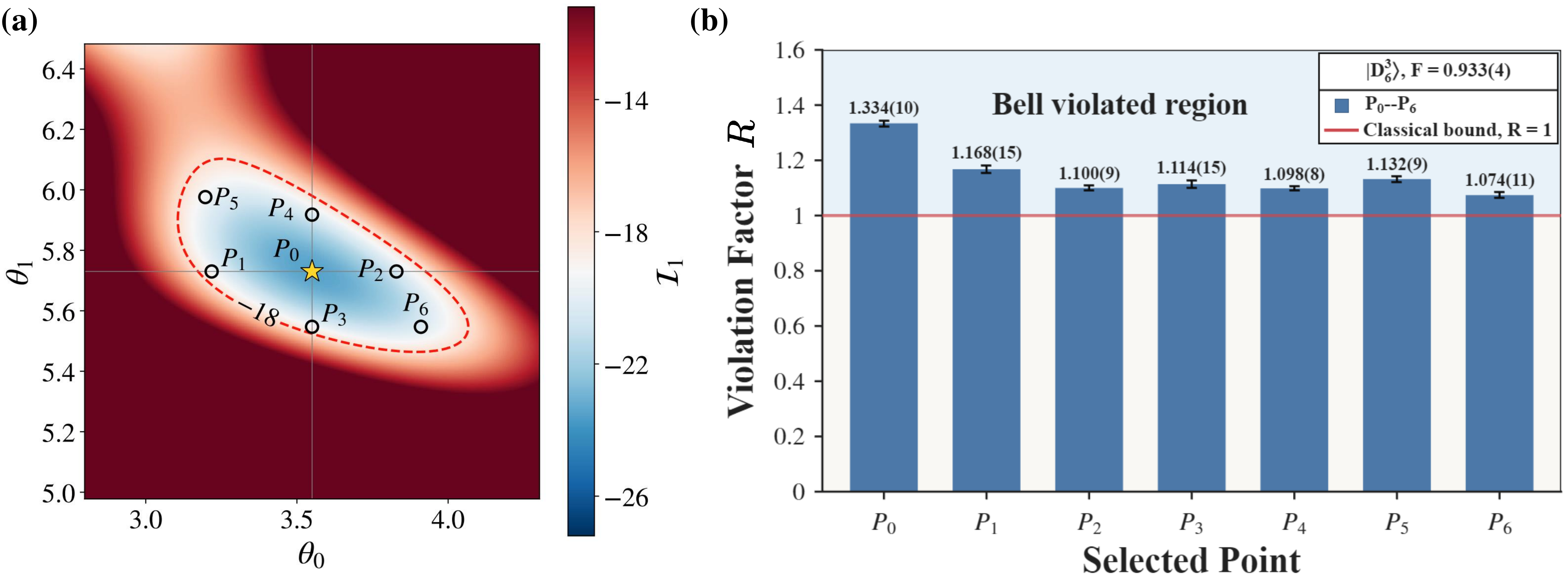}
\caption{Tolerance of $\mathcal{I}_1$ to collective misalignment of the measurement settings. (a) Predicted value $\mathcal{I}_1(\theta_0,\theta_1)$ for ideal $\ket{D_6^{\,3}}$ ($F=1$), as a function of the shared angles $(\theta_0,\theta_1)$. The red dashed contour marks the classical bound $\mathcal{I}_1=-18$ and Bell correlations are certified inside. The points $P_0-P_6$ denote experimental tested settings. (b) Coordinates and measured violation factors $R$ of the selected settings ($P_0-P_6$). Corresponding coordinates are listed in Table.\ref{tab:selected_points} (End Matter). Error bars denote one standard deviation obtained from Poisson parametric-bootstrap realizations of the raw coincidence counts.}
\label{fig:tolerance}
\end{figure*}

\textit{Discussion and outlook.}---
We have prepared high-fidelity six- and eight-photon Dicke states with fidelities of $0.933(4)$ and $0.916(11)$, respectively. The six-photon fidelity substantially exceeds the $0.654(24)$ reported in Ref.~\cite{wieczorek_experimental_2009}. These high-quality resources enable experimental certification of Bell correlations using PI witnesses containing only two- and four-body correlators. Measurements across different six-photon preparations distinguish witnesses with different noise tolerances, and the more noise-tolerant witness retains significant violations under common angular offsets of the local measurement bases. The eight-photon test extends the certification while retaining the same four symmetric correlators and maximal correlator order. These results establish that, for the experimentally prepared six- and eight-photon Dicke resources, multipartite Bell correlations can be certified without increasing the accessed correlation order, providing an experimentally practical alternative to full-body Bell measurements. 
More broadly, few-body PI Bell correlation witnesses offer a
correlation-order-scalable approach to probing nonclassical correlations
in increasingly large multipartite systems.

\vspace{.2cm}\noindent
\textit{Acknowledgements.}
This work was supported by Quantum Science and Technology-National Science and Technology Major Project (Grant No. 2025ZD0301000) ,  by the National Natural Science Foundation of China under the General Program, Grant No. 12474487, Shanghai Pilot Program for Basic Research
(TQ20240204), the Science and Technology Commission of Shanghai Municipality (231c1402900), and Shanghai Science and Technology Innovation Action Plan (Grant No.~24LZ1400200).
M.H., P.C. and V.S. are supported by the National Research Foundation, Singapore through the National Quantum Office, hosted in A*STAR, under its Centre for Quantum Technologies Funding Initiative (S24Q2d0009).
J.F.C. and J.T. acknowledge the support received from the European Union's Horizon Europe research and innovation programme through the ERC StG FINE-TEA-SQUAD (Grant No.~101040729).
This publication is part of the `Quantum Inspire - the Dutch Quantum Computer in the Cloud' project (with project number [NWA.1292.19.194]) of the NWA research program `Research on Routes by Consortia (ORC)', which is funded by the Netherlands Organization for Scientific Research (NWO).  The views and opinions expressed here are solely those of the authors and do not necessarily reflect those of the funding institutions. None of the funding institutions can be held responsible for them.

\vspace{.2cm}\noindent
\textit{Data and Code availability.}---Code and data will be released upon publication.
\bibliography{references}

\onecolumngrid
\begin{center}
    \textbf{\large End Matter}
\end{center}

\begin{table}[H]
\caption{PI Bell correlation witnesses evaluated in this work. For each witness, we list its explicit form with the classical bound $\Bc$, computed exactly on the vertices of their corresponding PI local polytopes. The quantum value $\Bq=\min_{\boldsymbol{\theta}}\bra{D_N^{\,N/2}}\hat{I}(\boldsymbol{\theta})\ket{D_N^{\,N/2}}$ on the ideal target state and the shared measurement angles $(\theta_0,\theta_1)$, in radians, at which it is attained. The critical white-noise fraction $p_c=1-\Bc/\Bq$, corresponding to the critical fidelity $F_c=1-p_c(1-2^{-N})$ quoted in the text. We also list the measured value $I$, violation factor $R=I/\Bc$, and significance $(|I|-|\Bc|)/\sigma_I$ at the highest preparation fidelity, $F=0.933(4)$ for $N=6$ and $F=0.916(11)$ for $N=8$. Numbers in parentheses are one-standard-deviation uncertainties in the last digits, from a Poisson parametric bootstrap of the raw coincidence counts (see Supplemental Material below). $\mathcal{I}_{\mathrm{T}}$ is the facet of its corresponding PI local polytope with two-body correlators~\cite{tura_detecting_2014,tura_nonlocality_2015}.  $\mathcal{I}_1$, $\mathcal{I}_2$, and $\mathcal{I}_3$ are facets of their corresponding PI local polytopes with two- and four-body correlators (see Supplemental Material below).}
\label{tab:witnesses}
\setlength{\tabcolsep}{3pt}
\begin{ruledtabular}
\begin{tabular}{llcccccc}
\multicolumn{2}{l}{Witness} & $\Bq$ & $(\theta_0,\theta_1)$ & $p_c$ & Measured $I$ & $R$ & Significance \\
\hline
\multicolumn{8}{l}{$(6,2,2)$ scenario, target $\ket{D_6^{\,3}}$, $p_\mathrm{exp}\approx0.068$:} \\
$\mathcal{I}_{\mathrm{T}}$ & $15\,\mathcal{S}_{00} + 6\,\mathcal{S}_{01} - \mathcal{S}_{11} \geq -120$ & $-126.041$ & $(3.264,\,2.165)$ & $0.048$ & $-119.74(166)$ & $0.998(14)$ & none \\
$\mathcal{I}_1$ & $\mathcal{S}_{00} + \mathcal{S}_{11} - \tfrac{1}{4}\,\mathcal{S}_{0011} + \tfrac{1}{6}\,\mathcal{S}_{1111} \geq -18$ & $-25.764$ & $(3.550,\,5.730)$ & $0.301$ & $-24.012(186)$ & $1.334(10)$ & $32\sigma$ \\
$\mathcal{I}_2$ & $15\,\mathcal{S}_{00} - 6\,\mathcal{S}_{11} - \tfrac{5}{4}\,\mathcal{S}_{0011} + \tfrac{1}{8}\,\mathcal{S}_{1111} \geq -135$ & $-157.907$ & $(2.864,\,0.749)$ & $0.145$ & $-146.58(168)$ & $1.086(12)$ & $6.9\sigma$ \\
\hline
\multicolumn{8}{l}{$(8,2,2)$ scenario, target $\ket{D_8^{\,4}}$, $p_\mathrm{exp}\approx0.084$:} \\
$\mathcal{I}_3$ & $\tfrac{3}{2}\,\mathcal{S}_{00} + \tfrac{9}{2}\,\mathcal{S}_{11} - \tfrac{1}{4}\,\mathcal{S}_{0011} + \tfrac{1}{4}\,\mathcal{S}_{1111} \geq -48$ & $-62.014$ & $(0.369,\,5.844)$ & $0.226$ & $-55.905(889)$ & $1.165(19)$ & $8.9\sigma$ \\
\end{tabular}
\end{ruledtabular}
\end{table}

\twocolumngrid
\raggedbottom

\textit{Witnesses used in this work.}---Table~\ref{tab:witnesses} collects the four PI Bell correlation witnesses evaluated in the experiment. $\mathcal{I}_1$, $\mathcal{I}_2$, and $\mathcal{I}_3$ are sparse members of the two- and four-body PI family of Eq.~\eqref{eq:PIansatz} with two settings, whose eight coefficients $\boldsymbol{\alpha}=(\alpha_{00},\alpha_{01},\alpha_{11};\alpha_{0000},\alpha_{0001},\alpha_{0011},\alpha_{0111},\alpha_{1111})$ multiply $\mathcal{S}_{00}/2$, $\mathcal{S}_{01}$, $\mathcal{S}_{11}/2$, $\mathcal{S}_{0000}/4!$, $\mathcal{S}_{0001}/3!$, $\mathcal{S}_{0011}/(2!)^2$, $\mathcal{S}_{0111}/3!$, and $\mathcal{S}_{1111}/4!$, respectively. In the three selected witnesses, the two-body correlator $\mathcal{S}_{01}$ and the four-body correlators $\mathcal{S}_{0000}$, $\mathcal{S}_{0001}$, and $\mathcal{S}_{0111}$ carry vanishing weight, so that each witness is evaluated from the $\binom{N}{2}$ configurations with two photons in setting $0$ and $N-2$ photons in setting $1$. The optimal member of the family reaches $p_c\approx0.46$ for $N=6$ and $0.38$ for $N=8$ but places large weights on all eight correlators (see Supplemental Material below).

\textit{Noise robustness versus system size.}---
To assess how the advantage of four-body correlators evolves with
system size, Fig.~\ref{fig:p_crit_vs_n} shows the largest critical
noise threshold $p_c$ we found within each PI family for the
half-excitation Dicke states $\ket{D_N^{\,N/2}}$ up to $N=10$.
The two- and four-body family with two measurement settings
reaches $p_c\approx0.53$, $0.46$, $0.38$, and $0.37$ for
$N=4$, $6$, $8$, and $10$, whereas the two-body families with
two and three measurement settings stay below $0.1$ for $N\geq6$.
The slow decrease of the four-body threshold indicates that the fixed-order strategy remains applicable beyond the sizes demonstrated here. 

\begin{figure}[H]
    \centering
    \includegraphics[width=\columnwidth]{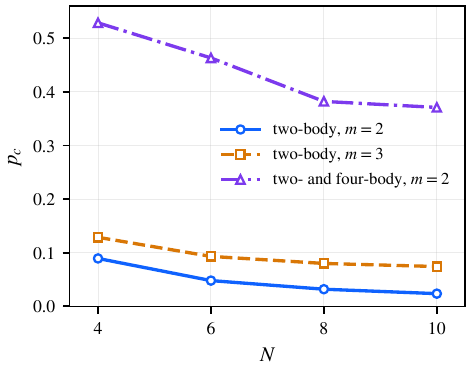}
    \caption{
    Optimal critical noise threshold $p_c$ of PI Bell inequalities
    for the half-excitation Dicke states $\ket{D_N^{\,N/2}}$,
    $N=4$--$10$, obtained by a see-saw optimization alternating a
    linear program over the vertices of $\mathcal{L}_\Pi$ with
    measurement-angle optimization in the symmetric subspace.
    The three families contain two-body correlators with $m=2$ or
    $m=3$ settings, or two- and four-body correlators with $m=2$.
    }
    \label{fig:p_crit_vs_n}
\end{figure}

\textit{Measurement-angle tolerance.}---
Table~\ref{tab:selected_points} lists the coordinates of the
measurement settings $P_0$--$P_6$ in the angle-tolerance test of
Fig.~\ref{fig:tolerance}, together with the measured violation factors of
$\mathcal{I}_1$. All six displaced settings violate the classical
bound $\Bc=-18$ by $6.8\sigma$ to $14\sigma$, confirming that the
violation does not rely on finely tuned measurement directions.

\begin{table}[H]
\caption{
Coordinates and measured violation factors of the selected settings
for $\mathcal{I}_1$ at $F=0.933(4)$. In the
$(\theta_0,\theta_1)$ plane, $P_0$ denotes the optimal setting and
$P_1$--$P_6$ denote different basis deviations. Error bars represent
one standard deviation obtained from Poisson parametric bootstrapping
of the raw sixfold coincidence counts.
}
\label{tab:selected_points}
\centering
\begin{ruledtabular}
\begin{tabular}{cccc}
Setting
& $(\theta_0,\theta_1)$ (rad)
& $R$
& Significance \\
\hline
$P_0$ & $(3.550,\;5.730)$ & $1.334 \pm 0.010$ & $32\sigma$ \\
$P_1$ & $(3.218,\;5.730)$ & $1.168 \pm 0.015$ & $12\sigma$ \\
$P_2$ & $(3.829,\;5.730)$ & $1.100 \pm 0.009$ & $12\sigma$ \\
$P_3$ & $(3.550,\;5.547)$ & $1.114 \pm 0.015$ & $7.8\sigma$ \\
$P_4$ & $(3.550,\;5.918)$ & $1.098 \pm 0.008$ & $13\sigma$ \\
$P_5$ & $(3.197,\;5.976)$ & $1.132 \pm 0.009$ & $14\sigma$ \\
$P_6$ & $(3.910,\;5.547)$ & $1.074 \pm 0.011$ & $6.8\sigma$ \\
\end{tabular}
\end{ruledtabular}
\end{table}

\textit{Measured PI correlators.}---Figure~\ref{fig:correlators} compares the two- and four-body PI correlators entering $\mathcal{I}_1$, $\mathcal{I}_2$, $\mathcal{I}_3$, and the displaced settings $P_5$ and $P_6$ with the predictions for the ideal Dicke states as a function of the measurement angle.

\begin{figure}[H]
\centering
\includegraphics[width=0.43\textwidth]{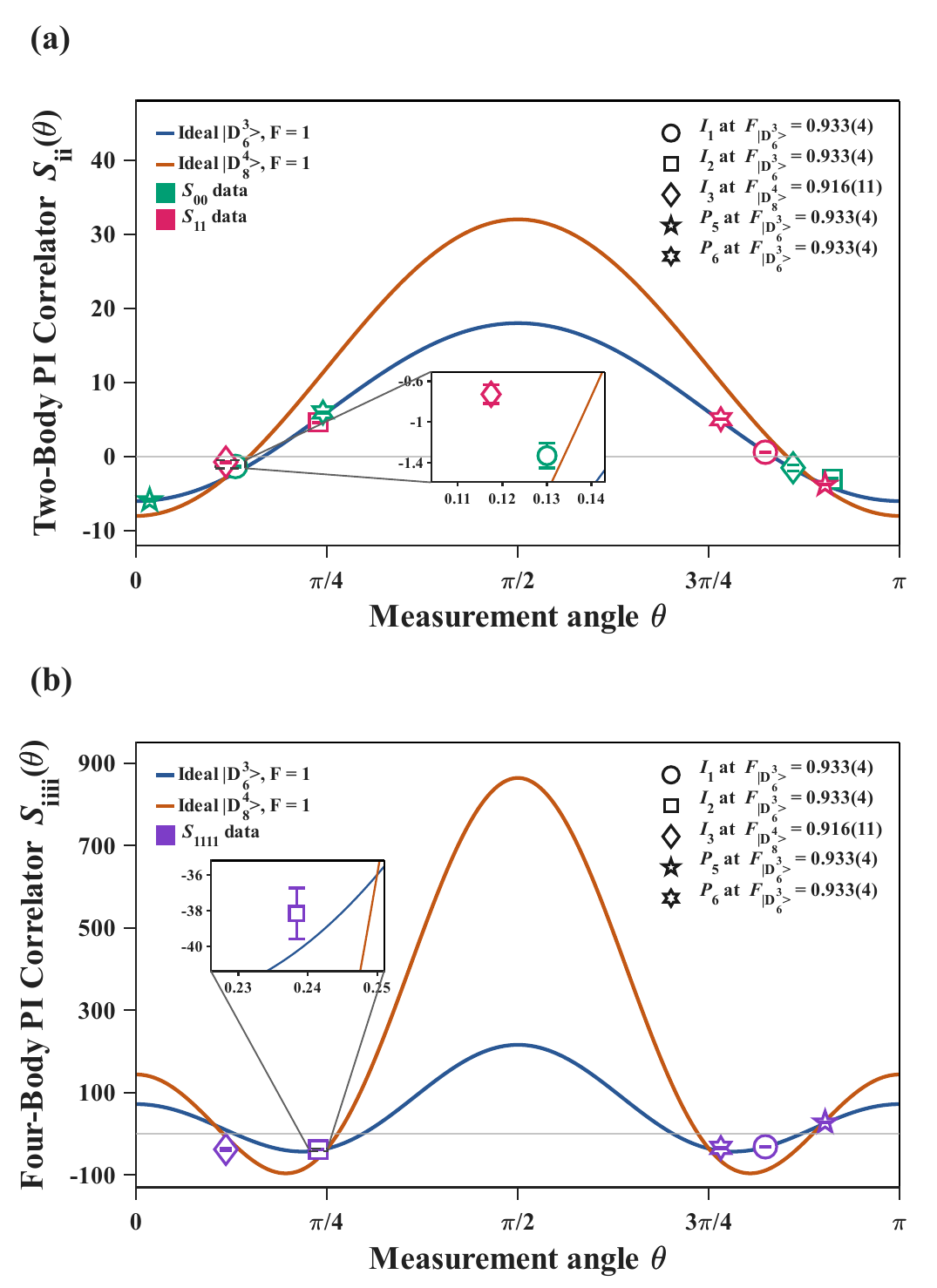}
\caption{Angular dependence of (a) the two-body PI correlators $\mathcal{S}_{xx}$, $x\in\{0,1\}$, and (b) the four-body PI correlator $\mathcal{S}_{1111}$. Here $\theta$ specifies the measurement direction $A(\theta)=\cos\theta\,\sigma_z+\sin\theta\,\sigma_x$ in the $xz$ plane. The solid blue and orange curves show the predictions for the ideal states $\ket{D_6^{\,3}}$ and $\ket{D_8^{\,4}}$, respectively. Hollow symbols denote the measured correlator sums obtained from the measurement settings of $\mathcal{I}_1$, $\mathcal{I}_2$, and $\mathcal{I}_3$ and from the displaced settings $P_5$ and $P_6$. Colors distinguish the measured correlators, while marker shapes identify the corresponding measurement sets.
Because several expectation values are extracted from the same multiphoton coincidence data, their statistical correlations are retained using a Poisson parametric bootstrap (see Supplemental Material). Error bars denote one standard deviation obtained from Poisson parametric-bootstrap realizations of the raw coincidence counts.}
\label{fig:correlators}
\end{figure}
\clearpage
\onecolumngrid
\setcounter{section}{0}
\setcounter{subsection}{0}
\setcounter{subsubsection}{0}
\setcounter{equation}{0}
\setcounter{figure}{0}
\setcounter{table}{0}
\setcounter{secnumdepth}{3}
\renewcommand{\thesection}{S\Roman{section}}
\renewcommand{\theequation}{S\arabic{equation}}
\renewcommand{\thefigure}{S\arabic{figure}}
\renewcommand{\thetable}{S\Roman{table}}
\renewcommand{\theHsection}{supp.\arabic{section}}
\renewcommand{\theHsubsection}{supp.\arabic{section}.\arabic{subsection}}
\renewcommand{\theHsubsubsection}{supp.\arabic{section}.\arabic{subsection}.\arabic{subsubsection}}
\renewcommand{\theHequation}{supp.\arabic{equation}}
\renewcommand{\theHfigure}{supp.\arabic{figure}}
\renewcommand{\theHtable}{supp.\arabic{table}}

\begin{center}
{\large\bfseries Supplemental Material}\\[0.3em]
{\bfseries Experimental certification of multipartite Bell correlations using only few-body symmetric correlations}
\end{center}

\tableofcontents
\section{Permutation-invariant Bell Inequalities}
\label{supp:sec:pi_bell}
This section introduces the theory of permutation-invariant (PI) Bell inequalities used in this work: the Bell scenario, the PI ansatz and its classical bound, the parity selection rule obeyed by the target Dicke states and the evaluation of the quantum value with the target Dicke state fixed. We also show that two-body PI Bell inequalities cannot certify
the experimentally prepared states.
In the main text, these inequalities are used as \emph{Bell correlation witnesses}: a violation of the classical bound by the measured, postselected correlations certifies that the prepared state is of Bell correlation. Since the classical bound holds for arbitrary local measurements, this conclusion requires no calibration of the local settings. On the other hand, the data are postselected on $N$-fold coincidences and no loophole is closed, so no loophole-free or device-independent statement is made.

\subsection{Multipartite Bell scenarios}
\label{supp:sec:scenarios}

Consider $N$ spatially separated parties $A_1,\dots,A_N$ sharing an
$N$-partite resource. Each party $A_i$ chooses one of $m$ 
measurements, labelled $x_i\in\{0,\dots,m-1\}$, with outcomes
$a_i=\pm 1$. We call this the Bell scenario $(N,m,2)$. The experiment
can be described by the behavior $\{P(\mathbf{a}|\mathbf{x})\}$.
By non-signaling, marginal distributions on any subset of parties are well
defined.

Correlations that admit a local hidden variable (LHV) model form a polytope $\mathcal{L}$ whose vertices are the local
deterministic strategies, in which party $A_i$ outputs a preassigned
value $s^{(i)}_{x}=\pm 1$ for every setting $x$.
A
single party thus has $2^{m}$ deterministic strategies, collected in
the set $\LDSset$. Quantum correlations arise from Born's rule,
$P(\mathbf{a}|\mathbf{x})
=\Tr[\rho_N\,\mathcal{M}^{a_1}_{1,x_1}\otimes\cdots\otimes
\mathcal{M}^{a_N}_{N,x_N}]$, and form a convex set
$\mathcal{Q}$.

For dichotomic outcomes it is convenient to work with correlators. 
For
any subset of parties $\{i_1,\dots,i_r\}$ and settings
$x_1,\dots,x_r$,
\begin{equation}\label{eq:correlator_def}
\expval*{A_{x_1}^{(i_1)}\cdots A_{x_r}^{(i_r)}}
=\sum_{\mathbf{a}} a_{i_1}\cdots a_{i_r}\,P(\mathbf{a}|\mathbf{x}),
\end{equation}
the expectation value of the product of the corresponding outcomes. On
a local deterministic strategy (LDS) every correlator factorizes,
$\expval*{A_{x_1}^{(i_1)}\cdots A_{x_r}^{(i_r)}}
=s^{(i_1)}_{x_1}\cdots s^{(i_r)}_{x_r}$.

A linear Bell inequality is a bound $\mathcal{I}\geq\Bc$ satisfied by
all of $\mathcal{L}$, where $\mathcal{I}$ is a linear combination of
correlators and the classical bound $\Bc$ is attained at a vertex. The
tight inequalities are the facets of $\mathcal{L}$. A canonical
example is the CHSH inequality
$\expval{A_0 B_0}+\expval{A_0 B_1}+\expval{A_1 B_0}
-\expval{A_1 B_1}\geq -2$ for $N=m=2$. The corresponding quantum
optimum $\Bq=\inf_{\mathcal{Q}}\mathcal{I}$ is the Tsirelson bound~\cite{cirelson_quantum_1980}.

\subsection{Permutation-invariant Bell inequalities with few-body correlators}
\label{supp:sec:pi_fewbody}

A Bell expression is permutation-invariant (PI) if it is unchanged
under any relabeling of the parties. It can then be written in terms of
the symmetrized correlators
\begin{equation}\label{eq:Scorr_supp}
\mathcal{S}_{x_1\cdots x_r}
=\sum_{(i_1,\dots,i_r)}
\expval*{A_{x_1}^{(i_1)}\cdots A_{x_r}^{(i_r)}},
\end{equation}
where the sum runs over all $r$-tuples of pairwise different parties.
A
general PI Bell inequality with correlators up to $k$-body order reads
\begin{equation}\label{eq:PIansatz_supp}
I(\boldsymbol{\alpha})
= \sum_{r=1}^{k}\frac{1}{r!}\sum_{x_1,\dots,x_r=0}^{m-1}
\alpha_{x_1\cdots x_r}\,\mathcal{S}_{x_1\cdots x_r}
\;\geq\;\Bc(\boldsymbol{\alpha}),
\end{equation}
with coefficients symmetric in their indices; this is Eq.~(1) of the
main text. 
For example, for $m=2$ settings and up to two-body correlators, $k=2$, it
reduces to the five-parameter family of PI Bell inequalities~\cite{tura_detecting_2014},
\begin{equation}\label{eq:PI_22_5coeff}
I(\boldsymbol{\alpha})=\alpha_0\,\mathcal{S}_0+\alpha_1\,\mathcal{S}_1
+\frac{\alpha_{00}}{2}\,\mathcal{S}_{00}
+\alpha_{01}\,\mathcal{S}_{01}
+\frac{\alpha_{11}}{2}\,\mathcal{S}_{11}\;\geq\;\Bc(\boldsymbol{\alpha}) .
\end{equation}

To evaluate the classical bound $\Bc(\boldsymbol{\alpha})$, we collect the
independent symmetrized correlators into a vector
$\vec{\mathcal{S}}
=\bigl(\mathcal{S}_{x_1\cdots x_r}\bigr)_{x_1\le\cdots\le x_r,\,r\le k}
\in\mathbb{R}^{D}$, with $D=\sum_{r=1}^{k}\binom{m+r-1}{r}$.  
The Bell
expression \eqref{eq:PIansatz_supp} depends on the underlying
correlations only through $\vec{\mathcal{S}}$ and is linear in it. For
a point $v\in\mathbb{R}^{D}$, we write $I(\boldsymbol{\alpha};v)$ for
\eqref{eq:PIansatz_supp} evaluated at the correlator values $v$, so
that $I(\boldsymbol{\alpha})=I(\boldsymbol{\alpha};\vec{\mathcal{S}})$.
The classical bound can therefore be computed on the PI local polytope
$\mathcal{L}_\Pi\subset\mathbb{R}^{D}$, the projection of the local
polytope $\mathcal{L}$ onto the coordinates \eqref{eq:Scorr_supp}. Hence the classical bound
\begin{equation}\label{eq:Bc_vertex}
  \Bc(\boldsymbol{\alpha})
  =\min_{v\in\mathrm{vert}(\mathcal{L}_\Pi)} I(\boldsymbol{\alpha};v),
\end{equation}
where $\mathrm{vert}(\mathcal{L}_\Pi)$ denotes the vertex set of
$\mathcal{L}_\Pi$.  Every vertex of $\mathcal{L}_\Pi$ is the image of
an LDS under this projection, and since the coordinates
\eqref{eq:Scorr_supp} are invariant under relabeling the parties, the
image depends only on \emph{how many} parties apply each of the $2^m$
single-party deterministic strategies in $\LDSset$, not on which ones.
Distinct images are thus labeled by the counts
$(n_1,\dots,n_{2^m})$, $n_s\ge 0$, $\sum_s n_s=N$, so the number of
vertices is at most
$\binom{N+2^m-1}{2^m-1}=O(N^{2^m-1})$. 
So it is polynomial in $N$ at fixed $m$,
in contrast to the $2^{mN}$ vertices of $\mathcal{L}$.  
In the
two-body $(N,2,2)$ case, $D=5$ and $\mathcal{L}_\Pi\subset\mathbb{R}^5$
has $2(N^2+1)$ vertices~\cite{tura_detecting_2014}.
This keeps $\Bc(\boldsymbol{\alpha})$ exactly computable at
every system size considered in this work.

\subsection{Parity selection rule for half-excitation Dicke states}
\label{supp:sec:parity}

For even $N$, the half-excitation Dicke state is
\begin{equation}
\label{eq:dicke_supp}
    \ket{D_N^{\,N/2}} = \binom{N}{N/2}^{-1/2}
    \sum_i\mathcal{P}_i (|\underbrace{11\cdots 1}_{N/2} 0\cdots 0 \rangle ),
\end{equation}
where the sum runs over the $\binom{N}{N/2}$ distinct permutations
$\mathcal{P}_i$ of the qubits.
 
\begin{theorem}[Parity selection rule]\label{thm:selection}
Every correlator of odd order $r$ built from measurements in the $xz$ plane, $A_x=\cos\theta_x\,\sigma_z+\sin\theta_x\,\sigma_x$ (Eq.~\eqref{eq:sym_meas_supp} below), vanishes on
$\ket{D_N^{\,N/2}}$, that is,
\begin{equation}\label{eq:selection}
\avgD{A_{x_1}^{(i_1)}\cdots A_{x_r}^{(i_r)}}=0
\qquad\text{for } r \text{ odd}.
\end{equation}
In particular, all PI correlators $\mathcal{S}_{x_1\cdots x_r}$ with
odd $r$ vanish. 
More generally, a Pauli string on distinct sites containing $a$ factors $\sigma_x$ and $c$ factors $\sigma_z$ has vanishing expectation value on $\ket{D_N^{\,N/2}}$ unless $a$ and $c$ are both even.
\end{theorem}
 
\begin{proof}
Writing $A_{x}=w_z(\theta_x)\,\sigma_z+w_x(\theta_x)\,\sigma_x$ with
$w_z(\theta)=\cos\theta$ and $w_x(\theta)=\sin\theta$, the product of
$r$ observables on distinct sites expands into $2^{r}$ Pauli strings,
\begin{equation}\label{eq:corr_expansion}
A_{x_1}^{(i_1)}\cdots A_{x_r}^{(i_r)}
=\sum_{\mathbf{s}\in\{x,z\}^{r}}
\Bigl[\prod_{j=1}^{r} w_{s_j}(\theta_{x_j})\Bigr]\,P_{\mathbf{s}},
\qquad
P_{\mathbf{s}}=\sigma_{s_1}^{(i_1)}\cdots\sigma_{s_r}^{(i_r)},
\end{equation}
so it suffices to show that $\avgD{P_{\mathbf{s}}}=0$ for every string
$P_{\mathbf{s}}$ containing $a$ factors $\sigma_x$ and $c=r-a$ factors
$\sigma_z$ with $a$ or $c$ odd.
Consider the global parity operators
$\sigma_z^{\otimes N}=\bigotimes_{i=1}^{N}\sigma_z^{(i)}$ and
$\sigma_x^{\otimes N}=\bigotimes_{i=1}^{N}\sigma_x^{(i)}$. Both
stabilize the state,
\begin{equation}
\label{eq:parity_stab}
    \sigma_z^{\otimes N}\ket{D_N^{\,N/2}}=(-1)^{N/2}\ket{D_N^{\,N/2}}, \qquad \sigma_x^{\otimes N}\ket{D_N^{\,N/2}}=\ket{D_N^{\,N/2}},
\end{equation}
since every basis state in \eqref{eq:dicke_supp} carries $N/2$
excitations, and $\sigma_x^{\otimes N}$ maps it to its complement, again with $N/2$
excitations. Since $\sigma_z^{\otimes N}$ anticommutes with $\sigma_x$ and
commutes with $\sigma_z$, while $\sigma_x^{\otimes N}$ does the opposite,
$\sigma_z^{\otimes N}P_{\mathbf{s}}\sigma_z^{\otimes N}=(-1)^a P_{\mathbf{s}}$ and
$\sigma_x^{\otimes N}P_{\mathbf{s}}\sigma_x^{\otimes N}=(-1)^c P_{\mathbf{s}}$.
Combined with Eq.~\eqref{eq:parity_stab} and $[(-1)^{N/2}]^2=1$, this gives
\begin{equation}\label{eq:parity_signs}
\avgD{P_{\mathbf{s}}}=\avgD{\sigma_z^{\otimes N} P_{\mathbf{s}}\,\sigma_z^{\otimes N}}=(-1)^{a}\avgD{P_{\mathbf{s}}},
\qquad
\avgD{P_{\mathbf{s}}}=\avgD{\sigma_x^{\otimes N} P_{\mathbf{s}}\,\sigma_x^{\otimes N}}=(-1)^{c}\avgD{P_{\mathbf{s}}},
\end{equation}
so $\avgD{P_{\mathbf{s}}}=0$ unless $a$ and $c$ are both even, which is impossible
for odd $r=a+c$. Every term in \eqref{eq:corr_expansion} therefore vanishes, which proves \eqref{eq:selection}. Since
$\mathcal{S}_{x_1\cdots x_r}$ is a sum of such correlators over
pairwise distinct parties, all PI correlators with odd $r$ vanish as
well.
\end{proof}
 
\begin{cor}\label{cor:even}
Let $I(\boldsymbol{\alpha})\ge\Bc(\boldsymbol{\alpha})$ be a PI Bell
inequality with measurements in the $xz$ plane, and let
$\boldsymbol{\alpha}_{\mathrm{even}}$ retain only its even-order
coefficients.  Then
\begin{equation}\label{eq:even_wlog}
\Bc(\boldsymbol{\alpha}_{\mathrm{even}})\ge\Bc(\boldsymbol{\alpha}),
\qquad
\bigl\langle I(\boldsymbol{\alpha}_{\mathrm{even}})\bigr\rangle_{\rho(p)}
=\bigl\langle I(\boldsymbol{\alpha})\bigr\rangle_{\rho(p)},
\end{equation}
with
$\rho(p)=(1-p)\ket{D_N^{\,N/2}}\!\bra{D_N^{\,N/2}}+p\,\mathbb{I}/2^N$,
so the critical noise strengths obey
$p_c(\boldsymbol{\alpha}_{\mathrm{even}})\ge p_c(\boldsymbol{\alpha})$.
The search for noise-robust PI Bell inequalities for
$\ket{D_N^{\,N/2}}$ can therefore be restricted to even-order
correlators.
\end{cor}
 
\begin{proof}
Flipping all outcomes, $a_x^{(i)}\mapsto-a_x^{(i)}$, maps local
deterministic strategies to local deterministic strategies and
multiplies every correlator of order $r$ by $(-1)^r$.  On
$\mathbb{R}^{D}$ it therefore acts as the linear involution $T$ that
flips the sign of the odd-order coordinates.
Since flipping outcomes
commutes with permuting parties, $T$ maps the PI local polytope $\mathcal{L}_\Pi$ onto
itself and permutes its vertices.  Now
$I(\boldsymbol{\alpha}_{\mathrm{even}};v)
=\tfrac12\bigl[I(\boldsymbol{\alpha};v)+I(\boldsymbol{\alpha};Tv)\bigr]$,
and by \eqref{eq:Bc_vertex} both terms are bounded below by
$\Bc(\boldsymbol{\alpha})$ on $\mathrm{vert}(\mathcal{L}_\Pi)$, so
\begin{equation}
\Bc(\boldsymbol{\alpha}_{\mathrm{even}})
=\min_{v\in\mathrm{vert}(\mathcal{L}_\Pi)}
\tfrac12\bigl[I(\boldsymbol{\alpha};v)+I(\boldsymbol{\alpha};Tv)\bigr]
\ge\Bc(\boldsymbol{\alpha}),
\end{equation}
the first relation in~\eqref{eq:even_wlog}.  The second relation
holds because $\rho(p)$ cannot see the odd part of
$I(\boldsymbol{\alpha})$: odd-order correlators vanish on
$\ket{D_N^{\,N/2}}$ by \Cref{thm:selection}, and all correlators
vanish on the white-noise part, the observables being traceless and
acting on distinct sites.  Combining the two relations, every
violation of $I(\boldsymbol{\alpha})\ge\Bc(\boldsymbol{\alpha})$ by
$\rho(p)$ is also a violation of
$I(\boldsymbol{\alpha}_{\mathrm{even}})
\ge\Bc(\boldsymbol{\alpha}_{\mathrm{even}})$, and
$p_c(\boldsymbol{\alpha}_{\mathrm{even}})\ge p_c(\boldsymbol{\alpha})$
follows.
\end{proof}

Note that PI is not essential in the above argument, which applies to any Bell expression with measurements in the $xz$ plane. 
Throughout this work we therefore restrict the search for
noise-robust PI Bell inequalities for $\ket{D_N^{\,N/2}}$ to even-order correlators.

\subsection{Quantum value on the target Dicke state}
\label{supp:sec:quantum_value}

Here we have two different notions of ``quantum value'': the unrestricted optimum over all states and measurements, which benchmarks the inequality itself, and the state-fixed optimum, which benchmarks a given target state. We introduce both and explain how each is computed. 
Throughout this work, the certification of the prepared Dicke states uses the state-fixed value.

To evaluate them we use the
symmetric measurement
strategy~\cite{tura_detecting_2014,tura_nonlocality_2015}, in which
all parties measure the same observables,
\begin{equation}\label{eq:sym_meas_supp}
A^{(i)}_x=\cos\theta_x\,\sigma_z^{(i)}+\sin\theta_x\,\sigma_x^{(i)},
\qquad x=0,\dots,m-1,
\end{equation}
where $\sigma_{x}^{(i)},\sigma_{y}^{(i)},\sigma_{z}^{(i)}$ are the
Pauli operators of party $i$ and
$\boldsymbol{\theta}=(\theta_0,\dots,\theta_{m-1})$ are measurement
angles shared by all parties.
For $m=2$ this choice involves no loss of generality~\cite{masanes_asymptotic_2006}.
Beyond this
case the reductions are not known to be complete, and
Eq.~\eqref{eq:sym_meas_supp} serves as a variational ansatz.
We have verified numerically that optimization over arbitrary Bloch vectors does not improve any of the optima reported in this work and
always returns to coplanar settings.
Then each correlator
\eqref{eq:Scorr_supp} is the expectation value of the operator
$\hat{\mathcal{S}}_{x_1\cdots x_r}(\boldsymbol{\theta})
=\sum_{(i_1,\dots,i_r)}A_{x_1}^{(i_1)}\cdots A_{x_r}^{(i_r)}$, and the Bell expression \eqref{eq:PIansatz_supp} becomes the expectation value of the Bell operator
\begin{equation}\label{eq:Bell_operator_supp}
\hat{I}(\boldsymbol{\theta})
=\sum_{r=1}^{k}\frac{1}{r!}\sum_{x_1,\dots,x_r=0}^{m-1}
\alpha_{x_1\cdots x_r}\,
\hat{\mathcal{S}}_{x_1\cdots x_r}(\boldsymbol{\theta}),
\end{equation}
where the dependence of $\hat{I}$ on the coefficients $\boldsymbol{\alpha}$ is left implicit.
For one- and two-body correlators,
\begin{equation}\label{eq:collective_supp}
\hat{\mathcal{S}}_{x}=2\cos\theta_x\,S_z+2\sin\theta_x\,S_x,
\qquad
\hat{\mathcal{S}}_{xy}
=\tfrac12\{\hat{\mathcal{S}}_{x},\hat{\mathcal{S}}_{y}\}
-\hat{\mathcal{Z}}_{xy},
\end{equation}
where
\begin{equation}\label{eq:Zterm_supp}
\hat{\mathcal{Z}}_{xy}
=\tfrac12\sum_{i=1}^{N}\{A^{(i)}_x,A^{(i)}_y\}
=N\cos(\theta_x-\theta_y)\,\mathbb{I}
\end{equation}
subtracts the coinciding-party ($i=j$) terms contained in
$\hat{\mathcal{S}}_{x}\hat{\mathcal{S}}_{y}$~\cite{chen_optimizing_2025,chen_optimizing_2026}.  
Analogous expressions
at higher orders follow from a known
recursion~\cite{tura_nonlocality_2015,hassani_many-body_2025}.  The state-optimized quantum value of the Bell expression is
\begin{align}
\BqGS(\boldsymbol{\alpha})
=
\min_{\boldsymbol{\theta},\,\ket{\psi}}
\bra{\psi}
\hat{I}(\boldsymbol{\theta})
\ket{\psi},\label{eq:betaqexact}
\end{align}
where $\ket{\psi}$ is restricted to the symmetric (spin-$N/2$) subspace.  
Note that we can
diagonalize $\hat{I}(\boldsymbol{\theta})$ in the symmetric
(spin-$N/2$) subspace of dimension $N+1$ instead of $2^N$~\cite{tura_nonlocality_2015,chen_optimizing_2025,hassani_many-body_2025}. 

Since our aim is to certify a specific target state, the calculation used throughout this work is the state-fixed quantum value obtained by fixing the state to the target Dicke state,
\begin{align}
\Bq(\boldsymbol{\alpha})
=
\min_{\boldsymbol{\theta}}
\bra{D_N^{\,N/2}}
\hat{I}(\boldsymbol{\theta})
\ket{D_N^{\,N/2}},\label{eq:Bq_dicke}
\end{align}
which is the quantity denoted $\Bq$ in the main text and in all tables of this Supplemental Material. 
By construction $\BqGS(\boldsymbol{\alpha})\le\Bq(\boldsymbol{\alpha})$, and the two are in general different.

\subsection{Depolarizing noise, critical noise threshold, and fidelity}
\label{supp:sec:noise_threshold}
In this section, we consider the depolarized state
\begin{equation}\label{eq:rho_p}
\rho(p)=(1-p)\ketbra{D_N^{\,N/2}}{D_N^{\,N/2}}+p\,\frac{\openone}{2^N},
\qquad p\in[0,1].
\end{equation}
Every collective correlator operator $\hat{\mathcal{S}}_{x_1\cdots x_r}$
is traceless
so
$\Tr[\hat{I}(\boldsymbol{\theta})]=0$ and
\begin{equation}
\expval*{\hat{I}(\boldsymbol{\theta})}_{\rho(p)}
=(1-p)\,
\avgD{\hat{I}(\boldsymbol{\theta})}.
\end{equation}
The inequality is violated by $\rho(p)$ if and only if
$(1-p)\Bq(\boldsymbol{\alpha})<\Bc(\boldsymbol{\alpha})$, i.e., for all noise fractions below
the critical threshold
\begin{equation}\label{eq:pcrit_supp}
p_c(\boldsymbol{\alpha})=1-\frac{\Bc(\boldsymbol{\alpha})}{\Bq(\boldsymbol{\alpha})}.
\end{equation}

The fidelity of $\rho(p)$ with the target state is
$F=\bra{D_N^{\,N/2}}\rho(p)\ket{D_N^{\,N/2}}=1-p\,(1-2^{-N})$.
Conversely, a measured fidelity $F$ corresponds to the equivalent
white-noise fraction
\begin{equation}\label{eq:pexp}
\pexp=\frac{1-F}{1-2^{-N}},
\end{equation}
and the certification condition reads $\pexp<p_c(\boldsymbol{\alpha})$. 
The Dicke
states prepared in this work have fidelities $F=0.933(4)$, $0.797(5)$,
and $0.638(5)$ for $N=6$ and $F=0.916(11)$ for $N=8$ (Sec.~\ref{supp:sec:exp_details}]), corresponding to
\begin{equation}\label{eq:pexp_values}
\pexp\approx 0.068,\quad 0.206,\quad 0.368\ (N{=}6),
\qquad
\pexp\approx 0.084\ (N{=}8).
\end{equation}

\subsection{Two-body PI Bell inequalities with two measurements cannot certify the prepared states}
\label{supp:sec:twobody_insufficient}

The two-body PI Bell inequality of the $(6,2,2)$ scenario with the largest threshold for $\ket{D_6^{\,3}}$ is the five-parameter
inequality \eqref{eq:PI_22_5coeff} with
$(\alpha_0,\alpha_1,\alpha_{00},\alpha_{01},\alpha_{11})=(0,0,30,6,-2)$
\cite{tura_detecting_2014,tura_nonlocality_2015},
\begin{equation}\label{eq:Tura_supp}
\mathcal{I}_{\mathrm{T}}
= 15\,\mathcal{S}_{00} + 6\,\mathcal{S}_{01} - \mathcal{S}_{11}
\;\geq\; -120.
\end{equation}
Its quantum value on $\ket{D_6^{\,3}}$
admits a simple closed form. Using
$\ket{D_6^{\,3}}=\ket{J{=}3,M{=}0}$ and the total-spin identities
$\expval{S_z^2}=M^2=0$ and
$\expval{S_x^2}=\expval{S_y^2}=\tfrac12[J(J+1)-M^2]=6$, together with
$\sum_{i<j}\expval{\sigma_\mu^{(i)}\sigma_\mu^{(j)}}
=\tfrac12(4\expval{S_\mu^2}-N)$, the two-body correlations per (ordered)
pair are
\begin{equation}\label{eq:perpair}
  \avgDs{\sigma_z^{(i)}\sigma_z^{(j)}} = -\tfrac{1}{5},
  \qquad
  \avgDs{\sigma_x^{(i)}\sigma_x^{(j)}} = +\tfrac{3}{5},
  \qquad
  \avgDs{\sigma_x^{(i)}\sigma_z^{(j)}+\sigma_z^{(i)}\sigma_x^{(j)}} = 0,
\end{equation}
where the cross term vanishes by Theorem~\ref{thm:selection}. Expanding each
$A_x^{(i)}$ by \eqref{eq:sym_meas_supp} and summing over the $30$
ordered pairs,
\begin{equation}\label{eq:Skl_state}
\mathcal{S}_{xy}\big|_{D_6^{\,3}}
=18\,\sin\theta_x\sin\theta_y-6\,\cos\theta_x\cos\theta_y ,
\end{equation}
so that for any coefficients $(\alpha_{00},\alpha_{01},\alpha_{11})$ of
$\tfrac{\alpha_{00}}{2}\mathcal{S}_{00}+\alpha_{01}\,\mathcal{S}_{01}
+\tfrac{\alpha_{11}}{2}\mathcal{S}_{11}$,
\begin{align}
\avgDs{\hat{I}}&=-3\,Q_z+9\,Q_x,\nonumber\\
Q_z&=\alpha_{00}\cos^2\!\theta_0+2\alpha_{01}\cos\theta_0\cos\theta_1
+\alpha_{11}\cos^2\!\theta_1,\nonumber\\
Q_x&=\alpha_{00}\sin^2\!\theta_0+2\alpha_{01}\sin\theta_0\sin\theta_1
+\alpha_{11}\sin^2\!\theta_1.\label{eq:mBmC}
\end{align}
For $\mathcal{I}_{\mathrm{T}}$, minimizing \eqref{eq:mBmC} gives
$\Bq\approx-126.041$ at
$(\theta_0,\theta_1)\approx(3.264,\,2.165)$, hence
\begin{equation}
p_c(\mathcal{I}_{\mathrm{T}})=1-\frac{-120}{-126.041}\approx 0.048 .
\end{equation}
Certification with $\mathcal{I}_{\mathrm{T}}$ therefore requires
$F\gtrsim 95.3\%$, well above the prepared fidelity $F=0.933(4)$,
for which $\pexp\approx0.068>p_c$. This prediction is confirmed
experimentally: the measured value
$\mathcal{I}_{\mathrm{T}}=-119.74(1.66)$
(Table~\ref{tab:PI_GAO_Tura}) does not violate the classical bound.

For $N=8$, we have
$\max p_c\approx 0.032<\pexp\approx 0.084$. Thus no
PI Bell inequality with up to two-body correlators and two measurement
settings per party certifies the states prepared in this experiment,
at either system size.

Enlarging the scenario to $m=3$ settings per party improves the
optimal two-body threshold only moderately. For the $(6,3,2)$ scenario, the largest threshold we
found is attained by
\begin{equation}\label{eq:BI632}
2\,\mathcal{S}_{01}-2\,\mathcal{S}_{02}
+\tfrac{5}{2}\,\mathcal{S}_{11}+5\,\mathcal{S}_{12}
+\tfrac{15}{2}\,\mathcal{S}_{22}\;\geq\;-90,
\end{equation}
with $\Bq\approx-99.234$ at
$\boldsymbol{\theta}\approx(1.5415,\,3.4888,\,2.9695)$ and
$p_c\approx 0.093$. This exceeds $\pexp\approx0.068$ of the
highest-fidelity six-photon state only marginally, would require
$21$ additional measurement configurations, and leaves no margin for
the deliberately degraded states studied in the main text. 
For
$N=8$, the optimal $(8,3,2)$ two-body threshold we found is
$p_c\approx 0.080<\pexp\approx 0.084$, insufficient even for the
highest-fidelity eight-photon state. This motivates keeping $m=2$ and
extending the PI ansatz to four-body correlators instead.

\section{PI Bell inequalities with two- and four-body correlators}
\label{supp:sec:fourbody}

\subsection{Construction and optimization for $N=6$}
\label{supp:sec:fourbody_N6}

By Corollary~\ref{cor:even} we restrict the PI ansatz \eqref{eq:PIansatz_supp}
with $m=2$ and $k=4$ to even-body correlators, leaving the eight
independent coefficients
$\boldsymbol{\alpha}=(\alpha_{00},\alpha_{01},\alpha_{11};\,
\alpha_{0000},\alpha_{0001},\alpha_{0011},\alpha_{0111},\alpha_{1111})$ of
\begin{equation}\label{eq:BI_24body_supp}
\begin{split}
I(\boldsymbol{\alpha})
={}& \frac{\alpha_{00}}{2}\,\mathcal{S}_{00}
   + \alpha_{01}\,\mathcal{S}_{01}
   + \frac{\alpha_{11}}{2}\,\mathcal{S}_{11}
 + \frac{\alpha_{0000}}{4!}\,\mathcal{S}_{0000}
   + \frac{\alpha_{0001}}{3!}\,\mathcal{S}_{0001}\\
 &+ \frac{\alpha_{0011}}{(2!)^{2}}\,\mathcal{S}_{0011}
   + \frac{\alpha_{0111}}{3!}\,\mathcal{S}_{0111}
   + \frac{\alpha_{1111}}{4!}\,\mathcal{S}_{1111} \;\geq\; \Bc(\boldsymbol{\alpha}).
\end{split}
\end{equation}
The classical bound is evaluated exactly by
\eqref{eq:Bc_vertex}, and the quantum value on the target Dicke state by Eq.~\eqref{eq:Bq_dicke}. We then maximize the critical noise
threshold,
\begin{equation}\label{eq:pc_opt_problem}
\boldsymbol{\alpha}^{\star}
= \arg\max_{\boldsymbol{\alpha}}
    \left[1-\frac{\Bc(\boldsymbol{\alpha})}{\Bq(\boldsymbol{\alpha})}\right],
\end{equation}
by a see-saw optimization. Starting from random measurement angles $\boldsymbol{\theta}$, we alternate two steps: (i) at fixed $\boldsymbol{\theta}$, the coefficients are updated by the linear program $\min_{\boldsymbol{\alpha}}\avgD{\hat{I}(\boldsymbol{\theta})}$ subject to $I(\boldsymbol{\alpha};v)\ge-1$ for all $v\in\mathrm{vert}(\mathcal{L}_\Pi)$, i.e., the most violated inequality with classical bound $-1$ at the current angles; (ii) at fixed $\boldsymbol{\alpha}$, the angles are updated by minimizing $\avgD{\hat{I}(\boldsymbol{\theta})}$, which yields $\Bq(\boldsymbol{\alpha})$ and hence $p_c(\boldsymbol{\alpha})$. The iteration is repeated from many random starting points and the largest $p_c$ is retained. The same procedure gives the two-body thresholds quoted in Sec.~\ref{supp:sec:twobody_insufficient}.

The optimization converges to
\begin{equation}\label{eq:BI_24body_optimal_supp}
\begin{split}
I(\boldsymbol{\alpha}^{\star})
={}& \frac{4428}{2}\,\mathcal{S}_{00}
   - 252\,\mathcal{S}_{01}
   + \frac{4904}{2}\,\mathcal{S}_{11}
 - \frac{1757}{24}\,\mathcal{S}_{0000}
   - \frac{1599}{6}\,\mathcal{S}_{0001} \\
 &+ \frac{2987}{4}\,\mathcal{S}_{0011}
 + \frac{1753}{6}\,\mathcal{S}_{0111}
   - \frac{2471}{24}\,\mathcal{S}_{1111}
   \;\geq\; -35826,
\end{split}
\end{equation}
with
\begin{equation}
\Bc(\boldsymbol{\alpha}^{\star}) = -35826,
\quad
\Bq(\boldsymbol{\alpha}^{\star}) \approx -66730.3
\ \text{at}\
\boldsymbol{\theta}^{\star}\approx (4.5246,\,6.1724),
\quad
p_{c}(\boldsymbol{\alpha}^{\star})\approx 0.4631 .
\end{equation}
The threshold $p_c\approx 0.46$ far exceeds
$\pexp\approx 0.068$. Measuring \eqref{eq:BI_24body_optimal_supp} is
nevertheless unattractive: it weights all eight correlators, including
$\mathcal{S}_{01}$, $\mathcal{S}_{0000}$, $\mathcal{S}_{0001}$ and
$\mathcal{S}_{0111}$, which are not all accessible from a single family
of measurement configurations (Sec.~\ref{subsec:pi_bell_settings}), and its
large, incommensurate coefficients amplify the statistical uncertainty
of a finite-count experiment.

For the experiment we therefore select members of
the family \eqref{eq:BI_24body_supp} with few nonzero coefficients. Two
inequalities stand out, containing only the four correlators
$\mathcal{S}_{00}$, $\mathcal{S}_{11}$, $\mathcal{S}_{0011}$,
$\mathcal{S}_{1111}$, all of which are simultaneously accessible from
the $\binom{6}{2}=15$ configurations with two photons measured in
setting $0$ and four in setting $1$
(Sec.~\ref{subsec:pi_bell_settings}). The first,
$\boldsymbol{\alpha}^{\diamond}=(2,\,0,\,2;\;0,\,0,\,-1,\,0,\,4)$, reads
\begin{equation}\label{eq:BI_24body_sparse_supp}
\mathcal{I}_1
= \mathcal{S}_{00}
+ \mathcal{S}_{11}
- \frac{1}{4}\,\mathcal{S}_{0011}
+ \frac{1}{6}\,\mathcal{S}_{1111}
\;\geq\; -18,
\end{equation}
with
\begin{equation}\label{eq:facetC_values}
\Bc = -18,
\qquad
\Bq \approx -25.7638
\ \text{at}\
\boldsymbol{\theta}^{\diamond}\approx(3.5501,\,5.7302),
\qquad
p_{c}\approx 0.3014 .
\end{equation}

The second,
$\boldsymbol{\alpha}^{\circ}=(30,\,0,\,-12;\;0,\,0,\,-5,\,0,\,3)$, reads
\begin{equation}\label{eq:BI_24body_facetG_supp}
\mathcal{I}_2
= 15\,\mathcal{S}_{00}
- 6\,\mathcal{S}_{11}
- \frac{5}{4}\,\mathcal{S}_{0011}
+ \frac{1}{8}\,\mathcal{S}_{1111}
\;\geq\; -135,
\end{equation}
with
\begin{equation}\label{eq:facetG_values}
\Bc = -135,
\qquad
\Bq \approx -157.9068
\ \text{at}\
\boldsymbol{\theta}^{\circ} \approx (2.8642,\,0.7489),
\qquad
p_c \approx 0.1451 .
\end{equation}
Both inequalities are tight, the vertices
of $\mathcal{L}_\Pi$ saturating each bound span an affine subspace of
dimension $d-1$, both in the projection onto the eight even-order
coordinates ($d=8$) and in the full projection onto all correlators of
order $\le 4$ ($d=14$), so $\mathcal{I}_1$ and $\mathcal{I}_2$ are
facets of the corresponding PI local polytopes.
Figure~\ref{fig:pcrit_N6_k24} shows the critical noise threshold
$p_c(\theta_0,\theta_1)$ of $\mathcal{I}_1$ over the full range of
measurement angles. 
The optimum $\boldsymbol{\theta}^{\diamond}$ is one of
eight equivalent maxima, related by $\theta_j\to\theta_j+\pi$ and
$\boldsymbol{\theta}\to-\boldsymbol{\theta}$, and the broad maximum
surrounding it indicates the tolerance of the violation to
deviations of the measurement angles, quantified in
Sec.~\ref{supp:sec:tolerance}.

\begin{figure}[htb]
\centering
\begin{minipage}[b]{0.5\textwidth}
    \centering
    \includegraphics[width=\textwidth]{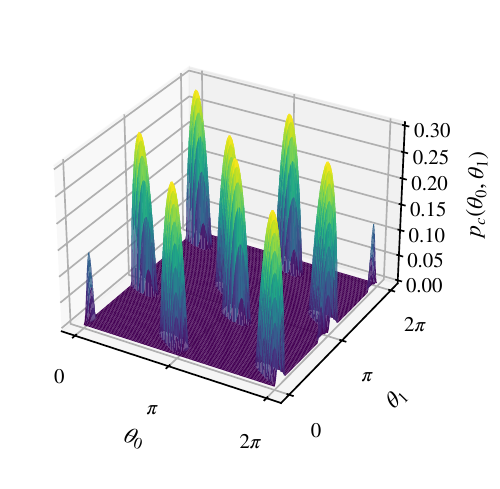}\\[-2pt]
    {\small (a)}
\end{minipage}\hfill
\begin{minipage}[b]{0.5\textwidth}
    \centering
    \includegraphics[width=\textwidth]{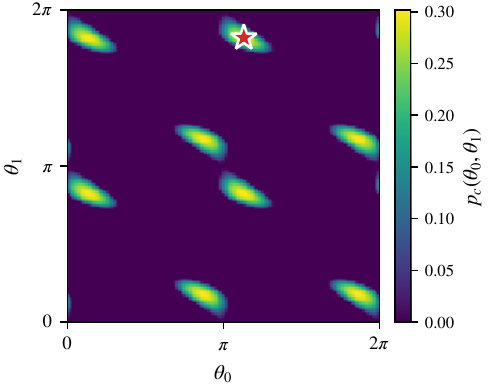}\\[-2pt]
    {\small (b)}
\end{minipage}
\caption{Critical noise threshold $p_{c}(\theta_0,\theta_1)
=\max\{0,\,1-\Bc/\avgDs{\hat{\mathcal{I}}_1(\theta_0,\theta_1)}\}$
of the Bell inequality $\mathcal{I}_1$, Eq.~\eqref{eq:BI_24body_sparse_supp}, as a function of the measurement angles.
(a) Surface plot over $(\theta_0,\theta_1)\in[0,2\pi]^{2}$. The eight
equivalent peaks reflect the symmetries
$\theta_j\to\theta_j+\pi$ (each retained correlator contains each
setting an even number of times) and
$\boldsymbol{\theta}\to-\boldsymbol{\theta}$ (invariance of
$\ket{D_6^{\,3}}$ correlations under reflection of the $xz$ plane).
(b) Heatmap. The star marks the optimum
$\boldsymbol{\theta}^{\diamond}\approx(3.5501,\,5.7302)$, where
$p_{c}\approx 0.3014$.}
\label{fig:pcrit_N6_k24}
\end{figure}

\subsection{Critical noise threshold versus fidelity: predicted hierarchy}
\label{supp:sec:pc_vs_F}

The threshold $p_c$ of \eqref{eq:pcrit_supp} is intrinsic to each
inequality. For a state prepared at fidelity $F$, corresponding to the
noise fraction $p_F=(1-F)/(1-2^{-N})$ of \eqref{eq:pexp}, the
\emph{residual} tolerance to additional white noise is
\begin{equation}\label{eq:pcF}
p_c(F)=1-\frac{\Bc}{(1-p_F)\,\Bq}=\frac{p_c-p_F}{1-p_F},
\end{equation}
which vanishes at the critical fidelity
\begin{equation}\label{eq:Fc}
F_c=1-p_c\,(1-2^{-N}).
\end{equation}
For the four inequalities tested in this work,
\begin{equation}\label{eq:Fc_values}
F_c(\mathcal{I}_{\mathrm{T}})\approx 0.953,\qquad
F_c(\mathcal{I}_{1})\approx 0.703,\qquad
F_c(\mathcal{I}_{2})\approx 0.857,\qquad
F_c(\mathcal{I}_{3})\approx 0.775 .
\end{equation}
Figure~\ref{fig:pcrit_vs_fidelity} shows $p_c(F)$ for the two six-photon
inequalities $\mathcal{I}_1$ and $\mathcal{I}_2$, whose contrasting
critical fidelities are probed experimentally with deliberately
degraded states.
Table~\ref{tab:hierarchy} compares, for each prepared
state, the white-noise prediction
$(1-\pexp)\,\Bq$ with the measured Bell values. The predicted hierarchy is reproduced: at $F=0.933$, $\mathcal{I}_1$ and $\mathcal{I}_2$ are violated but $\mathcal{I}_{\mathrm{T}}$ is not; at $F=0.797$ only $\mathcal{I}_1$ is violated; at $F=0.638$ neither is; and $\mathcal{I}_3$ is violated for $N=8$.

\begin{figure}[htb]
\centering
\includegraphics[width=0.5\textwidth]{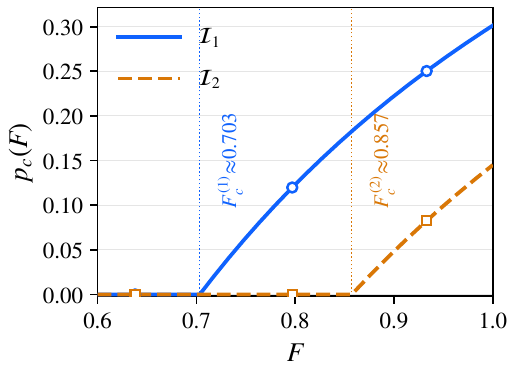}
\caption{Residual noise tolerance $p_{c}(F)$ of \eqref{eq:pcF} for the
two six-photon Bell inequalities $\mathcal{I}_1$ [Eq.~\eqref{eq:BI_24body_sparse_supp}] and $\mathcal{I}_2$ [Eq.~\eqref{eq:BI_24body_facetG_supp}]. Markers indicate the three
experimentally prepared fidelities
$F\approx 0.933,\,0.797,\,0.638$; dotted vertical lines mark the critical
fidelities $F_c(\mathcal{I}_1)\approx 0.703$ and
$F_c(\mathcal{I}_2)\approx 0.857$ of \eqref{eq:Fc}.}
\label{fig:pcrit_vs_fidelity}
\end{figure}

\begin{table}[tb]
\caption{Predicted versus measured Bell values for all prepared states.
The prediction is the white-noise-model value $(1-\pexp)\Bq$ with
$\pexp$ from \eqref{eq:pexp}; ``viol.'' indicates whether the classical
bound $\Bc$ is exceeded. Measured values and their $1\sigma$
uncertainties are taken from the data tables of
Section~\ref{supp:sec:exp_details}; the last column gives the measured
violation significance $(\Bc-\mathcal{I}_{\rm meas})/\sigma$.}
\label{tab:hierarchy}
\begin{ruledtabular}
\begin{tabular}{llcccccc}
State & Inequality & $\Bc$ & $\Bq$ & $(1-\pexp)\Bq$ & viol.\ predicted & measured & significance \\
\hline
$F=0.933(4)$, $N{=}6$ & $\mathcal{I}_{\mathrm{T}}$ & $-120$ & $-126.04$ & $-117.4$ & no  & $-119.74(1.66)$ & --- \\
                        & $\mathcal{I}_{1}$ & $-18$  & $-25.764$ & $-24.01$ & yes & $-24.012(186)$ & $32\sigma$ \\
                        & $\mathcal{I}_{2}$ & $-135$ & $-157.91$ & $-147.1$ & yes & $-146.58(1.68)$ & $6.9\sigma$ \\
$F=0.797(5)$, $N{=}6$ & $\mathcal{I}_{1}$ & $-18$  & $-25.764$ & $-20.45$ & yes & $-19.895(164)$ & $11.6\sigma$ \\
                        & $\mathcal{I}_{2}$ & $-135$ & $-157.91$ & $-125.3$ & no  & $-125.15(1.72)$ & --- \\
$F=0.638(5)$, $N{=}6$ & $\mathcal{I}_{1}$ & $-18$  & $-25.764$ & $-16.28$ & no  & $-16.095(177)$ & --- \\
                        & $\mathcal{I}_{2}$ & $-135$ & $-157.91$ & $-99.8$  & no  & $-104.94(1.73)$ & --- \\
$F=0.916(11)$, $N{=}8$ & $\mathcal{I}_{3}$ & $-48$ & $-62.014$ & $-56.79$ & yes & $-55.905(889)$ & $8.9\sigma$ \\
\end{tabular}
\end{ruledtabular}
\end{table}

\subsection{Tolerance to measurement-angle deviations}
\label{supp:sec:tolerance}

Imperfectly calibrated analyzers displace the measurement angles from
$\boldsymbol{\theta}^\diamond$. Within the white-noise model at the operating
fidelity $F=0.933$ (computations use the unrounded value $0.9328$), the
predicted Bell value at angles $(\theta_0,\theta_1)$ is
\begin{equation}\label{eq:tolerance_model}
\mathcal{I}_1^{\rm model}(\theta_0,\theta_1)
=(1-\pexp)\avgDs{\hat{\mathcal{I}}_1(\theta_0,\theta_1)},
\qquad \pexp\approx 0.068,
\end{equation}
with the value $-24.012$ at the optimum. The violation region
$\{\mathcal{I}_1^{\rm model}<-18\}$ around the optimum is wide:
holding one angle at its optimal value, the violation persists for
\begin{equation}\label{eq:tolerance_1d}
\theta_0\in[3.158,\,3.942]\ (\pm 22.5^{\circ}),
\qquad
\theta_1\in[5.516,\,5.992]\ (-12.3^{\circ}/{+}15.0^{\circ}),
\end{equation}
and the largest axis-aligned square centered at the optimum that fits
inside the region has half-width $\approx 9.0^{\circ}$ per angle.
For comparison, the corresponding region of $\mathcal{I}_2$ at the
same fidelity is markedly smaller along $\theta_0$ ($\pm 11.3^{\circ}$), consistent with its smaller relative violation. 
These angular tolerances are estimates within the white-noise model. Because the experimental noise is not strictly white, Fig.~3(a) of the main text instead shows the theoretical violation region for the ideal six-photon Dicke state ($F=1$).

Experimentally, in addition to the optimal setting $P_0$, six settings
$P_1$--$P_6$ on or near the boundary of the violation region are
measured ([Fig.~3(a) of the main text].
Table~\ref{tab:tolerance_points} lists their coordinates together with the measured values. All six deviated settings still
violate the classical bound, by $6.8\sigma$ to $14\sigma$, directly
demonstrating that the violation does not rely on a finely tuned choice
of measurement directions. 
\begin{table}[tb]
\caption{Measurement-angle tolerance test of $\mathcal{I}_1$ at
$F=0.933(4)$. $P_0$ is the optimal setting
\eqref{eq:facetC_values}; $P_1$--$P_6$ lie on or near the boundary of
the predicted violation region. The violation factor is
$R=\mathcal{I}_1/(-18)$; uncertainties in parentheses denote one
standard deviation. Measured values are from
Tables~\ref{tab:PI_GAO_0}, \ref{tab:PI_GAO_P1},
\ref{tab:PI_GAO_P2}, \ref{tab:PI_GAO_P3},
\ref{tab:PI_GAO_P4}, \ref{tab:PI_GAO_P5}, and
\ref{tab:PI_GAO_P6}.}
\label{tab:tolerance_points}
\begin{ruledtabular}
\begin{tabular}{ccccc}
Setting & $(\theta_0,\theta_1)$ (rad) &
Measured $\mathcal{I}_1$ & $R$ & Violation \\
\hline
$P_0$ & $(3.550,\;5.730)$ & $-24.012(186)$ & $1.334(10)$ & $32\sigma$ \\
$P_1$ & $(3.218,\;5.730)$ & $-21.017(262)$ & $1.168(15)$ & $12\sigma$ \\
$P_2$ & $(3.829,\;5.730)$ & $-19.798(155)$ & $1.100(9)$  & $12\sigma$ \\
$P_3$ & $(3.550,\;5.547)$ & $-20.048(262)$ & $1.114(15)$ & $7.8\sigma$ \\
$P_4$ & $(3.550,\;5.918)$ & $-19.772(139)$ & $1.098(8)$  & $13\sigma$ \\
$P_5$ & $(3.197,\;5.976)$ & $-20.377(170)$ & $1.132(9)$  & $14\sigma$ \\
$P_6$ & $(3.910,\;5.547)$ & $-19.332(195)$ & $1.074(11)$ & $6.8\sigma$ \\
\end{tabular}
\end{ruledtabular}
\end{table}

\subsection{Extension to $N=8$}
\label{supp:sec:fourbody_N8}

Now we extend the construction of Sec.~\ref{supp:sec:fourbody_N6} 
to the $(8,2,2)$ scenario with target state $\ket{D_8^{\,4}}$:
Theorem~\ref{thm:selection} and Corollary~\ref{cor:even} hold for all even $N$, so the ansatz is
again \eqref{eq:BI_24body_supp}. Maximizing $p_c$ over all eight
coefficients yields
\begin{equation}\label{eq:BI_24body_N8_optimal_supp}
\begin{split}
I(\boldsymbol{\alpha}^{\star})
={}& \frac{5580}{2}\,\mathcal{S}_{00}
   - 2260\,\mathcal{S}_{01}
   + \frac{10060}{2}\,\mathcal{S}_{11}
 - \frac{715}{24}\,\mathcal{S}_{0000}
   + \frac{907}{6}\,\mathcal{S}_{0001} \\
 &+ \frac{2075}{4}\,\mathcal{S}_{0011}
 - \frac{1915}{6}\,\mathcal{S}_{0111}
   - \frac{1155}{24}\,\mathcal{S}_{1111}
   \;\geq\; -65480,
\end{split}
\end{equation}
with
\begin{equation}
\Bc(\boldsymbol{\alpha}^{\star}) = -65480,
\quad
\Bq(\boldsymbol{\alpha}^{\star}) \approx -105964.3
\ \text{at}\
\boldsymbol{\theta}^{\star}\approx (1.7799,\,3.2237),
\quad
p_{c}(\boldsymbol{\alpha}^{\star})\approx 0.3821 .
\end{equation}
As for $N=6$, we select for the experiment the sparse member
$\boldsymbol{\alpha}^{\diamond}= (3,\,0,\,9;\;0,\,0,\,-1,\,0,\,6)$
[$\mathcal{I}_3$ of the main text],
\begin{equation}\label{eq:BI_24body_N8_sparse_supp}
\mathcal{I}_3
= \frac{3}{2}\,\mathcal{S}_{00}
+ \frac{9}{2}\,\mathcal{S}_{11}
- \frac{1}{4}\,\mathcal{S}_{0011}
+ \frac{1}{4}\,\mathcal{S}_{1111}
\;\geq\; -48,
\end{equation}
with
\begin{equation}\label{eq:sparse8_values}
\Bc = -48,
\qquad
\Bq \approx -62.0144
\ \text{at}\
\boldsymbol{\theta}^{\diamond}\approx(0.3690,\,5.8444),
\qquad
p_{c}\approx 0.2260,
\qquad
F_c\approx 0.775 .
\end{equation}

Like $\mathcal{I}_1$ and $\mathcal{I}_2$, the inequality
$\mathcal{I}_3$ is a facet of the projected PI polytope, involves only the four symmetric
correlators accessible from the $\binom{8}{2}=28$ two-versus-six
setting configurations, and comfortably certifies the prepared
eight-photon state, $p_c\approx 0.226\gg\pexp\approx 0.084$. The
predicted value at the prepared fidelity, $-56.8$, agrees with the
measured $\mathcal{I}_3=-55.905(889)$ within $1.0\sigma$
(Table~\ref{tab:hierarchy}). Figure~\ref{fig:pcrit_N8_k24} shows the
$p_c$ landscape of $\mathcal{I}_3$.

\begin{figure}[tb]
\centering
\begin{minipage}[b]{0.5\textwidth}
    \centering
    \includegraphics[width=\textwidth]{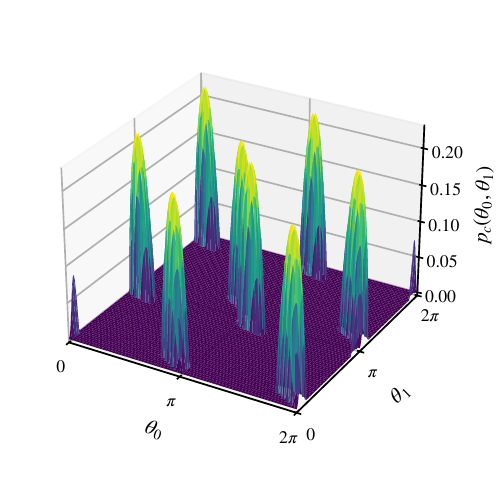}\\[-2pt]
    {\small (a)}
\end{minipage}\hfill
\begin{minipage}[b]{0.5\textwidth}
    \centering
    \includegraphics[width=\textwidth]{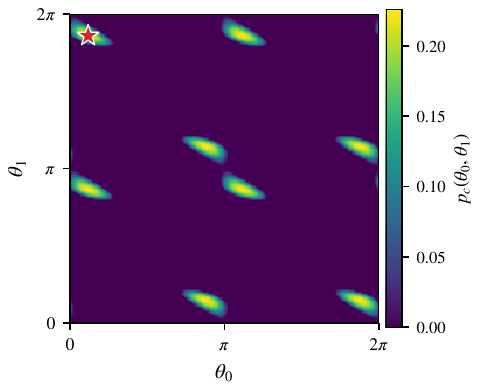}\\[-2pt]
    {\small (b)}
\end{minipage}
\caption{Critical noise threshold $p_{c}(\theta_0,\theta_1)$ of the eight-photon Bell inequality $\mathcal{I}_3$ in
\eqref{eq:BI_24body_N8_sparse_supp}, evaluated on
$\ket{D_{8}^{\,4}}$. (a)~Surface plot over
$(\theta_0,\theta_1)\in[0,2\pi]^{2}$. (b)~Heatmap. The star marks the
optimum $\boldsymbol{\theta}^{\diamond}\approx (0.3690,\,5.8444)$,
where $p_{c}\approx 0.2260$.}
\label{fig:pcrit_N8_k24}
\end{figure}

For $N=4$ and $N=10$, the optimal thresholds are $p_c\approx0.53$ and $0.37$, respectively, which are shown together with the two-body optima ($m=2$ and $m=3$) in the End Matter of the main text. The four-body threshold thus decreases only slowly with $N$ ($0.53$, $0.46$, $0.38$, $0.37$ for $N=4,6,8,10$), whereas the two-body thresholds remain below $0.1$ for $N\ge6$: the fixed-order strategy is not limited by the witness but by the preparation of larger Dicke states.

\section{Noise model for Dicke states}
\label{app:dicke_noise_model}
\label{app:eight_photon_noise_model}

The six- and eight-photon Dicke states are prepared by postselecting multipair emission from a spontaneous parametric down-conversion (SPDC) source. Higher-order SPDC emission is therefore an intrinsic and important source of noise. Let $\mu$ denote the probability that a single pump pulse generates one $HV$ photon pair. In the low-gain regime the probability of generating $n$ pairs scales as $\mu^n$, and the SPDC output can be written as 

\begin{equation}
\ket{\Psi}_{\mathrm{SPDC}}=\sqrt{1-\mu}\,\sum_{n=0}^\infty \mu^{n/2} \ket{n_H,n_V}
\label{eq:spdc_emission}
\end{equation}
where $n_H$ and $n_V$ denote the number of photons of polarization $H$ and $V$.
The parameter $\eta$ is the effective probability that an emitted photon is detected. In this noise model, we retain the target emission order and the leading higher-order contribution.

Throughout this section, we analyze polarization-resolved coincidence statistics in the $H/V$ basis, which defines the local $Z$ basis. The resulting model describes excitation-sector populations in this basis, rather than a complete reconstruction of the postselected state. For an $N$-photon Dicke state, the target contribution originates from $N/2$-pair emission, whereas the leading higher-order process produces $N/2+1$ pairs. Depending on $\eta$, part of the additional photons may be lost before detection. Spatial overlap can also route several photons into the same output channel, where a threshold detector registers only one click. Higher-order events can therefore contribute to the same $N$-fold coincidence sample and populate the excitation sectors $N/2-1$, $N/2$, and $N/2+1$. We describe the effective postselected Dicke state as
\begin{equation}
\rho_N^{\mathrm{eff}}=\sum_{j=N/2-1}^{N/2+1}M_j^{(N)}\ketbra{D_N^{\,j}}{D_N^{\,j}},\qquad
\sum_{j=N/2-1}^{N/2+1}M_j^{(N)}=1,
\label{eq:effective_dicke_noise_state}
\end{equation}
where $M_j^{(N)}$ is estimated from the fraction of $Z^{\otimes N}$ outcomes. The expressions below assume symmetric spatial routing and polarization-independent efficiency. Same-polarization multiphoton occupation produces one threshold-detector click, whereas $HV$ double clicks in one spatial mode are rejected.

\subsection{Six-photon Dicke state}

The target state $\ket{D_6^{\,3}}$ originates from three-pair emission, with event weight
\begin{equation}
W_{3\mathrm{p}}^{(6)}=\mu^3\binom{6}{3}(3!)^2\frac{\eta^6}{4^2 8^4}.
\label{eq:six_target_weight}
\end{equation}

The leading background arises from four-pair emission. We label each accepted higher-order contribution by $R_{a,b}$, defined as its event weight normalized to the target three-pair weight $W_{3\mathrm{p}}^{(6)}$, where $a\in\{6,7,8\}$ is the number of photons reaching the detectors after loss and $b\in\{2,3,4\}$ is the number of output modes registered as $\ket{1}$. All these contributions produce sixfold coincidences. When $a>6$, multiple photons of the same polarization occupy a common output channel and are registered as a single click by the threshold detector.

Accounting for the different spatial-mode overlap configurations associated with each detected photon number, the resulting relative weights are
\begin{align}
R_{6,2}=R_{6,4}&=6\mu(1-\eta)^2, & R_{6,3}&=16\mu(1-\eta)^2,\nonumber\\
R_{7,2}=R_{7,4}&=2\mu\eta(1-\eta), & R_{7,3}&=8\mu\eta(1-\eta),\nonumber\\
R_{8,2}=R_{8,4}&=\frac{7}{36}\mu\eta^2, & R_{8,3}&=\mu\eta^2.
\label{eq:six_relative_click_weights}
\end{align}

Before normalization, the relative populations $M_2$, $M_3$, and $M_4$ of $D_6^{(2)}$, $D_6^{(3)}$, and $D_6^{(4)}$ satisfy
\begin{equation}
\left(M_2,M_3,M_4\right)\propto
\left(\sum_{a=6}^{8}R_{a,2}\,;\;1+\sum_{a=6}^{8}R_{a,3}\,;\;\sum_{a=6}^{8}R_{a,4}\right),
\label{eq:six_click_probabilities}
\end{equation}
where the unit term in the central component represents the target three-pair contribution.

Experimentally, for $\mu=0.039$ and $\eta=0.726$, substituting these values into the relative weights and normalizing the three click classes gives
\begin{equation}
\left(M_2,M_3,M_4\right)_{\mathrm{model}}=(0.031,\,0.938,\,0.031).
\label{eq:six_model_prediction}
\end{equation}
The measured distribution is
\begin{equation}
\left(M_2,M_3,M_4\right)_{\mathrm{exp}}=(0.035\pm0.003,\,0.937\pm0.004,\,0.028\pm0.003).
\label{eq:six_experimental_distribution}
\end{equation}

\subsection{Eight-photon Dicke state}

The target state $\ket{D_8^{(4)}}$ originates from four-pair emission, with event weight
\begin{equation}
W_{4\mathrm{p}}^{(8)}=\mu^4\binom{8}{4}(4!)^2\frac{\eta^8}{8^8}.
\label{eq:four_pair_postselection_weight}
\end{equation}

The leading background arises from five-pair emission, which is analyzed following the same procedure as for the six-photon state. We label each accepted higher-order contribution by $R_{a,b}$, defined as its event weight normalized to the target four-pair weight $W_{4\mathrm{p}}^{(8)}$, where $a\in\{8,9,10\}$ is the number of photons reaching the detectors after loss and $b\in\{3,4,5\}$ is the number of output modes registered as $\ket{1}$. Accounting for the different spatial-mode overlap configurations associated with each detected photon number, the resulting relative weights are
\begin{align}
R_{8,3}=R_{8,5}&=10\mu(1-\eta)^2, & R_{8,4}&=25\mu(1-\eta)^2,\nonumber\\
R_{9,3}=R_{9,5}&=\frac{15}{4}\mu\eta(1-\eta), & R_{9,4}&=\frac{25}{2}\mu\eta(1-\eta),\nonumber\\
R_{10,3}=R_{10,5}&=\frac{25}{64}\mu\eta^2, & R_{10,4}&=\frac{25}{16}\mu\eta^2.
\label{eq:eight_relative_click_weights}
\end{align}

Before normalization, the relative populations $M_3$, $M_4$, and $M_5$ of $D_8^{(3)}$, $D_8^{(4)}$, and $D_8^{(5)}$, respectively, satisfy
\begin{equation}
\left(M_3,M_4,M_5\right)\propto
\left(\sum_{a=8}^{10}R_{a,3}\,;\;1+\sum_{a=8}^{10}R_{a,4}\,;\;\sum_{a=8}^{10}R_{a,5}\right).
\label{eq:eight_click_probabilities}
\end{equation}

For $\mu=0.0390$ and $\eta=0.726$, substituting these values into the relative weights and normalizing the three click classes gives
\begin{equation}
\left(M_3,M_4,M_5\right)_{\mathrm{model}}=(0.05,\,0.900,\,0.05).
\label{eq:eight_model_prediction}
\end{equation}
The measured distribution is
\begin{equation}
\left(M_3,M_4,M_5\right)_{\mathrm{exp}}=
(0.047\pm0.008,\,0.910\pm0.011,\,0.043\pm0.008),
\label{eq:eight_experimental_distribution}
\end{equation}
where the uncertainties denote one standard errors.

The theoretical predictions agree well with the measured $Z$-basis click-class distributions. This agreement supports the validity of the noise model and provides an independent consistency check on the reliable preparation of our high-fidelity multiphoton Dicke states.

\section{Measurement Protocols for Fidelity and PI Bell Inequalities}
\label{supp:sec:protocols}

\subsection{Projector Decomposition for Dicke-State Fidelity}

The fidelity between an experimentally prepared state $\rho$ and an ideal pure state $\ket{\psi}$ is defined as
\begin{equation}
F_{\ket{\psi}}(\rho):=\operatorname{Tr}\!\left(\ketbra{\psi}{\psi}\rho\right).
\end{equation}

For the postselected six- and eight-photon Dicke-state experiments, full quantum-state tomography would require a prohibitively large number of measurement settings. We therefore determine the state-preparation fidelity by measuring the projector onto the target Dicke state. Because symmetric Dicke states are invariant under permutations of the qubits, their projectors are also permutationally invariant operators. According to Observation~1 in Ref.~\cite{toth_practical_2009}, any permutationally invariant $N$-qubit operator $A$ can be decomposed as
\begin{equation}
A=\sum_n c_n a_n^{\otimes N},
\label{eq:PI_decomposition}
\end{equation}
where $c_n$ are real coefficients and $a_n$ are single-qubit operators. In the polarization-qubit representation,
\begin{equation}
a_n=\alpha_n I+x_nX+y_nY+z_nZ.
\end{equation}

The fidelity can therefore be expressed in terms of the permutation-averaged correlation functions
\begin{equation}
C_k^{(N)}(A)=\frac{1}{\binom{N}{k}}\sum_{1\leq j_1<\cdots<j_k\leq N}
\left\langle A_{j_1}\cdots A_{j_k}\right\rangle,
\qquad k=1,\ldots,N,
\label{eq:Ci_general}
\end{equation}
where $A_j$ denotes the single-qubit observable $A$ acting on qubit $j$. Experimentally, the same measurement basis is applied to all qubits, and correlations of different orders are extracted from the corresponding $N$-photon coincidence data.
Here, we introduce the permutation-summed operator
\begin{equation}
[A]_k^{(N)}=\sum_{1\leq j_1<\cdots<j_k\leq N}A_{j_1}\cdots A_{j_k},
\qquad
C_k^{(N)}(A)=\frac{\operatorname{Tr}\!\left(\rho[A]_k^{(N)}\right)}{\binom{N}{k}}.
\label{eq:C_operator_relation}
\end{equation}

For the six-photon Dicke state, the exact projector decomposition associated with the 21 collective measurement settings is~\cite{wieczorek_experimental_2009}
\begingroup
\small
\setlength{\jot}{2pt}
\begin{equation}
\begin{aligned}
\ketbra{D_6^{(3)}}{D_6^{(3)}}={}&
\frac{1}{64}I^{\otimes6}
+\frac{3}{320}\!\left([X]_2+[X]_4+[Y]_2+[Y]_4\right)
+\frac{1}{320}[Z]_2+\frac{3}{320}[Z]_4 \\
&+\frac{1}{320}[X\!\pm\!Y]_2+\frac{1}{160}[X\!\pm\!Y]_4+\frac{1}{40}[X\!\pm\!Y]_6 \\
&-\frac{1}{320}\!\left([X\!\pm\!Z]_2+[Y\!\pm\!Z]_2\right)
-\frac{1}{160}\!\left([X\!\pm\!Z]_4+[Y\!\pm\!Z]_4\right) \\
&+\frac{1}{80}\!\left([X\!\pm\!Z]_6+[Y\!\pm\!Z]_6\right)
-\frac{27}{1280}[X\!\pm\!Y\!\pm\!Z]_6
+\frac{1}{20}\!\left([X]_6+[Y]_6-2[Z]_6\right) \\
&+\frac{1}{20}\sum_{k=1,2,4,5}(-1)^k
\left\{
\left[\cos\frac{k\pi}{6}X+\sin\frac{k\pi}{6}Z\right]_6
+\left[\cos\frac{k\pi}{6}Y+\sin\frac{k\pi}{6}Z\right]_6
\right\}.
\end{aligned}
\label{eq:D6_projector_decomposition}
\end{equation}
\endgroup

Taking the expectation value of Eq.~\eqref{eq:D6_projector_decomposition} and using Eq.~\eqref{eq:C_operator_relation} gives the fidelity expression used in the experiment:
\begingroup
\small
\setlength{\jot}{2pt}
\begin{equation}
\begin{aligned}
F_{\ket{D_6^{(3)}}}(\rho)
={}&\operatorname{Tr}\!\left[\rho\ketbra{D_6^{(3)}}{D_6^{(3)}}\right] \\
={}&\frac{1}{64}
+\frac{9}{64}\!\left[C_2(X)+C_4(X)+C_2(Y)+C_4(Y)\right]
+\frac{3}{64}C_2(Z)+\frac{9}{64}C_4(Z) \\
&+\frac{3}{64}C_2(X\!\pm\!Y)+\frac{3}{32}C_4(X\!\pm\!Y)+\frac{1}{40}C_6(X\!\pm\!Y) \\
&-\frac{3}{64}\!\left[C_2(X\!\pm\!Z)+C_2(Y\!\pm\!Z)\right]
-\frac{3}{32}\!\left[C_4(X\!\pm\!Z)+C_4(Y\!\pm\!Z)\right] \\
&+\frac{1}{80}\!\left[C_6(X\!\pm\!Z)+C_6(Y\!\pm\!Z)\right]
-\frac{27}{1280}C_6(X\!\pm\!Y\!\pm\!Z)
+\frac{1}{20}\!\left[C_6(X)+C_6(Y)-2C_6(Z)\right] \\
&+\frac{1}{20}\sum_{k=1,2,4,5}(-1)^k
\left[
C_6\!\left(\cos\frac{k\pi}{6}X+\sin\frac{k\pi}{6}Z\right)
+C_6\!\left(\cos\frac{k\pi}{6}Y+\sin\frac{k\pi}{6}Z\right)
\right].
\end{aligned}
\label{eq:D6_c}
\end{equation}
\endgroup

For the eight-photon Dicke state, the corresponding projector has the following exact decomposition into 29 collective measurement settings:
\begingroup
\small
\setlength{\jot}{2pt}
\begin{equation}
\begin{aligned}
\ketbra{D_8^{(4)}}{D_8^{(4)}}={}&
\frac{1}{256}I^{\otimes8}
+\frac{1}{448}\!\left([X]_2+[Y]_2\right)-\frac{1}{1792}[Z]_2
+\frac{1}{560}\!\left([X]_4+[Y]_4\right)+\frac{11}{8960}[Z]_4 \\
&+\frac{9}{3584}\!\left([X]_6+[Y]_6\right)-\frac{19}{1792}[Z]_6
+\frac{93}{3584}\!\left([X]_8+[Y]_8\right)+\frac{17}{224}[Z]_8 \\
&-\frac{1}{1120}\!\left([X\!\pm\!Z]_4+[Y\!\pm\!Z]_4\right)
+\frac{3}{2240}[X\!\pm\!Y]_4
-\frac{11}{3360}\!\left([X\!\pm\!Z]_6+[Y\!\pm\!Z]_6\right)
+\frac{3}{1120}[X\!\pm\!Y]_6 \\
&+\frac{19}{360}\!\left([X\!\pm\!Z]_8+[Y\!\pm\!Z]_8\right)
+\frac{139}{20160}[X\!\pm\!Y]_8
+\frac{125}{21504}\!\left([X\!\pm\!2Z]_6+[Y\!\pm\!2Z]_6\right) \\
&-\frac{2125}{32256}\!\left([X\!\pm\!2Z]_8+[Y\!\pm\!2Z]_8\right)
-\frac{2375}{64512}\!\left([2X\!\pm\!Z]_8+[2Y\!\pm\!Z]_8\right) \\
&+\frac{125}{64512}\!\left([2X\!\pm\!Y]_8+[X\!\pm\!2Y]_8\right)
-\frac{27}{17920}[X\!\pm\!Y\!\pm\!Z]_6
-\frac{297}{17920}[X\!\pm\!Y\!\pm\!Z]_8
+\frac{27}{896}[X\!\pm\!Y\!\pm\!2Z]_8.
\end{aligned}
\label{eq:D8_projector_decomposition}
\end{equation}
\endgroup

Taking the expectation value of Eq.~\eqref{eq:D8_projector_decomposition} and using Eq.~\eqref{eq:C_operator_relation} gives
\begingroup
\small
\setlength{\jot}{2pt}
\begin{equation}
\begin{aligned}
F_{\ket{D_8^{(4)}}}(\rho)
={}&\operatorname{Tr}\!\left[\rho\ketbra{D_8^{(4)}}{D_8^{(4)}}\right] \\
={}&\frac{1}{256}
+\frac{1}{16}\!\left[C_2(X)+C_2(Y)\right]-\frac{1}{64}C_2(Z)
+\frac{1}{8}\!\left[C_4(X)+C_4(Y)\right]+\frac{11}{128}C_4(Z) \\
&+\frac{9}{128}\!\left[C_6(X)+C_6(Y)\right]-\frac{19}{64}C_6(Z)
+\frac{93}{3584}\!\left[C_8(X)+C_8(Y)\right]+\frac{17}{224}C_8(Z) \\
&-\frac{1}{16}\!\left[C_4(X\!\pm\!Z)+C_4(Y\!\pm\!Z)\right]
-\frac{11}{120}\!\left[C_6(X\!\pm\!Z)+C_6(Y\!\pm\!Z)\right]
+\frac{19}{360}\!\left[C_8(X\!\pm\!Z)+C_8(Y\!\pm\!Z)\right] \\
&+\frac{3}{32}C_4(X\!\pm\!Y)+\frac{3}{40}C_6(X\!\pm\!Y)
+\frac{139}{20160}C_8(X\!\pm\!Y)
+\frac{125}{768}\!\left[C_6(X\!\pm\!2Z)+C_6(Y\!\pm\!2Z)\right] \\
&-\frac{2125}{32256}\!\left[C_8(X\!\pm\!2Z)+C_8(Y\!\pm\!2Z)\right]
-\frac{2375}{64512}\!\left[C_8(2X\!\pm\!Z)+C_8(2Y\!\pm\!Z)\right] \\
&+\frac{125}{64512}\!\left[C_8(2X\!\pm\!Y)+C_8(X\!\pm\!2Y)\right]
-\frac{27}{640}C_6(X\!\pm\!Y\!\pm\!Z)
-\frac{297}{17920}C_8(X\!\pm\!Y\!\pm\!Z)
+\frac{27}{896}C_8(X\!\pm\!Y\!\pm\!2Z).
\end{aligned}
\label{eq:D8_c}
\end{equation}
\endgroup

Here, $I$ is the single-qubit identity operator, while $X$, $Y$, and $Z$ are the Pauli operators in the polarization-qubit representation. The notation $[A]_k^{(N)}$ denotes the unnormalized sum of all $\binom{N}{k}$ distinct $k$-qubit products of $A$, whereas $C_k^{(N)}(A)$ denotes their normalized expectation value. The superscript $(N)$ is omitted when the particle number is clear from the context.
Every linear combination $A=aX+bY+cZ$ denotes the normalized measurement observable
\begin{equation}
\widehat{A}=\frac{aX+bY+cZ}{\sqrt{a^2+b^2+c^2}},
\end{equation}
where the normalization symbol is suppressed in the compact direction labels. For example, $X+2Z$ denotes $(X+2Z)/\sqrt{5}$.

The $\pm$ notation denotes summation over all displayed sign choices:
\begingroup
\small
\setlength{\jot}{2pt}
\begin{equation}
\begin{aligned}
C_k(A\!\pm\!B)&=C_k(A+B)+C_k(A-B), \\
C_k(A\!\pm\!B\!\pm\!C)&=C_k(A+B+C)+C_k(A+B-C)+C_k(A-B+C)+C_k(A-B-C), \\
[A\!\pm\!B]_k&=[A+B]_k+[A-B]_k, \\
[A\!\pm\!B\!\pm\!C]_k&=[A+B+C]_k+[A+B-C]_k+[A-B+C]_k+[A-B-C]_k .
\end{aligned}
\label{eq:pm_notation}
\end{equation}
\endgroup

\subsection{Measurement Settings for the PI Bell Inequalities}
\label{subsec:pi_bell_settings}

We experimentally test four permutationally invariant (PI) Bell
inequalities with six- and eight-photon Dicke state, which are defined in Eq.\eqref{eq:Tura_supp}, Eq.\eqref{eq:BI_24body_sparse_supp}, Eq.\eqref{eq:BI_24body_facetG_supp} and Eq.\eqref{eq:BI_24body_N8_sparse_supp}.
Each photon $i$ is measured with one of the two local observables $A_0^{(i)}$, $A_1^{(i)}$ of Eq.~\eqref{eq:sym_meas_supp}, with outcome $a_i=\pm1$. The measurement directions lie in the $xz$ plane of the Bloch sphere; for each inequality we use its optimal angles $(\theta_0,\theta_1)$ of Secs.~\ref{supp:sec:twobody_insufficient} and \ref{supp:sec:fourbody}, collected in Table~\ref{tab:PI_measurement_angles}.

\begin{table}[htb]
\centering
\caption{Local measurement angles used for the four PI Bell inequalities.}
\label{tab:PI_measurement_angles}

\renewcommand{\arraystretch}{1.4}
\setlength{\tabcolsep}{12pt}

\begin{tabular}{|c|c|c|}
\hline
Inequality & $\theta_{0}$ (rad) & $\theta_{1}$ (rad) \\
\hline
$\mathcal{I}_{\mathrm{T}}$ & $3.264$ & $2.165$ \\
\hline
$\mathcal{I}_{1}$ & $3.550$ & $5.730$ \\
\hline
$\mathcal{I}_{2}$ & $2.864$ & $0.749$ \\
\hline
$\mathcal{I}_{3}$ & $0.369$ & $5.844$ \\
\hline
\end{tabular}

\end{table}

All four inequalities involve only the PI correlators $\mathcal{S}_{00}$, $\mathcal{S}_{11}$, $\mathcal{S}_{0011}$, $\mathcal{S}_{1111}$ (and $\mathcal{S}_{01}$ for $\mathcal{I}_{\mathrm{T}}$), which are all obtained from ``two-versus-rest'' configurations: in measurement run $r$, two photons are measured with $A_0$ and the remaining $N-2$ photons with $A_1$. The corresponding measured assignments are listed in Tables~\ref{tab:N6_PI_runs} and~\ref{tab:N8_PI_runs}. We denote the two photon sets by
\begin{equation}
    B_{0}^{(r)}
    =
    \left\{
        i:\text{photon }i\text{ is measured with }A_{0}
    \right\},
    \qquad
    B_{1}^{(r)}
    =
    \left\{
        i:\text{photon }i\text{ is measured with }A_{1}
    \right\},
    \label{eq:measurement_subsets}
\end{equation}
and by $\langle\,\cdot\,\rangle_r$ the average over the normalized $N$-fold coincidence outcomes of run $r$. From the coincidence data of run $r$ we extract the \emph{run marginals}, defined as the arithmetic means of all correlators of a given type available in that run,
\begin{equation}
\begin{aligned}
    m_{00}^{(r)}
    &= \left\langle a_i a_j \right\rangle_{r},
    && \{i,j\}=B_{0}^{(r)},
    \\[2pt]
    m_{11}^{(r)}
    &= \binom{N-2}{2}^{-1}\sum_{\{i,j\}\subset B_{1}^{(r)}}\left\langle a_i a_j \right\rangle_{r},
    \\[2pt]
    m_{01}^{(r)}
    &= \frac{1}{2(N-2)}\sum_{i\in B_{0}^{(r)},\,j\in B_{1}^{(r)}}\left\langle a_i a_j \right\rangle_{r},
    \\[2pt]
    m_{0011}^{(r)}
    &= \binom{N-2}{2}^{-1}\sum_{\{k,\ell\}\subset B_{1}^{(r)}}\left\langle a_i a_j a_k a_{\ell}\right\rangle_{r},
    && \{i,j\}=B_{0}^{(r)},
    \\[2pt]
    m_{1111}^{(r)}
    &= \binom{N-2}{4}^{-1}\sum_{\{i,j,k,\ell\}\subset B_{1}^{(r)}}\left\langle a_i a_j a_k a_{\ell}\right\rangle_{r}.
\end{aligned}
\label{eq:run_marginal_correlators}
\end{equation}
The complete measurement consists of $R_N=\binom{N}{2}$ runs, one for each choice of the pair $B_0^{(r)}$, and the run marginals are averaged with equal weight,
\begin{equation}
    \overline{m}_{\alpha}
    =
    \frac{1}{R_N}
    \sum_{r=1}^{R_N}m_{\alpha}^{(r)},
    \qquad
    \alpha\in\{00,01,11,0011,1111\},
    \label{eq:run_averaged_m}
\end{equation}
so that all permutation-related configurations contribute equally and the protocol itself is permutation invariant. Every unordered pair of photons is the $A_0$ pair of exactly one run, and every $2$- or $4$-subset of $B_1^{(r)}$ occurs equally often over the $R_N$ runs, so the run-averaged marginals coincide with the averages of the corresponding correlators over all photon subsets. Comparing with the definition \eqref{eq:Scorr_supp} of the PI correlators as sums over ordered tuples of distinct parties,
\begin{equation}
\begin{aligned}
    \mathcal{S}_{00}
        &= N(N-1)\overline{m}_{00}, \\
    \mathcal{S}_{01}
        &= N(N-1)\overline{m}_{01}, \\
    \mathcal{S}_{11}
        &= N(N-1)\overline{m}_{11}, \\
    \mathcal{S}_{0011}
        &= N(N-1)(N-2)(N-3)\overline{m}_{0011}, \\
    \mathcal{S}_{1111}
        &= N(N-1)(N-2)(N-3)\overline{m}_{1111}.
\end{aligned}
\label{eq:S_m_relation}
\end{equation}
The Bell value $I$ of each inequality follows by inserting Eq.~\eqref{eq:S_m_relation} into its definition, and we quote the violation factor $R=I/\Bc$, so that $R>1$ signals a violation.

For each configuration of Eq.~\eqref{eq:run_marginal_correlators}, an accepted $N$-fold coincidence provides a complete outcome string for all photons. By marginalizing the corresponding joint outcome distribution over photons not involved in a given correlator, we extract all admissible two- and four-body marginal correlators from that measurement configuration. Correlators related by photon permutations are averaged with equal weight within each run and subsequently over all $R_N=\binom{N}{2}$ runs, yielding the PI correlators and, through the relations above, the corresponding Bell value. The measurements contain approximately $4000$ sixfold events per configuration for $\mathcal{I}_1$ and $680$ eightfold events per configuration for $\mathcal{I}_3$, resulting in $\mathcal{I}_1=-24.012\pm0.186$ and $\mathcal{I}_3=-55.905\pm0.889$, respectively. Because multiple marginal correlators are derived from the same coincidence records, they are statistically correlated and are not treated as independent samples. These correlations are retained by the Poisson parametric bootstrap described in Sec.~\ref{supp:sec:error}, in which the raw coincidence counts are resampled jointly for each measurement configuration.

For the ideal state $\ket{D_N^{\,N/2}}$ the run marginals are independent of $r$ by permutation invariance, and they follow in closed form from the parity selection rule (Theorem~\ref{thm:selection} and its proof): only Pauli strings with an even number of $\sigma_x$ and an even number of $\sigma_z$ survive on $\ket{D_N^{\,N/2}}$, and permutation invariance reduces every symmetrized correlator to a single representative tuple. With $c_x\equiv\cos\theta_x$, $s_x\equiv\sin\theta_x$ and $u_1\equiv\sin^2\theta_1$,
\begin{align}
    m_{xy}(\theta_x,\theta_y)
    &= \frac{\tfrac{N}{2}\,s_x s_y - c_x c_y}{N-1},
    \qquad xy \in \{00,\,01,\,11\},
    \label{eq:m2_theory}
    \\[4pt]
    m_{0011}(\theta_0,\theta_1)
    &= \frac{3\,c_0^2 c_1^2
       - \tfrac{N}{2}\bigl(s_0^2 c_1^2 + c_0^2 s_1^2 + 4\,c_0 s_0 c_1 s_1\bigr)
       + \tfrac{3N(N-2)}{8}\,s_0^2 s_1^2}{(N-1)(N-3)},
    \label{eq:m0011_theory}
    \\[4pt]
    m_{1111}(\theta_1)
    &= \frac{3\Bigl[\tfrac{(N+2)(N+4)}{8}\,u_1^2 - (N+2)\,u_1 + 1\Bigr]}
            {(N-1)(N-3)}.
    \label{eq:m1111_theory}
\end{align}
The same function of $\theta_0$ would give $m_{0000}$.
For the depolarized state $\rho(p)$ of Eq.~\eqref{eq:rho_p} every marginal is rescaled by $(1-p)$. Inserting Eqs.~\eqref{eq:m2_theory}--\eqref{eq:m1111_theory} into \eqref{eq:S_m_relation} reproduces the quantum values $\Bq$ of Secs.~\ref{supp:sec:twobody_insufficient} and \ref{supp:sec:fourbody}, and multiplication by $(1-\pexp)$ gives the white-noise predictions for the run marginals listed in the data tables of Sec.~\ref{supp:sec:exp_details}; the scatter of $m^{(r)}_\alpha$ across $r$ is a direct consistency check of the permutation invariance of the prepared state.

\begin{table}[htb]
\centering
\caption{
Measurement configurations for the six-photon PI Bell tests.
Here, $0$ and $1$ denote the local observables $A_{0}$ and $A_{1}$ of Eq.~\eqref{eq:sym_meas_supp}, and the entry position labels the photon (spatial mode). Experimentally, each run $r$ corresponds to a SET in the table.
}
\label{tab:N6_PI_runs}
\small
\setlength{\tabcolsep}{8pt}
\renewcommand{\arraystretch}{1.15}

\begin{tabular}{|cc|cc|cc|}
\hline
Measurement Setting & Assignment
& Measurement Setting & Assignment
& Measurement Setting & Assignment \\
\hline

\textbf{SET1}  & $001111$
& \textbf{SET6}  & $100111$
& \textbf{SET11} & $110101$ \\

\textbf{SET2}  & $010111$
& \textbf{SET7}  & $101011$
& \textbf{SET12} & $110110$ \\

\textbf{SET3}  & $011011$
& \textbf{SET8}  & $101101$
& \textbf{SET13} & $111001$ \\

\textbf{SET4}  & $011101$
& \textbf{SET9}  & $101110$
& \textbf{SET14} & $111010$ \\

\textbf{SET5}  & $011110$
& \textbf{SET10} & $110011$
& \textbf{SET15} & $111100$ \\

\hline
\end{tabular}
\end{table}

\begin{table}[t]
\centering
\caption{
Measurement configurations for the eight-photon PI Bell test.
Here, $0$ and $1$ denote the local observables $A_{0}$ and $A_{1}$ of Eq.~\eqref{eq:sym_meas_supp}, and the entry position labels the photon (spatial mode). Experimentally, each run $r$ corresponds to a SET in the table.
}
\label{tab:N8_PI_runs}
\small
\setlength{\tabcolsep}{10pt}
\renewcommand{\arraystretch}{1.15}

\begin{tabular}{|cc|cc|}
\hline
Measurement Setting & Assignment
& Measurement Setting & Assignment \\
\hline

\textbf{SET1}  & $00111111$
& \textbf{SET15} & $11010111$ \\

\textbf{SET2}  & $01011111$
& \textbf{SET16} & $11011011$ \\

\textbf{SET3}  & $01101111$
& \textbf{SET17} & $11011101$ \\

\textbf{SET4}  & $01110111$
& \textbf{SET18} & $11011110$ \\

\textbf{SET5}  & $01111011$
& \textbf{SET19} & $11100111$ \\

\textbf{SET6}  & $01111101$
& \textbf{SET20} & $11101011$ \\

\textbf{SET7}  & $01111110$
& \textbf{SET21} & $11101101$ \\

\textbf{SET8}  & $10011111$
& \textbf{SET22} & $11101110$ \\

\textbf{SET9}  & $10101111$
& \textbf{SET23} & $11110011$ \\

\textbf{SET10} & $10110111$
& \textbf{SET24} & $11110101$ \\

\textbf{SET11} & $10111011$
& \textbf{SET25} & $11110110$ \\

\textbf{SET12} & $10111101$
& \textbf{SET26} & $11111001$ \\

\textbf{SET13} & $10111110$
& \textbf{SET27} & $11111010$ \\

\textbf{SET14} & $11001111$
& \textbf{SET28} & $11111100$ \\

\hline
\end{tabular}
\end{table}

\section{Experimental Details and Results}
\label{supp:sec:exp_details}
\subsection{Experimental-Platform Details and Long-Term Stability}

The photon source is based on a ppKTP crystal pumped at $775\,\mathrm{nm}$ to generate orthogonally polarized photon pairs at $1550\,\mathrm{nm}$. The crystal follows the source architecture developed for the Jiuzhang photonic processor~\cite{Zhong2020Jiuzhang} and has sufficient intrinsic spectral purity to operate without additional filtering. For the highest-fidelity state preparations, the crystal is stabilized at $28\,^{\circ}\mathrm{C}$. Temporal walk-off between the $H$- and $V$-polarized photons is compensated using a KDP crystal and a Soleil--Babinet compensator. QWP--HWP pairs correct the polarization transformations introduced by the single-mode fibers, while HWPs set to $0^{\circ}$ compensate the relative $\pi$ phase acquired by modes undergoing an odd number of BS reflections.

After polarization analysis, the two orthogonal polarization outputs of each spatial mode are coupled through optical fibers to separate channels of superconducting nanowire single-photon detectors (SNSPDs), whose detection efficiencies are approximately $90\%$. The six- and eight-photon measurements therefore use $12$ and $16$ detector channels, respectively. We postselect events in which exactly one of the two polarization channels clicks in every spatial mode while the other remains silent. This yields all $2^6=64$ six-photon and $2^8=256$ eight-photon polarization-outcome patterns used to evaluate the corresponding correlations.

Stable operation of the ppKTP source and the collection system is essential for long-term multiphoton measurements. As a representative test of the temperature stability, we continuously recorded the ppKTP-crystal temperature over a two-week period during the acquisition of the highest-fidelity data set, with the crystal stabilized at approximately $28\,^{\circ}\mathrm{C}$, as shown in Fig.~\ref{fig:XD}(a). We further monitor the photon-collection efficiency and collinear $HV$ two-photon interference in all eight spatial modes. The per-mode detection efficiency is determined from the $HV$ coincidence measurements after PBS separation and includes the collection and detection losses, while the interference contrast is defined as the ratio of the coincidence signals measured in the $Z$ and $X$ bases within the same spatial mode, after subtracting accidental coincidences. Their average values reach approximately $72\%$ and $150$, respectively, and remain stable during long-term operation, as shown in Fig.~\ref{fig:XD}(b).
Since the $\binom{N}{2}$ measurement configurations entering each witness are recorded one after another, this stability supports the assumption, stated in the main text, that the same state is prepared in all configurations.

\begin{figure}[htb]
\centering
\includegraphics[width=\textwidth]{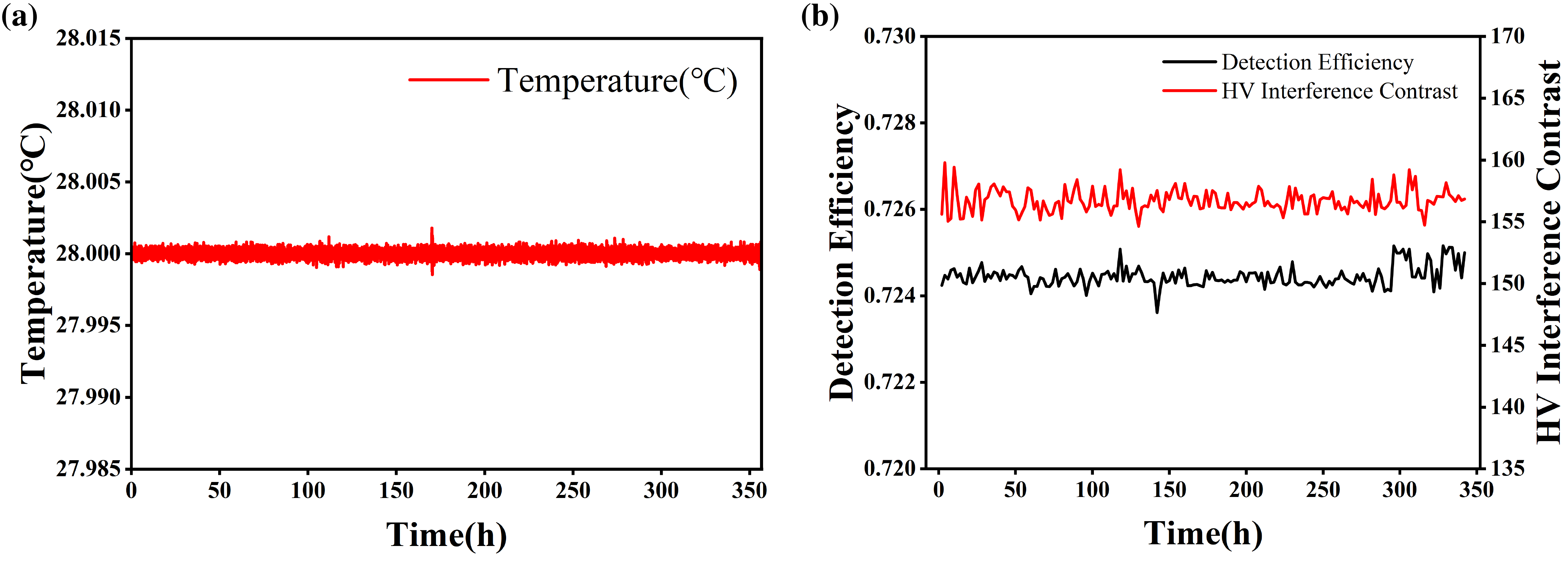}
\caption{
Long-term stability of the experimental platform.
(a) Temperature of the ppKTP crystal over two weeks, maintained at $28.000\pm0.005\,^{\circ}\mathrm{C}$. 
(b) Detection efficiency and collinear $HV$ two-photon interference contrast recorded over two weeks. The plotted detection efficiency is the sum of the $H/V$ two-photon coincidence efficiencies over the eight spatial modes, whereas the interference contrast is averaged over the eight modes. Six- and eight-photon coincidence records are grouped into intervals of $600\,\mathrm{s}$ and $7{,}200\,\mathrm{s}$, respectively. The apparatus is inspected every two hours to ensure comparable experimental conditions throughout data acquisition.}
\label{fig:XD}
\end{figure}

\subsection{Error Analysis}
\label{supp:sec:error}
Of particular importance, several operator expectation values entering the fidelity and PI Bell inequalities are extracted from the same set of raw multiphoton coincidence counts(for example, in the inequality $\mathcal{I}_{1}$, expectation values of $m_{00}$, $m_{11}$, $m_{0011}$, and $m_{0011}$ of SET1 are actually extracted from the original multiphoton coincidence counts of $001111$). They therefore share statistical fluctuations and cannot be treated as independent. To retain these correlations, we apply a Monte Carlo parametric bootstrap based on Poisson counting statistics directly to the raw counts.
For each measurement setting and bootstrap realization $b$, the count $n_k$ associated with outcome $k$ is resampled as
\begin{equation}
    n_k^{(b)}
    \sim
    \operatorname{Poisson}\!\left(n_k\right),
    \label{eq:poisson_bootstrap}
\end{equation}
and the corresponding relative frequency is
\begin{equation}
    \hat{p}_k^{(b)}
    =
    \frac{n_k^{(b)}}{\sum_l n_l^{(b)}}.
    \label{eq:bootstrap_probability}
\end{equation}
For each realization, all expectation values are evaluated from the same resampled count vector, and the full analysis is repeated to obtain the fidelity or PI Bell value $Q^{(b)}$. This procedure preserves correlations between quantities derived from the same measurement setting, while different settings are resampled independently.
The statistical uncertainty is taken as the sample standard deviation of the bootstrap results,
\begin{equation}
    u(Q)
    =
    \sqrt{
        \frac{1}{B-1}
        \sum_{b=1}^{B}
        \left(Q^{(b)}-\overline{Q}\right)^2
    },
    \qquad
    \overline{Q}
    =
    \frac{1}{B}
    \sum_{b=1}^{B} Q^{(b)}.
    \label{eq:bootstrap_uncertainty}
\end{equation}
The reported central values are obtained from the original experimental counts, and the error bars represent one standard deviation of the bootstrap distribution. We use $B=50\,000$ realizations, for which the estimated uncertainties are numerically stable.

\subsection{Experimental Results}

To confirm the performance of our six-photon source, we compare the fidelity and sixfold coincidence rate of the prepared $\lvert D_{6}^{3}\rangle$ state with previous photonic realizations in Table~\ref{tab:six_photon_dicke_comparison}. Among the experiments reporting a global fidelity with respect to the target Dicke state, our preparation achieves both the highest fidelity and the highest sixfold coincidence rate. Although the absolute rates also depend on the source architecture, pump power, collection efficiency, and detector performance, this comparison demonstrates the simultaneous improvement in state quality and generation rate achieved in our experiment.

\begin{table}[t]
    \centering
    \caption{
        Comparison of experimental realizations of the six-photon
        Dicke state $\lvert D_{6}^{3}\rangle$.
        All coincidence rates are given in events per second.
    }
    \label{tab:six_photon_dicke_comparison}
    \setlength{\tabcolsep}{8pt}
    \renewcommand{\arraystretch}{1.15}
    \begin{tabular}{|l|c|c|}
        \hline
        \textbf{Work}
        & \textbf{Reported fidelity}
        & \textbf{Sixfold rate (Hz)} \\
        \hline
        Prevedel \textit{et al.} (2009)~\cite{prevedel_experimental_2009}
        & $0.56(2)$
        & $0.003$ \\
        \hline
        Wieczorek \textit{et al.} (2009)~\cite{wieczorek_experimental_2009}
        & $0.654(24)$
        & $0.062$ \\
        \hline
        Schwemmer \textit{et al.} (2014)~\cite{schwemmer_experimental_2014}
        & $0.678(18)$
        & $0.038$ \\
        \hline
        \textbf{Our work}
        & $\mathbf{0.933(4)}$
        & $\mathbf{6.8}$ \\
        \hline
    \end{tabular}
\end{table}

The set-resolved violation factors and the complete Bell values obtained for the six- and eight-photon Dicke states are summarized in Fig.~\ref{fig:supp_runwise}. The corresponding experimental details, including the measurement settings, expectation values and their uncertainties, and coincidence counts used to determine the fidelities and evaluate the PI Bell inequalities for $\ket{D_6^{\,3}}$ and $\ket{D_8^{\,4}}$, are reported in Tables~\ref{tab:F_D6_GAO}--\ref{tab:PI_D8}.

\begin{figure}
\centering
\includegraphics[width=1\textwidth]{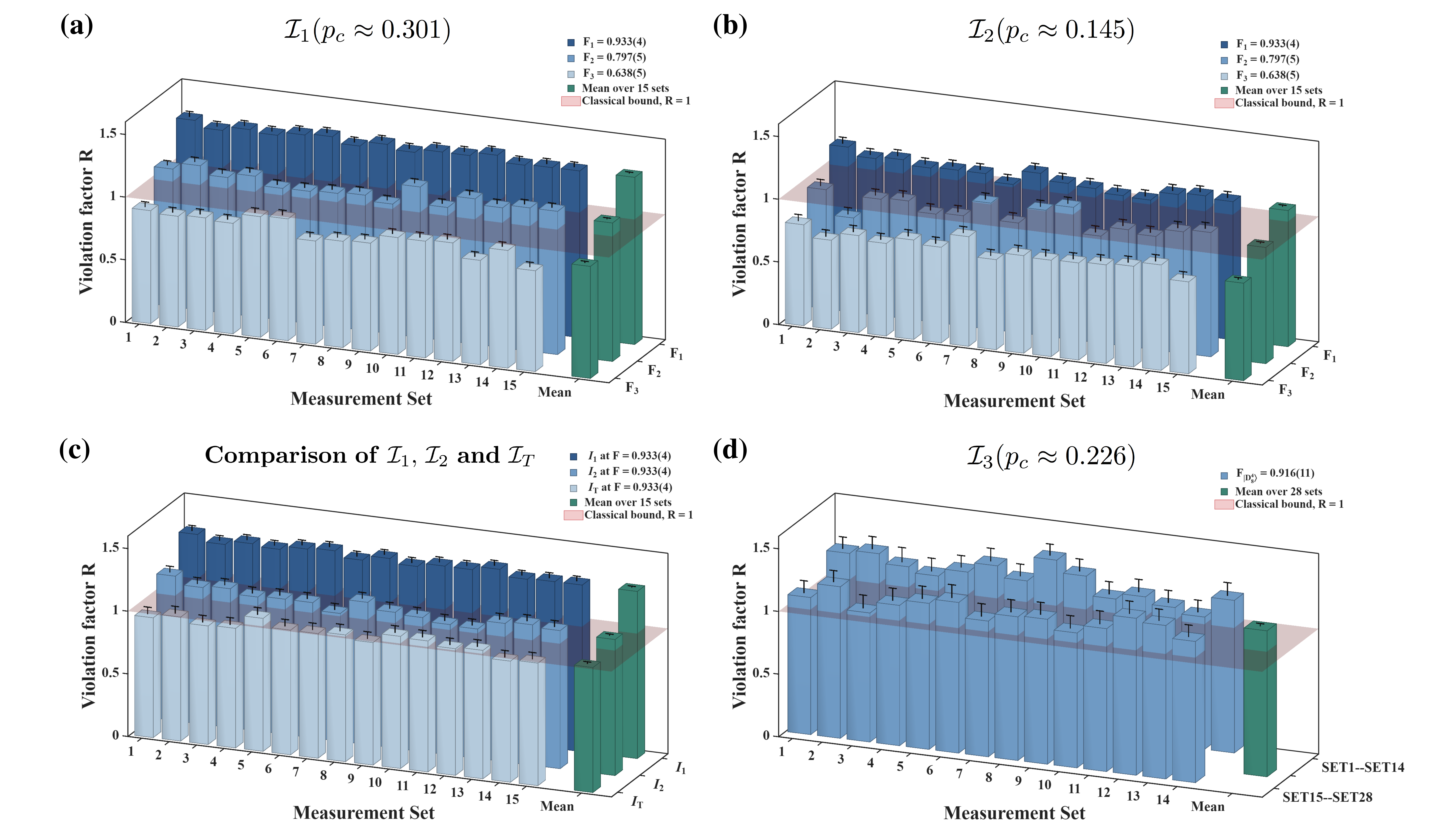}
\caption{
Set-resolved violation factors for the tested PI Bell inequalities.
The numbered bars show the violation factor $R=I/\Bc$ associated with the individual measurement sets, while the green bars give the complete Bell value $R$ obtained by averaging over all sets. 
(a),(b) Results for $\mathcal{I}_{1}$ and $\mathcal{I}_{2}$ from the 15 permutationally symmetric measurement sets using six-photon Dicke states with fidelities $F_{1}=0.933(4)$, $F_{2}=0.797(5)$, and $F_{3}=0.638(5)$.
(c) Comparison between $\mathcal{I}_{1}$, $\mathcal{I}_{2}$ and the reference two-body PI inequality $\mathcal{I}_{\mathrm{T}}$ for the six-photon Dicke state with $F=0.933(4)$.
(d) Results for $\mathcal{I}_{3}$ from the 28 symmetric measurement sets using the eight-photon Dicke state with $F=0.916(11)$.
The red plane denotes the classical bound $R=1$.
Because several expectation values are extracted from the same multiphoton coincidence data, their statistical correlations are retained using a Poisson parametric bootstrap (Supplemental Material).
Error bars denote one standard deviation obtained from Poisson parametric-bootstrap realizations of the raw coincidence counts.
}
\label{fig:supp_runwise}
\end{figure}

\begin{table}[t]
\caption{Experimental Results of the fidelity of six-photon Dicke state at $F=0.933\pm0.004$. \(C_2(A)\), \(C_4(A)\), \(C_6(A)\) are extracted from the six-body expectation of 21 operators according to Eq.(\ref{eq:D6_c}). Under experimental configuration with parametric photon no filtering, 3000 mW pump power, ppKTP crystal temperature stabilized at 28$^{\circ}$C. Each measurement basis accumulated approximately 4000 events over a 600-second integration time. Error bars are estimated using Poisson Monte Carlo parametric bootstrapping of the coincidence counts.}
\label{tab:F_D6_GAO}
\setlength{\tabcolsep}{2.5pt}
\renewcommand{\arraystretch}{0.85}
\begin{ruledtabular}
\footnotesize
\begin{tabular}{crrr}
\multicolumn{1}{c}{Measurement setting \(A\)}
& \multicolumn{1}{c}{\(C_2(A)\)}
& \multicolumn{1}{c}{\(C_4(A)\)}
& \multicolumn{1}{c}{\(C_6(A)\)} \\
\hline
\(X\) & \(0.609\pm0.008\) & \(0.593\pm0.008\) & \(0.902\pm0.007\) \\
\(Y\) & \(0.625\pm0.008\) & \(0.609\pm0.008\) & \(0.908\pm0.006\) \\
\(Z\) & \(-0.192\pm0.001\) & \(0.183\pm0.001\) & \(-0.873\pm0.008\) \\
\(X+Y\) & \(0.611\pm0.008\) & \(0.596\pm0.008\) & \(0.908\pm0.006\) \\
\(X-Y\) & \(0.625\pm0.008\) & \(0.609\pm0.008\) & \(0.897\pm0.007\) \\
\(Y+Z\) & \(0.208\pm0.005\) & \(-0.078\pm0.006\) & \(-0.030\pm0.016\) \\
\(Y-Z\) & \(0.187\pm0.005\) & \(-0.081\pm0.006\) & \(0.033\pm0.016\) \\
\(X+Z\) & \(0.200\pm0.005\) & \(-0.098\pm0.006\) & \(-0.053\pm0.016\) \\
\(X-Z\) & \(0.191\pm0.005\) & \(-0.087\pm0.006\) & \(0.010\pm0.015\) \\
\(X+Y+Z\) & \(0.337\pm0.006\) & \(0.039\pm0.008\) & \(-0.337\pm0.015\) \\
\(X+Y-Z\) & \(0.316\pm0.006\) & \(0.020\pm0.008\) & \(-0.323\pm0.015\) \\
\(X-Y+Z\) & \(0.334\pm0.006\) & \(0.030\pm0.008\) & \(-0.324\pm0.015\) \\
\(X-Y-Z\) & \(0.337\pm0.006\) & \(0.042\pm0.008\) & \(-0.311\pm0.015\) \\
\(\cos(\pi/6)X+\sin(\pi/6)Z\) & \(0.414\pm0.006\) & \(0.140\pm0.009\) & \(-0.327\pm0.015\) \\
\(\cos(2\pi/6)X+\sin(2\pi/6)Z\) & \(0.002\pm0.003\) & \(-0.065\pm0.003\) & \(0.380\pm0.015\) \\
\(\cos(4\pi/6)X+\sin(4\pi/6)Z\) & \(-0.007\pm0.003\) & \(-0.061\pm0.003\) & \(0.377\pm0.015\) \\
\(\cos(5\pi/6)X+\sin(5\pi/6)Z\) & \(0.398\pm0.006\) & \(0.120\pm0.009\) & \(-0.321\pm0.015\) \\
\(\cos(\pi/6)Y+\sin(\pi/6)Z\) & \(0.417\pm0.007\) & \(0.162\pm0.009\) & \(-0.314\pm0.015\) \\
\(\cos(2\pi/6)Y+\sin(2\pi/6)Z\) & \(0.007\pm0.003\) & \(-0.075\pm0.003\) & \(0.385\pm0.015\) \\
\(\cos(4\pi/6)Y+\sin(4\pi/6)Z\) & \(-0.010\pm0.003\) & \(-0.060\pm0.003\) & \(0.379\pm0.015\) \\
\(\cos(5\pi/6)Y+\sin(5\pi/6)Z\) & \(0.385\pm0.006\) & \(0.112\pm0.009\) & \(-0.353\pm0.014\) \\
\hline
\multicolumn{1}{c}{Fidelity of \(\ket{D_6^{\,3}}\)}
& \multicolumn{3}{c}{\(0.933\pm0.004\)} \\
\end{tabular}
\end{ruledtabular}
\end{table}

\begin{table}[t]
\caption{
Experimental Results of  inequalities $\mathcal{I}_1$ sets in $\ket{D_6^{\,3}}$ at $F=0.933\pm0.004$ in optimal point $P_{0}$. \(m_{00}\), \(m_{11}\), \(m_{0011}\), and \(m_{1111}\) are extracted from the six-body expectation of 15 operators according to Eq.\eqref{eq:run_marginal_correlators}---Eq.\eqref{eq:S_m_relation}. Under experimental configuration with parametric photon no filtering, 3000 mW pump power, ppKTP crystal temperature stabilized at 28$^{\circ}$C. Each measurement basis accumulated approximately 4000 events over a 600-second integration time. Error bars are estimated using Poisson Monte Carlo parametric bootstrapping of the coincidence counts.}
\label{tab:PI_GAO_0}
\setlength{\tabcolsep}{2.5pt}
\renewcommand{\arraystretch}{0.85}
\begin{ruledtabular}
\footnotesize
\begin{tabular}{crrrr}
\multicolumn{1}{c}{\(\vphantom{\bm{\mathcal{RUN}}}\)Run}
& \multicolumn{1}{c}{\(m_{00}\)}
& \multicolumn{1}{c}{\(m_{11}\)}
& \multicolumn{1}{c}{\(m_{0011}\)}
& \multicolumn{1}{c}{\(m_{1111}\)} \\
\hline
\(\vphantom{\bm{\mathcal{RUN}}}\bm{\mathcal{RUN}}\,\bm{1}\)  & \(-0.057\pm0.016\) & \(0.019\pm0.006\) & \(0.209\pm0.006\) & \(-0.073\pm0.016\) \\
\(\vphantom{\bm{\mathcal{RUN}}}\bm{\mathcal{RUN}}\,\bm{2}\)  & \(-0.029\pm0.016\) & \(0.026\pm0.007\) & \(0.206\pm0.006\) & \(-0.079\pm0.016\) \\
\(\vphantom{\bm{\mathcal{RUN}}}\bm{\mathcal{RUN}}\,\bm{3}\)  & \(-0.018\pm0.016\) & \(0.021\pm0.006\) & \(0.201\pm0.006\) & \(-0.099\pm0.016\) \\
\(\vphantom{\bm{\mathcal{RUN}}}\bm{\mathcal{RUN}}\,\bm{4}\)  & \(-0.065\pm0.016\) & \(0.014\pm0.006\) & \(0.198\pm0.006\) & \(-0.071\pm0.016\) \\
\(\vphantom{\bm{\mathcal{RUN}}}\bm{\mathcal{RUN}}\,\bm{5}\)  & \(-0.067\pm0.016\) & \(0.014\pm0.006\) & \(0.201\pm0.006\) & \(-0.075\pm0.016\) \\
\(\vphantom{\bm{\mathcal{RUN}}}\bm{\mathcal{RUN}}\,\bm{6}\)  & \(-0.069\pm0.016\) & \(0.015\pm0.006\) & \(0.194\pm0.006\) & \(-0.090\pm0.016\) \\
\(\vphantom{\bm{\mathcal{RUN}}}\bm{\mathcal{RUN}}\,\bm{7}\)  & \(-0.045\pm0.016\) & \(0.019\pm0.006\) & \(0.202\pm0.006\) & \(-0.077\pm0.015\) \\
\(\vphantom{\bm{\mathcal{RUN}}}\bm{\mathcal{RUN}}\,\bm{8}\)  & \(-0.023\pm0.016\) & \(0.024\pm0.006\) & \(0.202\pm0.006\) & \(-0.103\pm0.016\) \\
\(\vphantom{\bm{\mathcal{RUN}}}\bm{\mathcal{RUN}}\,\bm{9}\)  & \(-0.049\pm0.016\) & \(0.021\pm0.006\) & \(0.203\pm0.006\) & \(-0.076\pm0.016\) \\
\(\vphantom{\bm{\mathcal{RUN}}}\bm{\mathcal{RUN}}\,\bm{10}\) & \(-0.050\pm0.016\) & \(0.016\pm0.006\) & \(0.192\pm0.006\) & \(-0.099\pm0.016\) \\
\(\vphantom{\bm{\mathcal{RUN}}}\bm{\mathcal{RUN}}\,\bm{11}\) & \(-0.043\pm0.016\) & \(0.016\pm0.006\) & \(0.193\pm0.006\) & \(-0.100\pm0.016\) \\
\(\vphantom{\bm{\mathcal{RUN}}}\bm{\mathcal{RUN}}\,\bm{12}\) & \(-0.046\pm0.016\) & \(0.021\pm0.007\) & \(0.207\pm0.006\) & \(-0.089\pm0.016\) \\
\(\vphantom{\bm{\mathcal{RUN}}}\bm{\mathcal{RUN}}\,\bm{13}\) & \(-0.012\pm0.016\) & \(0.029\pm0.006\) & \(0.201\pm0.006\) & \(-0.103\pm0.016\) \\
\(\vphantom{\bm{\mathcal{RUN}}}\bm{\mathcal{RUN}}\,\bm{14}\) & \(-0.052\pm0.016\) & \(0.023\pm0.006\) & \(0.199\pm0.006\) & \(-0.085\pm0.016\) \\
\(\vphantom{\bm{\mathcal{RUN}}}\bm{\mathcal{RUN}}\,\bm{15}\) & \(-0.039\pm0.016\) & \(0.017\pm0.006\) & \(0.196\pm0.006\) & \(-0.094\pm0.016\) \\
\hline
\(\bar{m}\) & \(-0.044\pm0.004\) & \(0.020\pm0.002\) & \(0.200\pm0.001\) & \(-0.088\pm0.004\) \\
\(S\)      & \(-1.331\pm0.121\)  & \(0.590\pm0.049\)  & \(72.095\pm0.515\) & \(-31.485\pm1.451\) \\
\hline
\(I\)      & \multicolumn{4}{c}{\(-24.012\pm0.186\)} \\
\(R\)      & \multicolumn{4}{c}{\(1.334\pm0.010\)} \\
\end{tabular}
\end{ruledtabular}
\end{table}

\begin{table}[t]
\caption{
Experimental Results of inequalities $\mathcal{I}_1$ sets in $\ket{D_6^{\,3}}$ at $F=0.933\pm0.004$ in coordinate $P_{1}$. \(m_{00}\), \(m_{11}\), \(m_{0011}\), and \(m_{1111}\) are extracted from the six-body expectation of 15 operators according to Eq.\eqref{eq:run_marginal_correlators}---Eq.\eqref{eq:S_m_relation}. Under experimental configuration with parametric photon no filtering, 3000 mW pump power, ppKTP crystal temperature stabilized at 28$^{\circ}$C. Each measurement basis accumulated approximately 4000 events over a 600-second integration time. Error bars are estimated using Poisson Monte Carlo parametric bootstrapping of the coincidence counts.}
\label{tab:PI_GAO_P1}
\setlength{\tabcolsep}{2.5pt}
\renewcommand{\arraystretch}{0.85}
\begin{ruledtabular}
\small
\begin{tabular}{crrrr}
\multicolumn{1}{c}{\(\vphantom{\bm{\mathcal{RUN}}}\)Run}
& \multicolumn{1}{c}{\(m_{00}\)}
& \multicolumn{1}{c}{\(m_{11}\)}
& \multicolumn{1}{c}{\(m_{0011}\)}
& \multicolumn{1}{c}{\(m_{1111}\)} \\
\hline
\(\vphantom{\bm{\mathcal{RUN}}}\bm{\mathcal{RUN}}\,\bm{1}\)  & \(-0.166\pm0.016\) & \(0.030\pm0.007\) & \(0.124\pm0.006\) & \(-0.069\pm0.016\) \\
\(\vphantom{\bm{\mathcal{RUN}}}\bm{\mathcal{RUN}}\,\bm{2}\)  & \(-0.182\pm0.016\) & \(0.021\pm0.007\) & \(0.130\pm0.006\) & \(-0.074\pm0.016\) \\
\(\vphantom{\bm{\mathcal{RUN}}}\bm{\mathcal{RUN}}\,\bm{3}\)  & \(-0.196\pm0.016\) & \(0.023\pm0.007\) & \(0.116\pm0.006\) & \(-0.067\pm0.016\) \\
\(\vphantom{\bm{\mathcal{RUN}}}\bm{\mathcal{RUN}}\,\bm{4}\)  & \(-0.145\pm0.016\) & \(0.028\pm0.007\) & \(0.136\pm0.006\) & \(-0.096\pm0.016\) \\
\(\vphantom{\bm{\mathcal{RUN}}}\bm{\mathcal{RUN}}\,\bm{5}\)  & \(-0.186\pm0.016\) & \(0.014\pm0.006\) & \(0.133\pm0.006\) & \(-0.107\pm0.016\) \\
\(\vphantom{\bm{\mathcal{RUN}}}\bm{\mathcal{RUN}}\,\bm{6}\)  & \(-0.193\pm0.016\) & \(0.033\pm0.007\) & \(0.126\pm0.006\) & \(-0.077\pm0.016\) \\
\(\vphantom{\bm{\mathcal{RUN}}}\bm{\mathcal{RUN}}\,\bm{7}\)  & \(-0.190\pm0.016\) & \(0.014\pm0.006\) & \(0.128\pm0.006\) & \(-0.103\pm0.016\) \\
\(\vphantom{\bm{\mathcal{RUN}}}\bm{\mathcal{RUN}}\,\bm{8}\)  & \(-0.191\pm0.015\) & \(0.017\pm0.006\) & \(0.135\pm0.006\) & \(-0.100\pm0.016\) \\
\(\vphantom{\bm{\mathcal{RUN}}}\bm{\mathcal{RUN}}\,\bm{9}\)  & \(-0.187\pm0.015\) & \(0.021\pm0.006\) & \(0.117\pm0.006\) & \(-0.102\pm0.016\) \\
\(\vphantom{\bm{\mathcal{RUN}}}\bm{\mathcal{RUN}}\,\bm{10}\) & \(-0.165\pm0.016\) & \(0.012\pm0.006\) & \(0.117\pm0.006\) & \(-0.079\pm0.016\) \\
\(\vphantom{\bm{\mathcal{RUN}}}\bm{\mathcal{RUN}}\,\bm{11}\) & \(-0.170\pm0.016\) & \(0.023\pm0.007\) & \(0.128\pm0.006\) & \(-0.069\pm0.016\) \\
\(\vphantom{\bm{\mathcal{RUN}}}\bm{\mathcal{RUN}}\,\bm{12}\) & \(-0.185\pm0.015\) & \(0.005\pm0.006\) & \(0.114\pm0.006\) & \(-0.081\pm0.016\) \\
\(\vphantom{\bm{\mathcal{RUN}}}\bm{\mathcal{RUN}}\,\bm{13}\) & \(-0.175\pm0.016\) & \(0.019\pm0.007\) & \(0.127\pm0.006\) & \(-0.055\pm0.016\) \\
\(\vphantom{\bm{\mathcal{RUN}}}\bm{\mathcal{RUN}}\,\bm{14}\) & \(-0.197\pm0.015\) & \(0.021\pm0.006\) & \(0.119\pm0.006\) & \(-0.089\pm0.016\) \\
\(\vphantom{\bm{\mathcal{RUN}}}\bm{\mathcal{RUN}}\,\bm{15}\) & \(-0.187\pm0.015\) & \(0.025\pm0.006\) & \(0.119\pm0.006\) & \(-0.079\pm0.016\) \\
\hline
\(\bar{m}\) & \(-0.181\pm0.004\) & \(0.020\pm0.002\) & \(0.125\pm0.002\) & \(-0.083\pm0.004\) \\
\(S\) & \(-5.428\pm0.120\) & \(0.613\pm0.050\) & \(44.853\pm0.570\) & \(-29.931\pm1.461\) \\
\hline
\(I\) & \multicolumn{4}{c}{\(-21.017\pm0.262\)} \\
\(R\) & \multicolumn{4}{c}{\(1.168\pm0.015\)} \\
\end{tabular}
\end{ruledtabular}
\end{table}

\begin{table}[t]
\caption{
Experimental Results of inequalities $\mathcal{I}_1$ sets in $\ket{D_6^{\,3}}$ at $F=0.933\pm0.004$ in coordinate $P_{2}$. \(m_{00}\), \(m_{11}\), \(m_{0011}\), and \(m_{1111}\) are extracted from the six-body expectation of 15 operators according to Eq.\eqref{eq:run_marginal_correlators}---Eq.\eqref{eq:S_m_relation}. Under experimental configuration with parametric photon no filtering, 3000 mW pump power, ppKTP crystal temperature stabilized at 28$^{\circ}$C. Each measurement basis accumulated approximately 4000 events over a 600-second integration time. Error bars are estimated using Poisson Monte Carlo parametric bootstrapping of the coincidence counts.}
\label{tab:PI_GAO_P2}
\setlength{\tabcolsep}{2.5pt}
\renewcommand{\arraystretch}{0.85}
\begin{ruledtabular}
\small
\begin{tabular}{crrrr}
\multicolumn{1}{c}{\(\vphantom{\bm{\mathcal{RUN}}}\)Run}
& \multicolumn{1}{c}{\(m_{00}\)}
& \multicolumn{1}{c}{\(m_{11}\)}
& \multicolumn{1}{c}{\(m_{0011}\)}
& \multicolumn{1}{c}{\(m_{1111}\)} \\
\hline
\(\vphantom{\bm{\mathcal{RUN}}}\bm{\mathcal{RUN}}\,\bm{1}\)  & \(0.141\pm0.016\) & \(0.021\pm0.006\) & \(0.221\pm0.005\) & \(-0.082\pm0.016\) \\
\(\vphantom{\bm{\mathcal{RUN}}}\bm{\mathcal{RUN}}\,\bm{2}\)  & \(0.125\pm0.016\) & \(0.019\pm0.007\) & \(0.220\pm0.006\) & \(-0.074\pm0.016\) \\
\(\vphantom{\bm{\mathcal{RUN}}}\bm{\mathcal{RUN}}\,\bm{3}\)  & \(0.155\pm0.016\) & \(0.027\pm0.007\) & \(0.224\pm0.006\) & \(-0.089\pm0.016\) \\
\(\vphantom{\bm{\mathcal{RUN}}}\bm{\mathcal{RUN}}\,\bm{4}\)  & \(0.143\pm0.016\) & \(0.016\pm0.006\) & \(0.217\pm0.005\) & \(-0.078\pm0.016\) \\
\(\vphantom{\bm{\mathcal{RUN}}}\bm{\mathcal{RUN}}\,\bm{5}\)  & \(0.128\pm0.015\) & \(0.026\pm0.006\) & \(0.221\pm0.005\) & \(-0.083\pm0.015\) \\
\(\vphantom{\bm{\mathcal{RUN}}}\bm{\mathcal{RUN}}\,\bm{6}\)  & \(0.140\pm0.015\) & \(0.019\pm0.006\) & \(0.217\pm0.005\) & \(-0.100\pm0.015\) \\
\(\vphantom{\bm{\mathcal{RUN}}}\bm{\mathcal{RUN}}\,\bm{7}\)  & \(0.167\pm0.015\) & \(0.026\pm0.006\) & \(0.216\pm0.005\) & \(-0.113\pm0.016\) \\
\(\vphantom{\bm{\mathcal{RUN}}}\bm{\mathcal{RUN}}\,\bm{8}\)  & \(0.135\pm0.016\) & \(0.013\pm0.006\) & \(0.220\pm0.005\) & \(-0.060\pm0.016\) \\
\(\vphantom{\bm{\mathcal{RUN}}}\bm{\mathcal{RUN}}\,\bm{9}\)  & \(0.140\pm0.016\) & \(0.032\pm0.007\) & \(0.228\pm0.006\) & \(-0.067\pm0.016\) \\
\(\vphantom{\bm{\mathcal{RUN}}}\bm{\mathcal{RUN}}\,\bm{10}\) & \(0.159\pm0.016\) & \(0.019\pm0.006\) & \(0.214\pm0.006\) & \(-0.088\pm0.016\) \\
\(\vphantom{\bm{\mathcal{RUN}}}\bm{\mathcal{RUN}}\,\bm{11}\) & \(0.133\pm0.016\) & \(0.019\pm0.007\) & \(0.226\pm0.005\) & \(-0.065\pm0.016\) \\
\(\vphantom{\bm{\mathcal{RUN}}}\bm{\mathcal{RUN}}\,\bm{12}\) & \(0.132\pm0.016\) & \(0.017\pm0.006\) & \(0.214\pm0.005\) & \(-0.091\pm0.016\) \\
\(\vphantom{\bm{\mathcal{RUN}}}\bm{\mathcal{RUN}}\,\bm{13}\) & \(0.196\pm0.015\) & \(0.031\pm0.007\) & \(0.218\pm0.006\) & \(-0.105\pm0.016\) \\
\(\vphantom{\bm{\mathcal{RUN}}}\bm{\mathcal{RUN}}\,\bm{14}\) & \(0.133\pm0.016\) & \(0.020\pm0.006\) & \(0.217\pm0.005\) & \(-0.086\pm0.016\) \\
\(\vphantom{\bm{\mathcal{RUN}}}\bm{\mathcal{RUN}}\,\bm{15}\) & \(0.118\pm0.016\) & \(0.008\pm0.006\) & \(0.210\pm0.005\) & \(-0.077\pm0.016\) \\
\hline
\(\bar{m}\) & \(0.143\pm0.004\) & \(0.021\pm0.002\) & \(0.219\pm0.001\) & \(-0.084\pm0.004\) \\
\(S\) & \(4.288\pm0.121\) & \(0.626\pm0.050\) & \(78.753\pm0.504\) & \(-30.143\pm1.457\) \\
\hline
\(I\) & \multicolumn{4}{c}{\(-19.798\pm0.155\)} \\
\(R\) & \multicolumn{4}{c}{\(1.100\pm0.009\)} \\
\end{tabular}
\end{ruledtabular}
\end{table}

\begin{table}[t]
\caption{
Experimental Results of inequalities $\mathcal{I}_1$ sets in $\ket{D_6^{\,3}}$ at $F=0.933\pm0.004$ in coordinate $P_{3}$. \(m_{00}\), \(m_{11}\), \(m_{0011}\), and \(m_{1111}\) are extracted from the six-body expectation of 15 operators according to Eq.\eqref{eq:run_marginal_correlators}---Eq.\eqref{eq:S_m_relation}. Under experimental configuration with parametric photon no filtering, 3000 mW pump power, ppKTP crystal temperature stabilized at 28$^{\circ}$C. Each measurement basis accumulated approximately 4000 events over a 600-second integration time. Error bars are estimated using Poisson Monte Carlo parametric bootstrapping of the coincidence counts.}
\label{tab:PI_GAO_P3}
\setlength{\tabcolsep}{2.5pt}
\renewcommand{\arraystretch}{0.85}
\begin{ruledtabular}
\small
\begin{tabular}{crrrr}
\multicolumn{1}{c}{\(\vphantom{\bm{\mathcal{RUN}}}\)Run}
& \multicolumn{1}{c}{\(m_{00}\)}
& \multicolumn{1}{c}{\(m_{11}\)}
& \multicolumn{1}{c}{\(m_{0011}\)}
& \multicolumn{1}{c}{\(m_{1111}\)} \\
\hline
\(\vphantom{\bm{\mathcal{RUN}}}\bm{\mathcal{RUN}}\,\bm{1}\)  & \(-0.064\pm0.016\) & \(0.144\pm0.008\) & \(0.192\pm0.007\) & \(-0.090\pm0.016\) \\
\(\vphantom{\bm{\mathcal{RUN}}}\bm{\mathcal{RUN}}\,\bm{2}\)  & \(-0.046\pm0.016\) & \(0.151\pm0.008\) & \(0.195\pm0.007\) & \(-0.099\pm0.016\) \\
\(\vphantom{\bm{\mathcal{RUN}}}\bm{\mathcal{RUN}}\,\bm{3}\)  & \(-0.054\pm0.016\) & \(0.161\pm0.008\) & \(0.192\pm0.008\) & \(-0.100\pm0.016\) \\
\(\vphantom{\bm{\mathcal{RUN}}}\bm{\mathcal{RUN}}\,\bm{4}\)  & \(-0.006\pm0.016\) & \(0.164\pm0.008\) & \(0.208\pm0.007\) & \(-0.098\pm0.016\) \\
\(\vphantom{\bm{\mathcal{RUN}}}\bm{\mathcal{RUN}}\,\bm{5}\)  & \(-0.054\pm0.016\) & \(0.148\pm0.008\) & \(0.182\pm0.007\) & \(-0.121\pm0.016\) \\
\(\vphantom{\bm{\mathcal{RUN}}}\bm{\mathcal{RUN}}\,\bm{6}\)  & \(-0.064\pm0.016\) & \(0.163\pm0.008\) & \(0.188\pm0.008\) & \(-0.091\pm0.016\) \\
\(\vphantom{\bm{\mathcal{RUN}}}\bm{\mathcal{RUN}}\,\bm{7}\)  & \(-0.063\pm0.016\) & \(0.150\pm0.007\) & \(0.179\pm0.007\) & \(-0.136\pm0.015\) \\
\(\vphantom{\bm{\mathcal{RUN}}}\bm{\mathcal{RUN}}\,\bm{8}\)  & \(-0.026\pm0.016\) & \(0.160\pm0.008\) & \(0.197\pm0.007\) & \(-0.100\pm0.016\) \\
\(\vphantom{\bm{\mathcal{RUN}}}\bm{\mathcal{RUN}}\,\bm{9}\)  & \(-0.038\pm0.016\) & \(0.167\pm0.008\) & \(0.200\pm0.008\) & \(-0.083\pm0.016\) \\
\(\vphantom{\bm{\mathcal{RUN}}}\bm{\mathcal{RUN}}\,\bm{10}\) & \(-0.020\pm0.016\) & \(0.156\pm0.008\) & \(0.188\pm0.007\) & \(-0.113\pm0.016\) \\
\(\vphantom{\bm{\mathcal{RUN}}}\bm{\mathcal{RUN}}\,\bm{11}\) & \(-0.022\pm0.016\) & \(0.161\pm0.008\) & \(0.195\pm0.007\) & \(-0.101\pm0.016\) \\
\(\vphantom{\bm{\mathcal{RUN}}}\bm{\mathcal{RUN}}\,\bm{12}\) & \(-0.035\pm0.016\) & \(0.166\pm0.008\) & \(0.191\pm0.007\) & \(-0.105\pm0.016\) \\
\(\vphantom{\bm{\mathcal{RUN}}}\bm{\mathcal{RUN}}\,\bm{13}\) & \(-0.033\pm0.016\) & \(0.160\pm0.008\) & \(0.196\pm0.007\) & \(-0.122\pm0.016\) \\
\(\vphantom{\bm{\mathcal{RUN}}}\bm{\mathcal{RUN}}\,\bm{14}\) & \(-0.054\pm0.016\) & \(0.159\pm0.008\) & \(0.193\pm0.007\) & \(-0.106\pm0.016\) \\
\(\vphantom{\bm{\mathcal{RUN}}}\bm{\mathcal{RUN}}\,\bm{15}\) & \(-0.039\pm0.016\) & \(0.158\pm0.008\) & \(0.181\pm0.007\) & \(-0.107\pm0.016\) \\
\hline
\(\bar{m}\) & \(-0.041\pm0.004\) & \(0.158\pm0.002\) & \(0.192\pm0.002\) & \(-0.105\pm0.004\) \\
\(S\) & \(-1.235\pm0.122\) & \(4.731\pm0.058\) & \(69.040\pm0.682\) & \(-37.706\pm1.454\) \\
\hline
\(I\) & \multicolumn{4}{c}{\(-20.048\pm0.262\)} \\
\(R\) & \multicolumn{4}{c}{\(1.114\pm0.015\)} \\
\end{tabular}
\end{ruledtabular}
\end{table}

\begin{table}[t]
\caption{
Experimental Results of inequalities $\mathcal{I}_1$ sets in $\ket{D_6^{\,3}}$ at $F=0.933\pm0.004$ in coordinate $P_{4}$. \(m_{00}\), \(m_{11}\), \(m_{0011}\), and \(m_{1111}\) are extracted from the six-body expectation of 15 operators according to Eq.\eqref{eq:run_marginal_correlators}---Eq.\eqref{eq:S_m_relation}. Under experimental configuration with parametric photon no filtering, 3000 mW pump power, ppKTP crystal temperature stabilized at 28$^{\circ}$C. Each measurement basis accumulated approximately 4000 events over a 600-second integration time. Error bars are estimated using Poisson Monte Carlo parametric bootstrapping of the coincidence counts.}
\label{tab:PI_GAO_P4}
\setlength{\tabcolsep}{2.5pt}
\renewcommand{\arraystretch}{0.85}
\begin{ruledtabular}
\small
\begin{tabular}{crrrr}
\multicolumn{1}{c}{\(\vphantom{\bm{\mathcal{RUN}}}\)Run}
& \multicolumn{1}{c}{\(m_{00}\)}
& \multicolumn{1}{c}{\(m_{11}\)}
& \multicolumn{1}{c}{\(m_{0011}\)}
& \multicolumn{1}{c}{\(m_{1111}\)} \\
\hline
\(\vphantom{\bm{\mathcal{RUN}}}\bm{\mathcal{RUN}}\,\bm{1}\)  & \(-0.051\pm0.016\) & \(-0.093\pm0.005\) & \(0.188\pm0.004\) & \(0.015\pm0.016\) \\
\(\vphantom{\bm{\mathcal{RUN}}}\bm{\mathcal{RUN}}\,\bm{2}\)  & \(-0.053\pm0.016\) & \(-0.098\pm0.005\) & \(0.185\pm0.004\) & \(0.011\pm0.016\) \\
\(\vphantom{\bm{\mathcal{RUN}}}\bm{\mathcal{RUN}}\,\bm{3}\)  & \(-0.028\pm0.016\) & \(-0.089\pm0.005\) & \(0.185\pm0.004\) & \(-0.003\pm0.016\) \\
\(\vphantom{\bm{\mathcal{RUN}}}\bm{\mathcal{RUN}}\,\bm{4}\)  & \(-0.032\pm0.016\) & \(-0.092\pm0.005\) & \(0.190\pm0.004\) & \(0.029\pm0.016\) \\
\(\vphantom{\bm{\mathcal{RUN}}}\bm{\mathcal{RUN}}\,\bm{5}\)  & \(-0.073\pm0.016\) & \(-0.092\pm0.005\) & \(0.198\pm0.004\) & \(0.046\pm0.016\) \\
\(\vphantom{\bm{\mathcal{RUN}}}\bm{\mathcal{RUN}}\,\bm{6}\)  & \(-0.052\pm0.016\) & \(-0.092\pm0.005\) & \(0.191\pm0.004\) & \(0.027\pm0.016\) \\
\(\vphantom{\bm{\mathcal{RUN}}}\bm{\mathcal{RUN}}\,\bm{7}\)  & \(-0.014\pm0.016\) & \(-0.085\pm0.005\) & \(0.189\pm0.004\) & \(0.004\pm0.016\) \\
\(\vphantom{\bm{\mathcal{RUN}}}\bm{\mathcal{RUN}}\,\bm{8}\)  & \(-0.026\pm0.016\) & \(-0.087\pm0.005\) & \(0.188\pm0.004\) & \(0.016\pm0.016\) \\
\(\vphantom{\bm{\mathcal{RUN}}}\bm{\mathcal{RUN}}\,\bm{9}\)  & \(-0.062\pm0.016\) & \(-0.094\pm0.005\) & \(0.189\pm0.004\) & \(0.033\pm0.016\) \\
\(\vphantom{\bm{\mathcal{RUN}}}\bm{\mathcal{RUN}}\,\bm{10}\) & \(-0.059\pm0.016\) & \(-0.098\pm0.005\) & \(0.191\pm0.004\) & \(0.046\pm0.016\) \\
\(\vphantom{\bm{\mathcal{RUN}}}\bm{\mathcal{RUN}}\,\bm{11}\) & \(-0.032\pm0.016\) & \(-0.095\pm0.005\) & \(0.189\pm0.004\) & \(0.030\pm0.016\) \\
\(\vphantom{\bm{\mathcal{RUN}}}\bm{\mathcal{RUN}}\,\bm{12}\) & \(-0.046\pm0.016\) & \(-0.094\pm0.005\) & \(0.182\pm0.004\) & \(0.003\pm0.016\) \\
\(\vphantom{\bm{\mathcal{RUN}}}\bm{\mathcal{RUN}}\,\bm{13}\) & \(0.005\pm0.016\) & \(-0.074\pm0.005\) & \(0.195\pm0.005\) & \(0.006\pm0.016\) \\
\(\vphantom{\bm{\mathcal{RUN}}}\bm{\mathcal{RUN}}\,\bm{14}\) & \(-0.078\pm0.016\) & \(-0.088\pm0.005\) & \(0.198\pm0.004\) & \(0.044\pm0.016\) \\
\(\vphantom{\bm{\mathcal{RUN}}}\bm{\mathcal{RUN}}\,\bm{15}\) & \(-0.020\pm0.016\) & \(-0.078\pm0.005\) & \(0.192\pm0.005\) & \(0.010\pm0.016\) \\
\hline
\(\bar{m}\) & \(-0.041\pm0.004\) & \(-0.090\pm0.001\) & \(0.190\pm0.001\) & \(0.021\pm0.004\) \\
\(S\) & \(-1.242\pm0.122\) & \(-2.695\pm0.039\) & \(68.387\pm0.400\) & \(7.570\pm1.468\) \\
\hline
\(I\) & \multicolumn{4}{c}{\(-19.772\pm0.139\)} \\
\(R\) & \multicolumn{4}{c}{\(1.098\pm0.008\)} \\
\end{tabular}
\end{ruledtabular}
\end{table}

\begin{table}[t]
\caption{
Experimental Results of inequalities $\mathcal{I}_1$ sets in $\ket{D_6^{\,3}}$ at $F=0.933\pm0.004$ in coordinate $P_{5}$. \(m_{00}\), \(m_{11}\), \(m_{0011}\), and \(m_{1111}\) are extracted from the six-body expectation of 15 operators according to Eq.\eqref{eq:run_marginal_correlators}---Eq.\eqref{eq:S_m_relation}. Under experimental configuration with parametric photon no filtering, 3000 mW pump power, ppKTP crystal temperature stabilized at 28$^{\circ}$C. Each measurement basis accumulated approximately 4000 events over a 600-second integration time. Error bars are estimated using Poisson Monte Carlo parametric bootstrapping of the coincidence counts.}
\label{tab:PI_GAO_P5}
\setlength{\tabcolsep}{2.5pt}
\renewcommand{\arraystretch}{0.85}
\begin{ruledtabular}
\small
\begin{tabular}{crrrr}
\multicolumn{1}{c}{\(\vphantom{\bm{\mathcal{RUN}}}\)Run}
& \multicolumn{1}{c}{\(m_{00}\)}
& \multicolumn{1}{c}{\(m_{11}\)}
& \multicolumn{1}{c}{\(m_{0011}\)}
& \multicolumn{1}{c}{\(m_{1111}\)} \\
\hline
\(\vphantom{\bm{\mathcal{RUN}}}\bm{\mathcal{RUN}}\,\bm{1}\)  & \(-0.185\pm0.016\) & \(-0.125\pm0.005\) & \(0.172\pm0.004\) & \(0.079\pm0.016\) \\
\(\vphantom{\bm{\mathcal{RUN}}}\bm{\mathcal{RUN}}\,\bm{2}\)  & \(-0.209\pm0.016\) & \(-0.129\pm0.005\) & \(0.167\pm0.004\) & \(0.076\pm0.016\) \\
\(\vphantom{\bm{\mathcal{RUN}}}\bm{\mathcal{RUN}}\,\bm{3}\)  & \(-0.195\pm0.016\) & \(-0.119\pm0.005\) & \(0.169\pm0.004\) & \(0.062\pm0.016\) \\
\(\vphantom{\bm{\mathcal{RUN}}}\bm{\mathcal{RUN}}\,\bm{4}\)  & \(-0.205\pm0.016\) & \(-0.124\pm0.005\) & \(0.176\pm0.004\) & \(0.080\pm0.016\) \\
\(\vphantom{\bm{\mathcal{RUN}}}\bm{\mathcal{RUN}}\,\bm{5}\)  & \(-0.207\pm0.016\) & \(-0.129\pm0.005\) & \(0.170\pm0.004\) & \(0.087\pm0.016\) \\
\(\vphantom{\bm{\mathcal{RUN}}}\bm{\mathcal{RUN}}\,\bm{6}\)  & \(-0.192\pm0.016\) & \(-0.123\pm0.005\) & \(0.170\pm0.004\) & \(0.076\pm0.016\) \\
\(\vphantom{\bm{\mathcal{RUN}}}\bm{\mathcal{RUN}}\,\bm{7}\)  & \(-0.180\pm0.016\) & \(-0.130\pm0.004\) & \(0.162\pm0.004\) & \(0.059\pm0.016\) \\
\(\vphantom{\bm{\mathcal{RUN}}}\bm{\mathcal{RUN}}\,\bm{8}\)  & \(-0.211\pm0.016\) & \(-0.124\pm0.005\) & \(0.175\pm0.004\) & \(0.065\pm0.016\) \\
\(\vphantom{\bm{\mathcal{RUN}}}\bm{\mathcal{RUN}}\,\bm{9}\)  & \(-0.205\pm0.016\) & \(-0.130\pm0.005\) & \(0.170\pm0.004\) & \(0.097\pm0.016\) \\
\(\vphantom{\bm{\mathcal{RUN}}}\bm{\mathcal{RUN}}\,\bm{10}\) & \(-0.198\pm0.016\) & \(-0.133\pm0.005\) & \(0.168\pm0.004\) & \(0.105\pm0.016\) \\
\(\vphantom{\bm{\mathcal{RUN}}}\bm{\mathcal{RUN}}\,\bm{11}\) & \(-0.170\pm0.016\) & \(-0.121\pm0.005\) & \(0.170\pm0.004\) & \(0.066\pm0.016\) \\
\(\vphantom{\bm{\mathcal{RUN}}}\bm{\mathcal{RUN}}\,\bm{12}\) & \(-0.202\pm0.016\) & \(-0.128\pm0.004\) & \(0.169\pm0.004\) & \(0.071\pm0.016\) \\
\(\vphantom{\bm{\mathcal{RUN}}}\bm{\mathcal{RUN}}\,\bm{13}\) & \(-0.190\pm0.016\) & \(-0.126\pm0.005\) & \(0.169\pm0.004\) & \(0.087\pm0.016\) \\
\(\vphantom{\bm{\mathcal{RUN}}}\bm{\mathcal{RUN}}\,\bm{14}\) & \(-0.206\pm0.016\) & \(-0.132\pm0.004\) & \(0.171\pm0.004\) & \(0.076\pm0.016\) \\
\(\vphantom{\bm{\mathcal{RUN}}}\bm{\mathcal{RUN}}\,\bm{15}\) & \(-0.191\pm0.016\) & \(-0.121\pm0.005\) & \(0.173\pm0.004\) & \(0.067\pm0.016\) \\
\hline
\(\bar{m}\) & \(-0.196\pm0.004\) & \(-0.126\pm0.001\) & \(0.170\pm0.001\) & \(0.077\pm0.004\) \\
\(S\) & \(-5.893\pm0.122\) & \(-3.784\pm0.035\) & \(61.234\pm0.384\) & \(27.649\pm1.484\) \\
\hline
\(I\) & \multicolumn{4}{c}{\(-20.377\pm0.170\)} \\
\(R\) & \multicolumn{4}{c}{\(1.132\pm0.009\)} \\
\end{tabular}
\end{ruledtabular}
\end{table}

\begin{table}[t]
\caption{
Experimental Results of inequalities $\mathcal{I}_1$ sets in $\ket{D_6^{\,3}}$ at $F=0.933\pm0.004$ in coordinate $P_{6}$. \(m_{00}\), \(m_{11}\), \(m_{0011}\), and \(m_{1111}\) are extracted from the six-body expectation of 15 operators according to Eq.\eqref{eq:run_marginal_correlators}---Eq.\eqref{eq:S_m_relation}. Under experimental configuration with parametric photon no filtering, 3000 mW pump power, ppKTP crystal temperature stabilized at 28$^{\circ}$C. Each measurement basis accumulated approximately 4000 events over a 600-second integration time. Error bars are estimated using Poisson Monte Carlo parametric bootstrapping of the coincidence counts.}
\label{tab:PI_GAO_P6}
\setlength{\tabcolsep}{2.5pt}
\renewcommand{\arraystretch}{0.85}
\begin{ruledtabular}
\small
\begin{tabular}{crrrr}
\multicolumn{1}{c}{\(\vphantom{\bm{\mathcal{RUN}}}\)Run}
& \multicolumn{1}{c}{\(m_{00}\)}
& \multicolumn{1}{c}{\(m_{11}\)}
& \multicolumn{1}{c}{\(m_{0011}\)}
& \multicolumn{1}{c}{\(m_{1111}\)} \\
\hline
\(\vphantom{\bm{\mathcal{RUN}}}\bm{\mathcal{RUN}}\,\bm{1}\)  & \(0.181\pm0.015\) & \(0.165\pm0.007\) & \(0.270\pm0.007\) & \(-0.107\pm0.015\) \\
\(\vphantom{\bm{\mathcal{RUN}}}\bm{\mathcal{RUN}}\,\bm{2}\)  & \(0.190\pm0.015\) & \(0.163\pm0.007\) & \(0.270\pm0.007\) & \(-0.092\pm0.015\) \\
\(\vphantom{\bm{\mathcal{RUN}}}\bm{\mathcal{RUN}}\,\bm{3}\)  & \(0.203\pm0.015\) & \(0.172\pm0.008\) & \(0.274\pm0.007\) & \(-0.096\pm0.015\) \\
\(\vphantom{\bm{\mathcal{RUN}}}\bm{\mathcal{RUN}}\,\bm{4}\)  & \(0.215\pm0.015\) & \(0.185\pm0.008\) & \(0.289\pm0.007\) & \(-0.078\pm0.015\) \\
\(\vphantom{\bm{\mathcal{RUN}}}\bm{\mathcal{RUN}}\,\bm{5}\)  & \(0.176\pm0.016\) & \(0.169\pm0.008\) & \(0.279\pm0.007\) & \(-0.076\pm0.016\) \\
\(\vphantom{\bm{\mathcal{RUN}}}\bm{\mathcal{RUN}}\,\bm{6}\)  & \(0.206\pm0.015\) & \(0.179\pm0.008\) & \(0.271\pm0.007\) & \(-0.104\pm0.016\) \\
\(\vphantom{\bm{\mathcal{RUN}}}\bm{\mathcal{RUN}}\,\bm{7}\)  & \(0.187\pm0.015\) & \(0.164\pm0.008\) & \(0.269\pm0.007\) & \(-0.091\pm0.016\) \\
\(\vphantom{\bm{\mathcal{RUN}}}\bm{\mathcal{RUN}}\,\bm{8}\)  & \(0.229\pm0.016\) & \(0.169\pm0.008\) & \(0.273\pm0.007\) & \(-0.116\pm0.016\) \\
\(\vphantom{\bm{\mathcal{RUN}}}\bm{\mathcal{RUN}}\,\bm{9}\)  & \(0.202\pm0.015\) & \(0.157\pm0.008\) & \(0.257\pm0.007\) & \(-0.123\pm0.016\) \\
\(\vphantom{\bm{\mathcal{RUN}}}\bm{\mathcal{RUN}}\,\bm{10}\) & \(0.171\pm0.016\) & \(0.157\pm0.008\) & \(0.261\pm0.007\) & \(-0.091\pm0.016\) \\
\(\vphantom{\bm{\mathcal{RUN}}}\bm{\mathcal{RUN}}\,\bm{11}\) & \(0.203\pm0.015\) & \(0.174\pm0.008\) & \(0.290\pm0.007\) & \(-0.073\pm0.016\) \\
\(\vphantom{\bm{\mathcal{RUN}}}\bm{\mathcal{RUN}}\,\bm{12}\) & \(0.193\pm0.016\) & \(0.152\pm0.008\) & \(0.259\pm0.007\) & \(-0.116\pm0.016\) \\
\(\vphantom{\bm{\mathcal{RUN}}}\bm{\mathcal{RUN}}\,\bm{13}\) & \(0.221\pm0.015\) & \(0.177\pm0.008\) & \(0.282\pm0.007\) & \(-0.096\pm0.016\) \\
\(\vphantom{\bm{\mathcal{RUN}}}\bm{\mathcal{RUN}}\,\bm{14}\) & \(0.193\pm0.015\) & \(0.166\pm0.008\) & \(0.271\pm0.007\) & \(-0.091\pm0.015\) \\
\(\vphantom{\bm{\mathcal{RUN}}}\bm{\mathcal{RUN}}\,\bm{15}\) & \(0.206\pm0.015\) & \(0.171\pm0.008\) & \(0.284\pm0.007\) & \(-0.084\pm0.016\) \\
\hline
\(\bar{m}\) & \(0.198\pm0.004\) & \(0.168\pm0.002\) & \(0.273\pm0.002\) & \(-0.096\pm0.004\) \\
\(S\) & \(5.950\pm0.119\) & \(5.040\pm0.059\) & \(98.365\pm0.629\) & \(-34.385\pm1.445\) \\
\hline
\(I\) & \multicolumn{4}{c}{\(-19.332\pm0.195\)} \\
\(R\) & \multicolumn{4}{c}{\(1.074\pm0.011\)} \\
\end{tabular}
\end{ruledtabular}
\end{table}

\begin{table}[t]
\caption{
Experimental Results of  inequalities $\mathcal{I}_2$ sets in $\ket{D_6^{\,3}}$ at $F=0.933\pm0.004$. \(m_{00}\), \(m_{11}\), \(m_{0011}\), and \(m_{1111}\) are extracted from the six-body expectation of 15 operators according to Eq.\eqref{eq:run_marginal_correlators}---Eq.\eqref{eq:S_m_relation}. Under experimental configuration with parametric photon no filtering, 3000 mW pump power, ppKTP crystal temperature stabilized at \(28^\circ\mathrm{C}\). Each measurement basis accumulated approximately 4000 events over a 600-second integration time. Error bars are estimated using Poisson Monte Carlo parametric bootstrapping of the coincidence counts.}
\label{tab:PI_GAO_1}
\setlength{\tabcolsep}{2.5pt}
\renewcommand{\arraystretch}{0.85}
\begin{ruledtabular}
\small
\begin{tabular}{crrrr}
\multicolumn{1}{c}{\(\vphantom{\bm{\mathcal{RUN}}}\)Run}
& \multicolumn{1}{c}{\(m_{00}\)}
& \multicolumn{1}{c}{\(m_{11}\)}
& \multicolumn{1}{c}{\(m_{0011}\)}
& \multicolumn{1}{c}{\(m_{1111}\)} \\
\hline
\(\vphantom{\bm{\mathcal{RUN}}}\bm{\mathcal{RUN}}\,\bm{1}\)  & \(-0.107\pm0.015\) & \(0.159\pm0.007\) & \(0.166\pm0.007\) & \(-0.096\pm0.015\) \\
\(\vphantom{\bm{\mathcal{RUN}}}\bm{\mathcal{RUN}}\,\bm{2}\)  & \(-0.103\pm0.015\) & \(0.158\pm0.007\) & \(0.150\pm0.007\) & \(-0.107\pm0.015\) \\
\(\vphantom{\bm{\mathcal{RUN}}}\bm{\mathcal{RUN}}\,\bm{3}\)  & \(-0.118\pm0.015\) & \(0.147\pm0.007\) & \(0.143\pm0.007\) & \(-0.127\pm0.015\) \\
\(\vphantom{\bm{\mathcal{RUN}}}\bm{\mathcal{RUN}}\,\bm{4}\)  & \(-0.093\pm0.015\) & \(0.160\pm0.007\) & \(0.151\pm0.007\) & \(-0.128\pm0.015\) \\
\(\vphantom{\bm{\mathcal{RUN}}}\bm{\mathcal{RUN}}\,\bm{5}\)  & \(-0.104\pm0.016\) & \(0.154\pm0.008\) & \(0.148\pm0.008\) & \(-0.105\pm0.016\) \\
\(\vphantom{\bm{\mathcal{RUN}}}\bm{\mathcal{RUN}}\,\bm{6}\)  & \(-0.101\pm0.016\) & \(0.153\pm0.007\) & \(0.148\pm0.007\) & \(-0.135\pm0.015\) \\
\(\vphantom{\bm{\mathcal{RUN}}}\bm{\mathcal{RUN}}\,\bm{7}\)  & \(-0.093\pm0.015\) & \(0.152\pm0.007\) & \(0.142\pm0.007\) & \(-0.107\pm0.015\) \\
\(\vphantom{\bm{\mathcal{RUN}}}\bm{\mathcal{RUN}}\,\bm{8}\)  & \(-0.100\pm0.016\) & \(0.169\pm0.008\) & \(0.165\pm0.008\) & \(-0.086\pm0.016\) \\
\(\vphantom{\bm{\mathcal{RUN}}}\bm{\mathcal{RUN}}\,\bm{9}\)  & \(-0.101\pm0.015\) & \(0.155\pm0.007\) & \(0.152\pm0.007\) & \(-0.102\pm0.015\) \\
\(\vphantom{\bm{\mathcal{RUN}}}\bm{\mathcal{RUN}}\,\bm{10}\) & \(-0.104\pm0.015\) & \(0.146\pm0.007\) & \(0.147\pm0.007\) & \(-0.114\pm0.015\) \\
\(\vphantom{\bm{\mathcal{RUN}}}\bm{\mathcal{RUN}}\,\bm{11}\) & \(-0.078\pm0.016\) & \(0.159\pm0.008\) & \(0.163\pm0.008\) & \(-0.074\pm0.016\) \\
\(\vphantom{\bm{\mathcal{RUN}}}\bm{\mathcal{RUN}}\,\bm{12}\) & \(-0.089\pm0.015\) & \(0.149\pm0.007\) & \(0.153\pm0.007\) & \(-0.088\pm0.015\) \\
\(\vphantom{\bm{\mathcal{RUN}}}\bm{\mathcal{RUN}}\,\bm{13}\) & \(-0.105\pm0.015\) & \(0.156\pm0.008\) & \(0.154\pm0.008\) & \(-0.095\pm0.016\) \\
\(\vphantom{\bm{\mathcal{RUN}}}\bm{\mathcal{RUN}}\,\bm{14}\) & \(-0.131\pm0.015\) & \(0.145\pm0.007\) & \(0.137\pm0.007\) & \(-0.099\pm0.015\) \\
\(\vphantom{\bm{\mathcal{RUN}}}\bm{\mathcal{RUN}}\,\bm{15}\) & \(-0.105\pm0.016\) & \(0.152\pm0.007\) & \(0.152\pm0.007\) & \(-0.127\pm0.016\) \\
\hline
\(\bar{m}\) & \(-0.102\pm0.004\) & \(0.154\pm0.002\) & \(0.151\pm0.002\) & \(-0.106\pm0.004\) \\
\(S\)      & \(-3.062\pm0.119\) & \(4.626\pm0.057\) & \(54.503\pm0.689\) & \(-38.165\pm1.431\) \\
\hline
\(I\)      & \multicolumn{4}{c}{\(-146.578\pm1.675\)} \\
\(R\)      & \multicolumn{4}{c}{\(1.086\pm0.012\)} \\
\end{tabular}
\end{ruledtabular}
\end{table}

\begin{table}[t]
\caption{Experimental Results of inequalities $\mathcal{I}_T$ sets in $\ket{D_6^{\,3}}$ at $F=0.933\pm0.004$. \(m_{00}\), \(m_{11}\), and \(m_{01}\) are extracted from the six-body expectation of 15 operators according to Eq.\eqref{eq:run_marginal_correlators}---Eq.\eqref{eq:S_m_relation}. Under experimental configuration with parametric photon no filtering, 3000 mW pump power, ppKTP crystal temperature stabilized at \(28^\circ\mathrm{C}\). Each measurement basis accumulated approximately 4000 events over a 600-second integration time. Error bars are estimated using Poisson Monte Carlo parametric bootstrapping of the coincidence counts.}
\label{tab:PI_GAO_Tura}
\setlength{\tabcolsep}{2.5pt}
\renewcommand{\arraystretch}{0.85}
\begin{ruledtabular}
\small
\begin{tabular}{crrr}
\multicolumn{1}{c}{\(\vphantom{\bm{\mathcal{RUN}}}\)Run}
& \multicolumn{1}{c}{\(m_{00}\)}
& \multicolumn{1}{c}{\(m_{11}\)}
& \multicolumn{1}{c}{\(m_{01}\)} \\
\hline
\(\vphantom{\bm{\mathcal{RUN}}}\bm{\mathcal{RUN}}\,\bm{1}\)  & \(-0.163\pm0.016\) & \(0.307\pm0.009\) & \(-0.180\pm0.007\) \\
\(\vphantom{\bm{\mathcal{RUN}}}\bm{\mathcal{RUN}}\,\bm{2}\)  & \(-0.182\pm0.016\) & \(0.331\pm0.009\) & \(-0.150\pm0.007\) \\
\(\vphantom{\bm{\mathcal{RUN}}}\bm{\mathcal{RUN}}\,\bm{3}\)  & \(-0.153\pm0.015\) & \(0.308\pm0.008\) & \(-0.198\pm0.007\) \\
\(\vphantom{\bm{\mathcal{RUN}}}\bm{\mathcal{RUN}}\,\bm{4}\)  & \(-0.158\pm0.016\) & \(0.312\pm0.008\) & \(-0.190\pm0.007\) \\
\(\vphantom{\bm{\mathcal{RUN}}}\bm{\mathcal{RUN}}\,\bm{5}\)  & \(-0.183\pm0.016\) & \(0.332\pm0.008\) & \(-0.195\pm0.007\) \\
\(\vphantom{\bm{\mathcal{RUN}}}\bm{\mathcal{RUN}}\,\bm{6}\)  & \(-0.167\pm0.016\) & \(0.322\pm0.008\) & \(-0.192\pm0.007\) \\
\(\vphantom{\bm{\mathcal{RUN}}}\bm{\mathcal{RUN}}\,\bm{7}\)  & \(-0.161\pm0.015\) & \(0.320\pm0.008\) & \(-0.207\pm0.007\) \\
\(\vphantom{\bm{\mathcal{RUN}}}\bm{\mathcal{RUN}}\,\bm{8}\)  & \(-0.176\pm0.016\) & \(0.325\pm0.008\) & \(-0.181\pm0.007\) \\
\(\vphantom{\bm{\mathcal{RUN}}}\bm{\mathcal{RUN}}\,\bm{9}\)  & \(-0.161\pm0.016\) & \(0.333\pm0.009\) & \(-0.193\pm0.007\) \\
\(\vphantom{\bm{\mathcal{RUN}}}\bm{\mathcal{RUN}}\,\bm{10}\) & \(-0.180\pm0.015\) & \(0.331\pm0.008\) & \(-0.199\pm0.007\) \\
\(\vphantom{\bm{\mathcal{RUN}}}\bm{\mathcal{RUN}}\,\bm{11}\) & \(-0.182\pm0.016\) & \(0.327\pm0.008\) & \(-0.186\pm0.007\) \\
\(\vphantom{\bm{\mathcal{RUN}}}\bm{\mathcal{RUN}}\,\bm{12}\) & \(-0.172\pm0.015\) & \(0.353\pm0.008\) & \(-0.184\pm0.007\) \\
\(\vphantom{\bm{\mathcal{RUN}}}\bm{\mathcal{RUN}}\,\bm{13}\) & \(-0.172\pm0.015\) & \(0.335\pm0.008\) & \(-0.198\pm0.007\) \\
\(\vphantom{\bm{\mathcal{RUN}}}\bm{\mathcal{RUN}}\,\bm{14}\) & \(-0.161\pm0.015\) & \(0.343\pm0.008\) & \(-0.185\pm0.007\) \\
\(\vphantom{\bm{\mathcal{RUN}}}\bm{\mathcal{RUN}}\,\bm{15}\) & \(-0.166\pm0.015\) & \(0.333\pm0.008\) & \(-0.182\pm0.007\) \\
\hline
\(\bar{m}\) & \(-0.169\pm0.004\) & \(0.328\pm0.002\) & \(-0.188\pm0.002\) \\
\(S\)      & \(-5.071\pm0.120\) & \(9.826\pm0.065\) & \(-5.642\pm0.054\) \\
\hline
\(I\)      & \multicolumn{3}{c}{\(-119.743\pm1.656\)} \\
\(R\)      & \multicolumn{3}{c}{\(0.998\pm0.014\)} \\
\end{tabular}
\end{ruledtabular}
\end{table}

\begin{table}[t]
\caption{Experimental Results of the fidelity of six-photon Dicke state at $F=0.797\pm0.005$. \(C_2(A)\), \(C_4(A)\), \(C_6(A)\) are extracted from the six-body expectation of 21 operators according to Eq.(\ref{eq:D6_c}). Under experimental configuration with parametric photon no filtering, 3000 mW pump power, ppKTP crystal temperature stabilized at \(35^\circ\mathrm{C}\) and walk-off compensation (KDP) at \(2.5^\circ\) offset. Each measurement basis accumulated approximately 4000 events over a 600-second integration time. Error bars are estimated using Poisson Monte Carlo parametric bootstrapping of the coincidence counts.}
\label{tab:F_D6_ZHONG}
\setlength{\tabcolsep}{2.5pt}
\renewcommand{\arraystretch}{0.85}
\begin{ruledtabular}
\footnotesize
\begin{tabular}{crrr}
\multicolumn{1}{c}{Measurement setting \(A\)}
& \multicolumn{1}{c}{\(C_2(A)\)}
& \multicolumn{1}{c}{\(C_4(A)\)}
& \multicolumn{1}{c}{\(C_6(A)\)} \\
\hline
\(X\) & \(0.618\pm0.007\) & \(0.574\pm0.008\) & \(0.753\pm0.010\) \\
\(Y\) & \(0.600\pm0.008\) & \(0.555\pm0.008\) & \(0.729\pm0.010\) \\
\(Z\) & \(-0.170\pm0.001\) & \(0.143\pm0.002\) & \(-0.599\pm0.012\) \\
\(X+Y\) & \(0.604\pm0.008\) & \(0.585\pm0.008\) & \(0.868\pm0.008\) \\
\(X-Y\) & \(0.589\pm0.008\) & \(0.534\pm0.009\) & \(0.659\pm0.012\) \\
\(Y+Z\) & \(0.154\pm0.005\) & \(-0.096\pm0.005\) & \(0.140\pm0.016\) \\
\(Y-Z\) & \(0.267\pm0.006\) & \(-0.028\pm0.007\) & \(-0.198\pm0.015\) \\
\(X+Z\) & \(0.282\pm0.006\) & \(-0.026\pm0.007\) & \(-0.244\pm0.015\) \\
\(X-Z\) & \(0.115\pm0.004\) & \(-0.106\pm0.004\) & \(0.221\pm0.015\) \\
\(X+Y+Z\) & \(0.342\pm0.006\) & \(0.048\pm0.008\) & \(-0.364\pm0.014\) \\
\(X+Y-Z\) & \(0.324\pm0.006\) & \(0.028\pm0.008\) & \(-0.320\pm0.014\) \\
\(X-Y+Z\) & \(0.431\pm0.007\) & \(0.169\pm0.009\) & \(-0.286\pm0.015\) \\
\(X-Y-Z\) & \(0.224\pm0.005\) & \(-0.070\pm0.006\) & \(-0.097\pm0.016\) \\
\(\cos(\pi/6)X+\sin(\pi/6)Z\) & \(0.481\pm0.007\) & \(0.279\pm0.009\) & \(-0.115\pm0.015\) \\
\(\cos(2\pi/6)X+\sin(2\pi/6)Z\) & \(0.094\pm0.004\) & \(-0.099\pm0.004\) & \(0.287\pm0.015\) \\
\(\cos(4\pi/6)X+\sin(4\pi/6)Z\) & \(-0.064\pm0.003\) & \(-0.003\pm0.003\) & \(0.193\pm0.016\) \\
\(\cos(5\pi/6)X+\sin(5\pi/6)Z\) & \(0.326\pm0.006\) & \(0.024\pm0.008\) & \(-0.302\pm0.015\) \\
\(\cos(\pi/6)Y+\sin(\pi/6)Z\) & \(0.345\pm0.006\) & \(0.064\pm0.008\) & \(-0.309\pm0.015\) \\
\(\cos(2\pi/6)Y+\sin(2\pi/6)Z\) & \(-0.045\pm0.003\) & \(-0.027\pm0.003\) & \(0.281\pm0.015\) \\
\(\cos(4\pi/6)Y+\sin(4\pi/6)Z\) & \(0.060\pm0.004\) & \(-0.101\pm0.004\) & \(0.325\pm0.015\) \\
\(\cos(5\pi/6)Y+\sin(5\pi/6)Z\) & \(0.464\pm0.007\) & \(0.243\pm0.009\) & \(-0.172\pm0.015\) \\
\hline
\multicolumn{1}{c}{Fidelity of \(\ket{D_6^{\,3}}\)}
& \multicolumn{3}{c}{\(0.797\pm0.005\)} \\
\end{tabular}
\end{ruledtabular}
\end{table}

\begin{table}[t]
\caption{Experimental Results of  inequalities $\mathcal{I}_1$ sets in $\ket{D_6^{\,3}}$ at $F=0.797\pm0.005$. \(m_{00}\), \(m_{11}\), \(m_{0011}\), and \(m_{1111}\) are extracted from the six-body expectation of 15 operators according to Eq.\eqref{eq:run_marginal_correlators}---Eq.\eqref{eq:S_m_relation}. Under experimental configuration with parametric photon no filtering, 3000 mW pump power, ppKTP crystal temperature stabilized at \(35^\circ\mathrm{C}\) and walk-off compensation (KDP) at \(2.5^\circ\) offset. Each measurement basis accumulated approximately 4000 events over a 600-second integration time. Error bars are estimated using Poisson Monte Carlo parametric bootstrapping of the coincidence counts.}
\label{tab:PI_ZHONG_0}
\setlength{\tabcolsep}{2.5pt}
\renewcommand{\arraystretch}{0.85}
\begin{ruledtabular}
\small
\begin{tabular}{crrrr}
\multicolumn{1}{c}{\(\vphantom{\bm{\mathcal{RUN}}}\)Run}
& \multicolumn{1}{c}{\(m_{00}\)}
& \multicolumn{1}{c}{\(m_{11}\)}
& \multicolumn{1}{c}{\(m_{0011}\)}
& \multicolumn{1}{c}{\(m_{1111}\)} \\
\hline
\(\vphantom{\bm{\mathcal{RUN}}}\bm{\mathcal{RUN}}\,\bm{1}\)  & \(0.029\pm0.016\) & \(-0.036\pm0.006\) & \(0.189\pm0.005\) & \(-0.042\pm0.016\) \\
\(\vphantom{\bm{\mathcal{RUN}}}\bm{\mathcal{RUN}}\,\bm{2}\)  & \(0.007\pm0.016\) & \(-0.046\pm0.006\) & \(0.179\pm0.005\) & \(-0.058\pm0.016\) \\
\(\vphantom{\bm{\mathcal{RUN}}}\bm{\mathcal{RUN}}\,\bm{3}\)  & \(0.024\pm0.016\) & \(-0.044\pm0.006\) & \(0.187\pm0.005\) & \(-0.035\pm0.016\) \\
\(\vphantom{\bm{\mathcal{RUN}}}\bm{\mathcal{RUN}}\,\bm{4}\)  & \(0.034\pm0.016\) & \(-0.032\pm0.006\) & \(0.198\pm0.005\) & \(-0.041\pm0.016\) \\
\(\vphantom{\bm{\mathcal{RUN}}}\bm{\mathcal{RUN}}\,\bm{5}\)  & \(0.019\pm0.016\) & \(-0.037\pm0.006\) & \(0.191\pm0.005\) & \(-0.021\pm0.016\) \\
\(\vphantom{\bm{\mathcal{RUN}}}\bm{\mathcal{RUN}}\,\bm{6}\)  & \(0.027\pm0.016\) & \(-0.040\pm0.006\) & \(0.188\pm0.005\) & \(-0.029\pm0.016\) \\
\(\vphantom{\bm{\mathcal{RUN}}}\bm{\mathcal{RUN}}\,\bm{7}\)  & \(0.022\pm0.016\) & \(-0.033\pm0.006\) & \(0.186\pm0.005\) & \(-0.035\pm0.016\) \\
\(\vphantom{\bm{\mathcal{RUN}}}\bm{\mathcal{RUN}}\,\bm{8}\)  & \(0.036\pm0.016\) & \(-0.036\pm0.006\) & \(0.193\pm0.005\) & \(-0.035\pm0.016\) \\
\(\vphantom{\bm{\mathcal{RUN}}}\bm{\mathcal{RUN}}\,\bm{9}\)  & \(0.037\pm0.016\) & \(-0.023\pm0.006\) & \(0.195\pm0.005\) & \(-0.025\pm0.016\) \\
\(\vphantom{\bm{\mathcal{RUN}}}\bm{\mathcal{RUN}}\,\bm{10}\) & \(0.021\pm0.016\) & \(-0.040\pm0.006\) & \(0.189\pm0.005\) & \(-0.068\pm0.016\) \\
\(\vphantom{\bm{\mathcal{RUN}}}\bm{\mathcal{RUN}}\,\bm{11}\) & \(0.026\pm0.016\) & \(-0.038\pm0.006\) & \(0.188\pm0.005\) & \(-0.032\pm0.016\) \\
\(\vphantom{\bm{\mathcal{RUN}}}\bm{\mathcal{RUN}}\,\bm{12}\) & \(0.008\pm0.016\) & \(-0.031\pm0.006\) & \(0.191\pm0.005\) & \(-0.050\pm0.016\) \\
\(\vphantom{\bm{\mathcal{RUN}}}\bm{\mathcal{RUN}}\,\bm{13}\) & \(0.069\pm0.016\) & \(-0.009\pm0.006\) & \(0.197\pm0.005\) & \(-0.070\pm0.016\) \\
\(\vphantom{\bm{\mathcal{RUN}}}\bm{\mathcal{RUN}}\,\bm{14}\) & \(-0.021\pm0.016\) & \(-0.039\pm0.006\) & \(0.197\pm0.005\) & \(-0.019\pm0.016\) \\
\(\vphantom{\bm{\mathcal{RUN}}}\bm{\mathcal{RUN}}\,\bm{15}\) & \(0.036\pm0.016\) & \(-0.030\pm0.006\) & \(0.191\pm0.005\) & \(-0.060\pm0.016\) \\
\hline
\(\bar{m}\) & \(0.025\pm0.004\) & \(-0.034\pm0.002\) & \(0.190\pm0.001\) & \(-0.041\pm0.004\) \\
\(S\) & \(0.748\pm0.121\) & \(-1.027\pm0.045\) & \(68.551\pm0.461\) & \(-14.866\pm1.451\) \\
\hline
\(I\) & \multicolumn{4}{c}{\(-19.895\pm0.164\)} \\
\(R\) & \multicolumn{4}{c}{\(1.105\pm0.009\)} \\
\end{tabular}
\end{ruledtabular}
\end{table}

\begin{table}[t]
\caption{Experimental Results of  inequalities $\mathcal{I}_2$ sets in $\ket{D_6^{\,3}}$ at $F=0.797\pm0.005$. \(m_{00}\), \(m_{11}\), \(m_{0011}\), and \(m_{1111}\) are extracted from the six-body expectation of 15 operators according to Eq.\eqref{eq:run_marginal_correlators}---Eq.\eqref{eq:S_m_relation}. Under experimental configuration with parametric photon no filtering, 3000 mW pump power, ppKTP crystal temperature stabilized at \(35^\circ\mathrm{C}\) and walk-off compensation (KDP) at \(2.5^\circ\) offset. Each measurement basis accumulated approximately 4000 events over a 600-second integration time. Error bars are estimated using Poisson Monte Carlo parametric bootstrapping of the coincidence counts.}
\label{tab:PI_ZHONG_1}
\setlength{\tabcolsep}{2.5pt}
\renewcommand{\arraystretch}{0.85}
\begin{ruledtabular}
\small
\begin{tabular}{crrrr}
\multicolumn{1}{c}{\(\vphantom{\bm{\mathcal{RUN}}}\)Run}
& \multicolumn{1}{c}{\(m_{00}\)}
& \multicolumn{1}{c}{\(m_{11}\)}
& \multicolumn{1}{c}{\(m_{0011}\)}
& \multicolumn{1}{c}{\(m_{1111}\)} \\
\hline
\(\vphantom{\bm{\mathcal{RUN}}}\bm{\mathcal{RUN}}\,\bm{1}\)  & \(-0.060\pm0.016\) & \(0.079\pm0.007\) & \(0.186\pm0.007\) & \(-0.072\pm0.016\) \\
\(\vphantom{\bm{\mathcal{RUN}}}\bm{\mathcal{RUN}}\,\bm{2}\)  & \(-0.023\pm0.016\) & \(0.070\pm0.007\) & \(0.166\pm0.007\) & \(-0.081\pm0.016\) \\
\(\vphantom{\bm{\mathcal{RUN}}}\bm{\mathcal{RUN}}\,\bm{3}\)  & \(-0.052\pm0.016\) & \(0.095\pm0.007\) & \(0.179\pm0.007\) & \(-0.092\pm0.016\) \\
\(\vphantom{\bm{\mathcal{RUN}}}\bm{\mathcal{RUN}}\,\bm{4}\)  & \(-0.039\pm0.016\) & \(0.105\pm0.007\) & \(0.193\pm0.007\) & \(-0.080\pm0.016\) \\
\(\vphantom{\bm{\mathcal{RUN}}}\bm{\mathcal{RUN}}\,\bm{5}\)  & \(-0.039\pm0.016\) & \(0.101\pm0.007\) & \(0.167\pm0.007\) & \(-0.131\pm0.016\) \\
\(\vphantom{\bm{\mathcal{RUN}}}\bm{\mathcal{RUN}}\,\bm{6}\)  & \(-0.042\pm0.016\) & \(0.098\pm0.007\) & \(0.172\pm0.007\) & \(-0.104\pm0.016\) \\
\(\vphantom{\bm{\mathcal{RUN}}}\bm{\mathcal{RUN}}\,\bm{7}\)  & \(-0.062\pm0.015\) & \(0.091\pm0.007\) & \(0.194\pm0.006\) & \(-0.097\pm0.015\) \\
\(\vphantom{\bm{\mathcal{RUN}}}\bm{\mathcal{RUN}}\,\bm{8}\)  & \(-0.025\pm0.016\) & \(0.101\pm0.007\) & \(0.188\pm0.007\) & \(-0.092\pm0.016\) \\
\(\vphantom{\bm{\mathcal{RUN}}}\bm{\mathcal{RUN}}\,\bm{9}\)  & \(-0.068\pm0.016\) & \(0.102\pm0.007\) & \(0.183\pm0.007\) & \(-0.095\pm0.016\) \\
\(\vphantom{\bm{\mathcal{RUN}}}\bm{\mathcal{RUN}}\,\bm{10}\) & \(-0.096\pm0.016\) & \(0.093\pm0.007\) & \(0.171\pm0.007\) & \(-0.126\pm0.016\) \\
\(\vphantom{\bm{\mathcal{RUN}}}\bm{\mathcal{RUN}}\,\bm{11}\) & \(-0.024\pm0.016\) & \(0.097\pm0.007\) & \(0.181\pm0.007\) & \(-0.114\pm0.016\) \\
\(\vphantom{\bm{\mathcal{RUN}}}\bm{\mathcal{RUN}}\,\bm{12}\) & \(-0.057\pm0.016\) & \(0.084\pm0.007\) & \(0.178\pm0.007\) & \(-0.107\pm0.016\) \\
\(\vphantom{\bm{\mathcal{RUN}}}\bm{\mathcal{RUN}}\,\bm{13}\) & \(-0.024\pm0.016\) & \(0.109\pm0.007\) & \(0.196\pm0.007\) & \(-0.077\pm0.016\) \\
\(\vphantom{\bm{\mathcal{RUN}}}\bm{\mathcal{RUN}}\,\bm{14}\) & \(-0.072\pm0.016\) & \(0.081\pm0.007\) & \(0.176\pm0.007\) & \(-0.109\pm0.016\) \\
\(\vphantom{\bm{\mathcal{RUN}}}\bm{\mathcal{RUN}}\,\bm{15}\) & \(-0.065\pm0.016\) & \(0.091\pm0.007\) & \(0.187\pm0.007\) & \(-0.095\pm0.016\) \\
\hline
\(\bar{m}\) & \(-0.050\pm0.004\) & \(0.093\pm0.002\) & \(0.181\pm0.002\) & \(-0.098\pm0.004\) \\
\(S\) & \(-1.497\pm0.121\) & \(2.792\pm0.055\) & \(65.221\pm0.617\) & \(-35.295\pm1.450\) \\
\hline
\(I\) & \multicolumn{4}{c}{\(-125.149\pm1.720\)} \\
\(R\) & \multicolumn{4}{c}{\(0.927\pm0.013\)} \\
\end{tabular}
\end{ruledtabular}
\end{table}

\begin{table}[t]
\caption{Experimental Results of the fidelity of six-photon Dicke state at $F=0.638\pm0.005$. \(C_2(A)\), \(C_4(A)\), \(C_6(A)\) are extracted from the six-body expectation of 21 operators according to Eq.(\ref{eq:D6_c}). Under experimental configuration with parametric photon no filtering, 3000 mW pump power, ppKTP crystal temperature stabilized at \(40^\circ\mathrm{C}\) and walk-off compensation (KDP) at \(6^\circ\) offset. Each measurement basis accumulated approximately 4000 events over a 600-second integration time. Error bars are estimated using Poisson Monte Carlo parametric bootstrapping of the coincidence counts.}
\label{tab:F_D6_DI}
\setlength{\tabcolsep}{2.5pt}
\renewcommand{\arraystretch}{0.85}
\begin{ruledtabular}
\footnotesize
\begin{tabular}{crrr}
\multicolumn{1}{c}{Measurement setting \(A\)}
& \multicolumn{1}{c}{\(C_2(A)\)}
& \multicolumn{1}{c}{\(C_4(A)\)}
& \multicolumn{1}{c}{\(C_6(A)\)} \\
\hline
\(X\) & \(0.565\pm0.008\) & \(0.497\pm0.009\) & \(0.583\pm0.013\) \\
\(Y\) & \(0.583\pm0.008\) & \(0.524\pm0.009\) & \(0.662\pm0.012\) \\
\(Z\) & \(-0.147\pm0.002\) & \(0.106\pm0.002\) & \(-0.357\pm0.015\) \\
\(X+Y\) & \(0.585\pm0.008\) & \(0.559\pm0.008\) & \(0.837\pm0.009\) \\
\(X-Y\) & \(0.553\pm0.008\) & \(0.470\pm0.009\) & \(0.451\pm0.014\) \\
\(Y+Z\) & \(0.117\pm0.005\) & \(-0.097\pm0.005\) & \(0.190\pm0.016\) \\
\(Y-Z\) & \(0.321\pm0.006\) & \(0.032\pm0.008\) & \(-0.271\pm0.015\) \\
\(X+Z\) & \(0.331\pm0.006\) & \(0.044\pm0.008\) & \(-0.301\pm0.015\) \\
\(X-Z\) & \(0.079\pm0.004\) & \(-0.092\pm0.004\) & \(0.262\pm0.015\) \\
\(X+Y+Z\) & \(0.364\pm0.006\) & \(0.087\pm0.008\) & \(-0.299\pm0.015\) \\
\(X+Y-Z\) & \(0.326\pm0.006\) & \(0.025\pm0.008\) & \(-0.303\pm0.015\) \\
\(X-Y+Z\) & \(0.479\pm0.007\) & \(0.289\pm0.009\) & \(-0.071\pm0.015\) \\
\(X-Y-Z\) & \(0.152\pm0.005\) & \(-0.087\pm0.005\) & \(0.133\pm0.016\) \\
\(\cos(\pi/6)X+\sin(\pi/6)Z\) & \(0.510\pm0.007\) & \(0.338\pm0.009\) & \(0.055\pm0.015\) \\
\(\cos(2\pi/6)X+\sin(2\pi/6)Z\) & \(0.141\pm0.005\) & \(-0.095\pm0.005\) & \(0.115\pm0.016\) \\
\(\cos(4\pi/6)X+\sin(4\pi/6)Z\) & \(-0.090\pm0.002\) & \(0.023\pm0.003\) & \(0.062\pm0.016\) \\
\(\cos(5\pi/6)X+\sin(5\pi/6)Z\) & \(0.266\pm0.006\) & \(-0.030\pm0.007\) & \(-0.186\pm0.016\) \\
\(\cos(\pi/6)Y+\sin(\pi/6)Z\) & \(0.315\pm0.006\) & \(0.022\pm0.008\) & \(-0.251\pm0.015\) \\
\(\cos(2\pi/6)Y+\sin(2\pi/6)Z\) & \(-0.053\pm0.003\) & \(-0.017\pm0.003\) & \(0.201\pm0.016\) \\
\(\cos(4\pi/6)Y+\sin(4\pi/6)Z\) & \(0.115\pm0.004\) & \(-0.109\pm0.004\) & \(0.189\pm0.016\) \\
\(\cos(5\pi/6)Y+\sin(5\pi/6)Z\) & \(0.494\pm0.007\) & \(0.304\pm0.009\) & \(-0.021\pm0.016\) \\
\hline
\multicolumn{1}{c}{Fidelity of \(\ket{D_6^{\,3}}\)}
& \multicolumn{3}{c}{\(0.638\pm0.005\)} \\
\end{tabular}
\end{ruledtabular}
\end{table}

\begin{table}[t]
\caption{Experimental Results of inequalities $\mathcal{I}_1$ sets in $\ket{D_6^{\,3}}$ at $F=0.638\pm0.005$. \(m_{00}\), \(m_{11}\), \(m_{0011}\), and \(m_{1111}\) are extracted from the six-body expectation of 15 operators according to Eq.\eqref{eq:run_marginal_correlators}---Eq.\eqref{eq:S_m_relation}. Under experimental configuration with parametric photon no filtering, 3000 mW pump power, ppKTP crystal temperature stabilized at \(40^\circ\mathrm{C}\) and walk-off compensation (KDP) at \(6^\circ\) offset. Each measurement basis accumulated approximately 4000 events over a 600-second integration time. Error bars are estimated using Poisson Monte Carlo parametric bootstrapping of the coincidence counts.}
\label{tab:PI_DI_0}
\setlength{\tabcolsep}{2.5pt}
\renewcommand{\arraystretch}{0.85}
\begin{ruledtabular}
\small
\begin{tabular}{crrrr}
\multicolumn{1}{c}{\(\vphantom{\bm{\mathcal{RUN}}}\)Run}
& \multicolumn{1}{c}{\(m_{00}\)}
& \multicolumn{1}{c}{\(m_{11}\)}
& \multicolumn{1}{c}{\(m_{0011}\)}
& \multicolumn{1}{c}{\(m_{1111}\)} \\
\hline
\(\vphantom{\bm{\mathcal{RUN}}}\bm{\mathcal{RUN}}\,\bm{1}\)  & \(0.083\pm0.016\) & \(-0.054\pm0.006\) & \(0.169\pm0.005\) & \(-0.032\pm0.016\) \\
\(\vphantom{\bm{\mathcal{RUN}}}\bm{\mathcal{RUN}}\,\bm{2}\)  & \(0.074\pm0.016\) & \(-0.058\pm0.006\) & \(0.162\pm0.005\) & \(-0.030\pm0.016\) \\
\(\vphantom{\bm{\mathcal{RUN}}}\bm{\mathcal{RUN}}\,\bm{3}\)  & \(0.031\pm0.016\) & \(-0.056\pm0.006\) & \(0.179\pm0.005\) & \(0.011\pm0.016\) \\
\(\vphantom{\bm{\mathcal{RUN}}}\bm{\mathcal{RUN}}\,\bm{4}\)  & \(0.068\pm0.015\) & \(-0.048\pm0.006\) & \(0.172\pm0.005\) & \(-0.017\pm0.016\) \\
\(\vphantom{\bm{\mathcal{RUN}}}\bm{\mathcal{RUN}}\,\bm{5}\)  & \(0.028\pm0.016\) & \(-0.063\pm0.005\) & \(0.172\pm0.005\) & \(-0.014\pm0.016\) \\
\(\vphantom{\bm{\mathcal{RUN}}}\bm{\mathcal{RUN}}\,\bm{6}\)  & \(0.035\pm0.016\) & \(-0.060\pm0.006\) & \(0.168\pm0.005\) & \(-0.027\pm0.016\) \\
\(\vphantom{\bm{\mathcal{RUN}}}\bm{\mathcal{RUN}}\,\bm{7}\)  & \(0.054\pm0.016\) & \(-0.052\pm0.006\) & \(0.175\pm0.005\) & \(0.014\pm0.016\) \\
\(\vphantom{\bm{\mathcal{RUN}}}\bm{\mathcal{RUN}}\,\bm{8}\)  & \(0.067\pm0.016\) & \(-0.057\pm0.006\) & \(0.171\pm0.005\) & \(-0.003\pm0.016\) \\
\(\vphantom{\bm{\mathcal{RUN}}}\bm{\mathcal{RUN}}\,\bm{9}\)  & \(0.065\pm0.016\) & \(-0.055\pm0.006\) & \(0.171\pm0.005\) & \(-0.007\pm0.016\) \\
\(\vphantom{\bm{\mathcal{RUN}}}\bm{\mathcal{RUN}}\,\bm{10}\) & \(0.061\pm0.016\) & \(-0.048\pm0.006\) & \(0.168\pm0.005\) & \(-0.034\pm0.016\) \\
\(\vphantom{\bm{\mathcal{RUN}}}\bm{\mathcal{RUN}}\,\bm{11}\) & \(0.078\pm0.016\) & \(-0.057\pm0.006\) & \(0.160\pm0.005\) & \(-0.050\pm0.016\) \\
\(\vphantom{\bm{\mathcal{RUN}}}\bm{\mathcal{RUN}}\,\bm{12}\) & \(0.017\pm0.016\) & \(-0.064\pm0.005\) & \(0.172\pm0.005\) & \(-0.001\pm0.016\) \\
\(\vphantom{\bm{\mathcal{RUN}}}\bm{\mathcal{RUN}}\,\bm{13}\) & \(0.100\pm0.015\) & \(-0.050\pm0.006\) & \(0.172\pm0.005\) & \(-0.017\pm0.015\) \\
\(\vphantom{\bm{\mathcal{RUN}}}\bm{\mathcal{RUN}}\,\bm{14}\) & \(0.056\pm0.015\) & \(-0.059\pm0.005\) & \(0.171\pm0.005\) & \(-0.024\pm0.015\) \\
\(\vphantom{\bm{\mathcal{RUN}}}\bm{\mathcal{RUN}}\,\bm{15}\) & \(0.076\pm0.016\) & \(-0.064\pm0.005\) & \(0.172\pm0.005\) & \(0.010\pm0.016\) \\
\hline
\(\bar{m}\) & \(0.060\pm0.004\) & \(-0.056\pm0.001\) & \(0.170\pm0.001\) & \(-0.015\pm0.004\) \\
\(S\) & \(1.787\pm0.121\) & \(-1.686\pm0.043\) & \(61.262\pm0.456\) & \(-5.281\pm1.458\) \\
\hline
\(I\) & \multicolumn{4}{c}{\(-16.095\pm0.177\)} \\
\(R\) & \multicolumn{4}{c}{\(0.894\pm0.010\)} \\
\end{tabular}
\end{ruledtabular}
\end{table}

\begin{table}[t]
\caption{Experimental Results of inequalities $\mathcal{I}_2$ sets in $\ket{D_6^{\,3}}$ at $F=0.638\pm0.005$. \(m_{00}\), \(m_{11}\), \(m_{0011}\), and \(m_{1111}\) are extracted from the six-body expectation of 15 operators according to Eq.\eqref{eq:run_marginal_correlators}---Eq.\eqref{eq:S_m_relation}. Under experimental configuration with parametric photon no filtering, 3000 mW pump power, ppKTP crystal temperature stabilized at \(40^\circ\mathrm{C}\) and walk-off compensation (KDP) at \(6^\circ\) offset. Each measurement basis accumulated approximately 4000 events over a 600-second integration time. Error bars are estimated using Poisson Monte Carlo parametric bootstrapping of the coincidence counts.}
\label{tab:PI_DI_1}
\setlength{\tabcolsep}{2.5pt}
\renewcommand{\arraystretch}{0.85}
\begin{ruledtabular}
\small
\begin{tabular}{crrrr}
\multicolumn{1}{c}{\(\vphantom{\bm{\mathcal{RUN}}}\)Run}
& \multicolumn{1}{c}{\(m_{00}\)}
& \multicolumn{1}{c}{\(m_{11}\)}
& \multicolumn{1}{c}{\(m_{0011}\)}
& \multicolumn{1}{c}{\(m_{1111}\)} \\
\hline
\(\vphantom{\bm{\mathcal{RUN}}}\bm{\mathcal{RUN}}\,\bm{1}\)  & \(-0.024\pm0.016\) & \(0.070\pm0.007\) & \(0.182\pm0.006\) & \(-0.074\pm0.016\) \\
\(\vphantom{\bm{\mathcal{RUN}}}\bm{\mathcal{RUN}}\,\bm{2}\)  & \(-0.010\pm0.016\) & \(0.049\pm0.007\) & \(0.175\pm0.006\) & \(-0.085\pm0.016\) \\
\(\vphantom{\bm{\mathcal{RUN}}}\bm{\mathcal{RUN}}\,\bm{3}\)  & \(-0.023\pm0.016\) & \(0.074\pm0.007\) & \(0.171\pm0.006\) & \(-0.089\pm0.016\) \\
\(\vphantom{\bm{\mathcal{RUN}}}\bm{\mathcal{RUN}}\,\bm{4}\)  & \(0.006\pm0.015\) & \(0.074\pm0.007\) & \(0.188\pm0.006\) & \(-0.098\pm0.015\) \\
\(\vphantom{\bm{\mathcal{RUN}}}\bm{\mathcal{RUN}}\,\bm{5}\)  & \(-0.018\pm0.016\) & \(0.077\pm0.007\) & \(0.182\pm0.007\) & \(-0.073\pm0.016\) \\
\(\vphantom{\bm{\mathcal{RUN}}}\bm{\mathcal{RUN}}\,\bm{6}\)  & \(-0.012\pm0.015\) & \(0.076\pm0.007\) & \(0.176\pm0.006\) & \(-0.106\pm0.015\) \\
\(\vphantom{\bm{\mathcal{RUN}}}\bm{\mathcal{RUN}}\,\bm{7}\)  & \(-0.044\pm0.015\) & \(0.060\pm0.007\) & \(0.184\pm0.006\) & \(-0.111\pm0.015\) \\
\(\vphantom{\bm{\mathcal{RUN}}}\bm{\mathcal{RUN}}\,\bm{8}\)  & \(0.021\pm0.016\) & \(0.075\pm0.007\) & \(0.197\pm0.006\) & \(-0.093\pm0.016\) \\
\(\vphantom{\bm{\mathcal{RUN}}}\bm{\mathcal{RUN}}\,\bm{9}\)  & \(0.003\pm0.016\) & \(0.082\pm0.007\) & \(0.192\pm0.006\) & \(-0.113\pm0.016\) \\
\(\vphantom{\bm{\mathcal{RUN}}}\bm{\mathcal{RUN}}\,\bm{10}\) & \(-0.031\pm0.016\) & \(0.040\pm0.007\) & \(0.176\pm0.006\) & \(-0.078\pm0.016\) \\
\(\vphantom{\bm{\mathcal{RUN}}}\bm{\mathcal{RUN}}\,\bm{11}\) & \(-0.015\pm0.015\) & \(0.061\pm0.007\) & \(0.183\pm0.006\) & \(-0.102\pm0.015\) \\
\(\vphantom{\bm{\mathcal{RUN}}}\bm{\mathcal{RUN}}\,\bm{12}\) & \(-0.024\pm0.015\) & \(0.066\pm0.007\) & \(0.174\pm0.006\) & \(-0.118\pm0.015\) \\
\(\vphantom{\bm{\mathcal{RUN}}}\bm{\mathcal{RUN}}\,\bm{13}\) & \(-0.017\pm0.015\) & \(0.075\pm0.007\) & \(0.184\pm0.006\) & \(-0.094\pm0.015\) \\
\(\vphantom{\bm{\mathcal{RUN}}}\bm{\mathcal{RUN}}\,\bm{14}\) & \(-0.045\pm0.016\) & \(0.056\pm0.007\) & \(0.175\pm0.006\) & \(-0.096\pm0.016\) \\
\(\vphantom{\bm{\mathcal{RUN}}}\bm{\mathcal{RUN}}\,\bm{15}\) & \(-0.013\pm0.015\) & \(0.061\pm0.007\) & \(0.174\pm0.006\) & \(-0.085\pm0.015\) \\
\hline
\(\bar{m}\) & \(-0.016\pm0.004\) & \(0.066\pm0.002\) & \(0.181\pm0.002\) & \(-0.094\pm0.004\) \\
\(S\) & \(-0.491\pm0.120\) & \(1.990\pm0.053\) & \(65.108\pm0.582\) & \(-33.944\pm1.438\) \\
\hline
\(I\) & \multicolumn{4}{c}{\(-104.935\pm1.731\)} \\
\(R\) & \multicolumn{4}{c}{\(0.777\pm0.013\)} \\
\end{tabular}
\end{ruledtabular}
\end{table}

\begin{table}[t]
\caption{
Experimental Results of the fidelity of eight-photon Dicke state at
$F=0.916\pm0.011$. \(C_2(A)\), \(C_4(A)\), \(C_6(A)\) and \(C_8(A)\) are extracted from the eight-body expectation of 29 operators according to Eq\eqref{eq:D8_c}. Under experimental configuration with parametric photon no filtering, 3080 mW pump power, ppKTP crystal temperature stabilized at 28$^{\circ}$C. Each measurement basis accumulated approximately 700 events over a 21600-second integration time. Error bars are estimated using Poisson Monte Carlo parametric bootstrapping of the coincidence counts.
}
\label{tab:F_D8}
\setlength{\tabcolsep}{4.5pt}
\renewcommand{\arraystretch}{0.92}
\begin{ruledtabular}
\begin{tabular}{crrrr}
\multicolumn{1}{c}{Measurement setting \(A\)}
& \multicolumn{1}{c}{\(C_2(A)\)}
& \multicolumn{1}{c}{\(C_4(A)\)}
& \multicolumn{1}{c}{\(C_6(A)\)}
& \multicolumn{1}{c}{\(C_8(A)\)} \\
\hline
\(X\) & \(0.605\pm0.017\) & \(0.544\pm0.019\) & \(0.580\pm0.018\) & \(0.866\pm0.018\) \\
\(Y\) & \(0.613\pm0.017\) & \(0.545\pm0.019\) & \(0.571\pm0.019\) & \(0.832\pm0.020\) \\
\(Z\) & \(-0.136\pm0.001\) & \(0.078\pm0.001\) & \(-0.124\pm0.002\) & \(0.820\pm0.022\) \\
\(Y-Z\) & \(0.209\pm0.011\) & \(-0.012\pm0.008\) & \(-0.052\pm0.013\) & \(0.252\pm0.038\) \\
\(Y+Z\) & \(0.229\pm0.011\) & \(-0.004\pm0.009\) & \(-0.043\pm0.014\) & \(0.276\pm0.036\) \\
\(X+Y\) & \(0.615\pm0.017\) & \(0.556\pm0.019\) & \(0.590\pm0.018\) & \(0.854\pm0.019\) \\
\(X-Y\) & \(0.605\pm0.017\) & \(0.555\pm0.019\) & \(0.581\pm0.018\) & \(0.862\pm0.018\) \\
\(X-Z\) & \(0.210\pm0.010\) & \(-0.020\pm0.008\) & \(-0.042\pm0.013\) & \(0.289\pm0.037\) \\
\(X+Z\) & \(0.237\pm0.011\) & \(0.003\pm0.010\) & \(-0.024\pm0.014\) & \(0.305\pm0.036\) \\
\(Y+2Z\) & \(0.001\pm0.006\) & \(-0.028\pm0.003\) & \(0.054\pm0.005\) & \(-0.372\pm0.037\) \\
\(Y-2Z\) & \(-0.003\pm0.005\) & \(-0.028\pm0.003\) & \(0.049\pm0.005\) & \(-0.309\pm0.037\) \\
\(X+2Z\) & \(0.004\pm0.006\) & \(-0.026\pm0.004\) & \(0.058\pm0.005\) & \(-0.332\pm0.036\) \\
\(X-2Z\) & \(-0.010\pm0.005\) & \(-0.024\pm0.004\) & \(0.059\pm0.005\) & \(-0.338\pm0.036\) \\
\(2Y-Z\) & \(0.469\pm0.016\) & \(0.268\pm0.018\) & \(0.090\pm0.024\) & \(-0.239\pm0.038\) \\
\(2Y+Z\) & \(0.442\pm0.014\) & \(0.222\pm0.015\) & \(0.026\pm0.021\) & \(-0.306\pm0.034\) \\
\(X+2Y\) & \(0.603\pm0.017\) & \(0.534\pm0.019\) & \(0.571\pm0.018\) & \(0.829\pm0.020\) \\
\(X-2Y\) & \(0.589\pm0.017\) & \(0.529\pm0.019\) & \(0.569\pm0.018\) & \(0.832\pm0.020\) \\
\(2X+Y\) & \(0.584\pm0.017\) & \(0.511\pm0.019\) & \(0.551\pm0.018\) & \(0.837\pm0.020\) \\
\(2X-Y\) & \(0.577\pm0.018\) & \(0.512\pm0.020\) & \(0.551\pm0.019\) & \(0.856\pm0.019\) \\
\(2X+Z\) & \(0.458\pm0.015\) & \(0.251\pm0.016\) & \(0.065\pm0.022\) & \(-0.286\pm0.035\) \\
\(2X-Z\) & \(0.415\pm0.014\) & \(0.192\pm0.015\) & \(-0.007\pm0.021\) & \(-0.308\pm0.035\) \\
\(X+Y+Z\) & \(0.335\pm0.013\) & \(0.092\pm0.013\) & \(-0.071\pm0.018\) & \(-0.021\pm0.037\) \\
\(X+Y-Z\) & \(0.337\pm0.013\) & \(0.093\pm0.013\) & \(-0.073\pm0.018\) & \(0.011\pm0.037\) \\
\(X-Y+Z\) & \(0.361\pm0.013\) & \(0.119\pm0.014\) & \(-0.047\pm0.019\) & \(-0.034\pm0.038\) \\
\(X-Y-Z\) & \(0.322\pm0.013\) & \(0.075\pm0.012\) & \(-0.082\pm0.017\) & \(0.016\pm0.038\) \\
\(X+Y+2Z\) & \(0.113\pm0.009\) & \(-0.035\pm0.007\) & \(0.024\pm0.010\) & \(0.049\pm0.040\) \\
\(X+Y-2Z\) & \(0.101\pm0.008\) & \(-0.044\pm0.006\) & \(0.043\pm0.009\) & \(0.017\pm0.039\) \\
\(X-Y+2Z\) & \(0.110\pm0.008\) & \(-0.052\pm0.005\) & \(0.017\pm0.009\) & \(0.058\pm0.039\) \\
\(X-Y-2Z\) & \(0.070\pm0.007\) & \(-0.052\pm0.004\) & \(0.040\pm0.007\) & \(-0.057\pm0.039\) \\
\hline
\multicolumn{1}{c}{Fidelity of \(\ket{D_8^{\,4}}\)}
& \multicolumn{4}{c}{\(0.916\pm0.011\)} \\
\end{tabular}
\end{ruledtabular}
\end{table}

\begin{table}[t]
\caption{
Experimental results of permutation-invariant Bell inequality $\mathcal{I}_3$ sets in $\ket{D_8^{\,4}}$ at $F=0.916\pm0.011$. \(m_{00}\), \(m_{11}\), \(m_{0011}\), and \(m_{1111}\) are extracted from the eight-body expectation values of measured operators according to Eq.\eqref{eq:run_marginal_correlators}---Eq.\eqref{eq:S_m_relation}. Under experimental configuration with parametric photon no filtering, 3080 mW pump power, ppKTP crystal temperature stabilized at 28$^{\circ}$C. Each measurement basis accumulated approximately
680 eight-fold coincidence events over a 21600-second integration time.
Error bars are estimated using Poisson Monte Carlo parametric bootstrapping of the coincidence counts.
}
\label{tab:PI_D8}
\setlength{\tabcolsep}{5pt}
\renewcommand{\arraystretch}{1.1}
\begin{ruledtabular}
\begin{tabular}{crrrr}
\multicolumn{1}{c}{\(\vphantom{\bm{\mathcal{RUN}}}\)Run}
& \multicolumn{1}{c}{\(m_{00}\)}
& \multicolumn{1}{c}{\(m_{11}\)}
& \multicolumn{1}{c}{\(m_{0011}\)}
& \multicolumn{1}{c}{\(m_{1111}\)} \\
\hline
\(\vphantom{\bm{\mathcal{RUN}}}\bm{\mathcal{RUN}}\,1\)  & \(-0.039\pm0.037\) & \(-0.014\pm0.008\) & \(0.089\pm0.007\) & \(-0.033\pm0.008\) \\
\(\vphantom{\bm{\mathcal{RUN}}}\bm{\mathcal{RUN}}\,2\)  & \(-0.053\pm0.037\) & \(-0.019\pm0.008\) & \(0.088\pm0.007\) & \(-0.031\pm0.008\) \\
\(\vphantom{\bm{\mathcal{RUN}}}\bm{\mathcal{RUN}}\,3\)  & \(-0.085\pm0.040\) & \(-0.015\pm0.009\) & \(0.098\pm0.009\) & \(-0.009\pm0.009\) \\
\(\vphantom{\bm{\mathcal{RUN}}}\bm{\mathcal{RUN}}\,4\)  & \(-0.040\pm0.038\) & \(0.003\pm0.009\) & \(0.095\pm0.008\) & \(-0.026\pm0.009\) \\
\(\vphantom{\bm{\mathcal{RUN}}}\bm{\mathcal{RUN}}\,5\)  & \(-0.024\pm0.038\) & \(-0.013\pm0.008\) & \(0.091\pm0.007\) & \(-0.031\pm0.008\) \\
\(\vphantom{\bm{\mathcal{RUN}}}\bm{\mathcal{RUN}}\,6\)  & \(-0.063\pm0.038\) & \(0.006\pm0.009\) & \(0.111\pm0.008\) & \(-0.023\pm0.009\) \\
\(\vphantom{\bm{\mathcal{RUN}}}\bm{\mathcal{RUN}}\,7\)  & \(-0.033\pm0.038\) & \(-0.014\pm0.009\) & \(0.102\pm0.008\) & \(-0.016\pm0.009\) \\
\(\vphantom{\bm{\mathcal{RUN}}}\bm{\mathcal{RUN}}\,8\)  & \(-0.081\pm0.039\) & \(-0.026\pm0.008\) & \(0.097\pm0.007\) & \(-0.027\pm0.008\) \\
\(\vphantom{\bm{\mathcal{RUN}}}\bm{\mathcal{RUN}}\,9\)  & \(-0.009\pm0.039\) & \(-0.011\pm0.008\) & \(0.100\pm0.007\) & \(-0.036\pm0.008\) \\
\(\vphantom{\bm{\mathcal{RUN}}}\bm{\mathcal{RUN}}\,10\) & \(-0.005\pm0.037\) & \(-0.003\pm0.009\) & \(0.101\pm0.008\) & \(-0.023\pm0.009\) \\
\(\vphantom{\bm{\mathcal{RUN}}}\bm{\mathcal{RUN}}\,11\) & \(-0.011\pm0.038\) & \(-0.002\pm0.009\) & \(0.108\pm0.008\) & \(-0.021\pm0.009\) \\
\(\vphantom{\bm{\mathcal{RUN}}}\bm{\mathcal{RUN}}\,12\) & \(0.026\pm0.039\) & \(-0.003\pm0.009\) & \(0.102\pm0.008\) & \(-0.027\pm0.009\) \\
\(\vphantom{\bm{\mathcal{RUN}}}\bm{\mathcal{RUN}}\,13\) & \(0.020\pm0.038\) & \(-0.017\pm0.008\) & \(0.081\pm0.007\) & \(-0.034\pm0.008\) \\
\(\vphantom{\bm{\mathcal{RUN}}}\bm{\mathcal{RUN}}\,14\) & \(-0.034\pm0.039\) & \(-0.010\pm0.009\) & \(0.105\pm0.008\) & \(-0.022\pm0.009\) \\
\(\vphantom{\bm{\mathcal{RUN}}}\bm{\mathcal{RUN}}\,15\) & \(0.001\pm0.038\) & \(-0.033\pm0.008\) & \(0.092\pm0.007\) & \(-0.014\pm0.008\) \\
\(\vphantom{\bm{\mathcal{RUN}}}\bm{\mathcal{RUN}}\,16\) & \(-0.033\pm0.037\) & \(-0.012\pm0.008\) & \(0.103\pm0.007\) & \(-0.022\pm0.008\) \\
\(\vphantom{\bm{\mathcal{RUN}}}\bm{\mathcal{RUN}}\,17\) & \(0.002\pm0.039\) & \(-0.012\pm0.009\) & \(0.092\pm0.008\) & \(-0.019\pm0.009\) \\
\(\vphantom{\bm{\mathcal{RUN}}}\bm{\mathcal{RUN}}\,18\) & \(-0.048\pm0.039\) & \(-0.028\pm0.009\) & \(0.100\pm0.008\) & \(-0.002\pm0.009\) \\
\(\vphantom{\bm{\mathcal{RUN}}}\bm{\mathcal{RUN}}\,19\) & \(-0.056\pm0.039\) & \(-0.032\pm0.008\) & \(0.089\pm0.008\) & \(-0.014\pm0.008\) \\
\(\vphantom{\bm{\mathcal{RUN}}}\bm{\mathcal{RUN}}\,20\) & \(-0.045\pm0.039\) & \(-0.015\pm0.009\) & \(0.096\pm0.008\) & \(-0.023\pm0.009\) \\
\(\vphantom{\bm{\mathcal{RUN}}}\bm{\mathcal{RUN}}\,21\) & \(0.040\pm0.039\) & \(-0.005\pm0.009\) & \(0.107\pm0.008\) & \(-0.023\pm0.009\) \\
\(\vphantom{\bm{\mathcal{RUN}}}\bm{\mathcal{RUN}}\,22\) & \(-0.014\pm0.039\) & \(-0.008\pm0.009\) & \(0.100\pm0.008\) & \(-0.024\pm0.009\) \\
\(\vphantom{\bm{\mathcal{RUN}}}\bm{\mathcal{RUN}}\,23\) & \(-0.024\pm0.039\) & \(-0.013\pm0.009\) & \(0.091\pm0.008\) & \(-0.028\pm0.009\) \\
\(\vphantom{\bm{\mathcal{RUN}}}\bm{\mathcal{RUN}}\,24\) & \(-0.054\pm0.038\) & \(-0.013\pm0.009\) & \(0.095\pm0.008\) & \(-0.009\pm0.009\) \\
\(\vphantom{\bm{\mathcal{RUN}}}\bm{\mathcal{RUN}}\,25\) & \(-0.044\pm0.039\) & \(-0.014\pm0.009\) & \(0.087\pm0.008\) & \(-0.026\pm0.009\) \\
\(\vphantom{\bm{\mathcal{RUN}}}\bm{\mathcal{RUN}}\,26\) & \(-0.057\pm0.039\) & \(-0.019\pm0.009\) & \(0.104\pm0.008\) & \(-0.016\pm0.009\) \\
\(\vphantom{\bm{\mathcal{RUN}}}\bm{\mathcal{RUN}}\,27\) & \(-0.022\pm0.039\) & \(-0.019\pm0.008\) & \(0.095\pm0.007\) & \(-0.029\pm0.008\) \\
\(\vphantom{\bm{\mathcal{RUN}}}\bm{\mathcal{RUN}}\,28\) & \(0.031\pm0.038\) & \(-0.004\pm0.009\) & \(0.109\pm0.008\) & \(-0.024\pm0.009\) \\
\hline
\(\bar{m}\) & \(-0.027\pm0.007\) & \(-0.013\pm0.002\) & \(0.097\pm0.002\) & \(-0.023\pm0.002\) \\
\(S\) & \(-1.509\pm0.405\) & \(-0.727\pm0.091\) & \(163.643\pm2.466\) & \(-37.832\pm2.722\) \\
\hline
\(I\) & \multicolumn{4}{c}{\(-55.905\pm0.889\)} \\
\(R\) & \multicolumn{4}{c}{\(1.165\pm0.019\)} \\
\end{tabular}
\end{ruledtabular}
\end{table}

\end{document}